\newcommand{\onlyFull}[1]{#1}
\newcommand{\onlyConference}[1]{}
\newif\ifFull
\Fulltrue
\documentclass[letterpaper, UKenglish, cleveref, autoref, thm-restate, numberwithinsect]{lipics-v2021}

\newcommand{\whenFull}[2]{\onlyFull{#1}\onlyConference{#2}}

\title{The complexity of downward closures of indexed languages}

\author{Richard Mandel}{Max Planck Institute for Software Systems (MPI-SWS), Germany}{rmandel@mpi-sws.org}{https://orcid.org/0009-0001-5062-0860}{}
\author{Corto Mascle}{Max Planck Institute for Software Systems (MPI-SWS), Germany}{cmascle@mpi-sws.org}{https://orcid.org/0009-0007-7976-7480}{}
\author{Georg Zetzsche}{Max Planck Institute for Software Systems (MPI-SWS), Germany}{georg@mpi-sws.org}{https://orcid.org/0000-0002-6421-4388}{}

\authorrunning{R. Mandel, C. Mascle, G. Zetzsche} %



\category{} %

\whenFull{}{\relatedversiondetails[cite=MandelMZ26arxiv]{Full version}{https://arxiv.org/abs/2601.19466}} %

\funding{\flag[3cm]{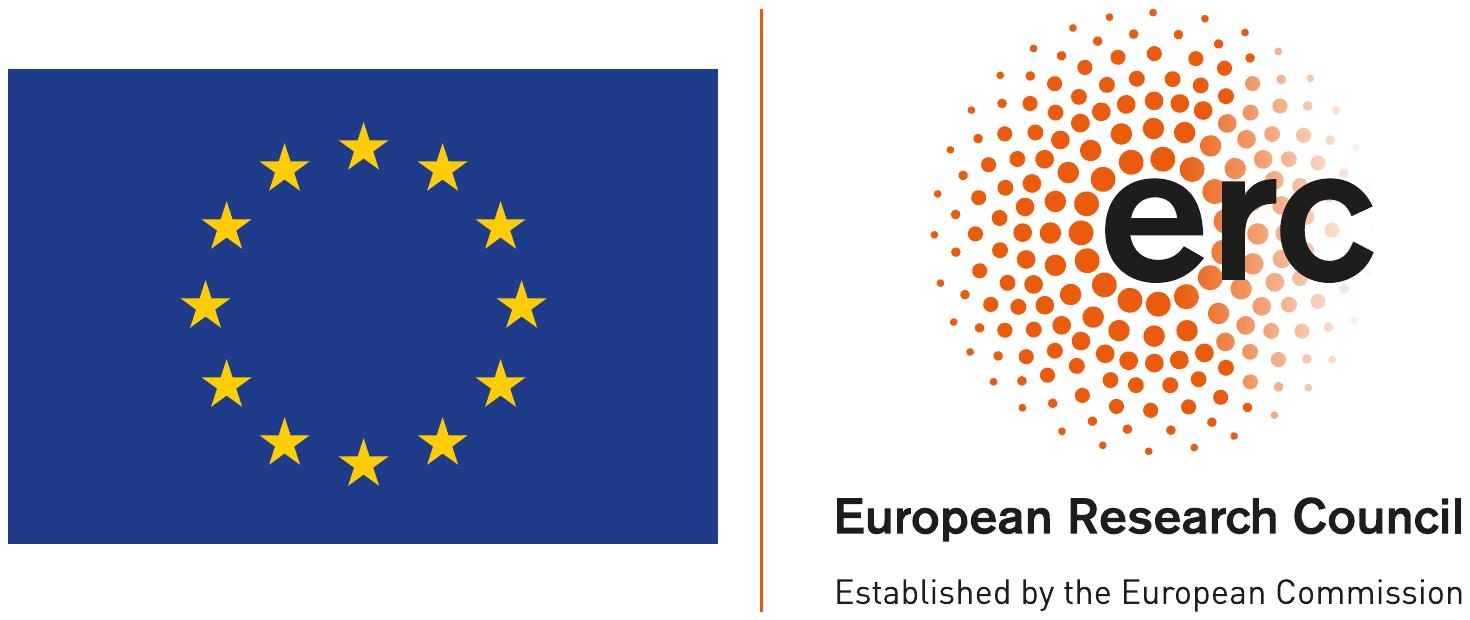}Funded by the European Union (ERC, FINABIS, 101077902). Views and opinions expressed are however those of the author(s) only and do not necessarily reflect those of the European Union or the European Research Council Executive Agency. Neither the European Union nor the granting authority can be held responsible for them.}

\acknowledgements{We are grateful to Alexander Kartzow for answering questions
about his work~\cite{DBLP:conf/csl/Kartzow11}, and to Hermann Gruber for making
us aware of literature on pumping constants.}

\Copyright{Richard Mandel, Corto Mascle, and Georg Zetzsche}

\nolinenumbers %

\onlyConference{%
\EventEditors{Claudia Faggian and Joost-Pieter Katoen}
\EventNoEds{2}
\EventLongTitle{41st Annual Symposium on Logic in Computer Science (LICS 2026)}
\EventShortTitle{LICS 2026}
\EventAcronym{LICS}
\EventYear{2026}
\EventDate{July 20--23, 2026}
\EventLocation{Lisbon, Portugal}
\EventLogo{}
\SeriesVolume{380}
\ArticleNo{64}
}

\mathchardef\mhyphen="2D

\usepackage{amsfonts}
\usepackage{amsmath}
\usepackage{url}
\usepackage{graphicx}
\usepackage{tikz}
\usepackage[nobreak=true]{mdframed}
\usepackage{hyperref}
\usepackage{multicol}
\usepackage{newfile}
\usepackage{xparse}
\usepackage[compress]{cite}
\usepackage[notion, quotation, electronic]{knowledge}
\usepackage{fontawesome5}

\input{indexed.kl}

\usepackage{algorithm}
\usepackage[noend]{algpseudocode}

\definecolor{defcolor}{HTML}{ccfaff}
\definecolor{lemcolor}{HTML}{eeffee}
\definecolor{thmcolor}{HTML}{ffeeee}
\definecolor{propcolor}{HTML}{ddffff}

\usepackage{enumitem}
\usepackage{todonotes}
\newcommand{\overbar}[1]{\mkern 1.5mu\overline{\mkern-1.5mu#1\mkern-1.5mu}\mkern 1.5mu}

\newcommand{\coNTIME}{\mathsf{coNTIME}}

\newcommand*{\coNEXP}[1][OPT]{\ifthenelse{\equal{#1}{OPT}}{\mathsf{coNEXP}}{\mathsf{co}\mhyphen{}#1\mhyphen\mathsf{NEXP}}}
\newcommand*{\EXPTIME}[1][OPT]{\ifthenelse{\equal{#1}{OPT}}{\mathsf{EXPTIME}}{#1\mhyphen\mathsf{EXPTIME}}}
\newcommand{\PTIME}{\mathsf{P}}
\newcommand{\coNP}{\mathsf{coNP}}
\newcommand{\ltr}[1]{\mathtt{#1}}

\newcommand{\set}[1]{\{#1\}}
\newcommand{\powerset}[1]{2^{#1}}

\renewcommand{\phi}{\varphi}
\renewcommand{\epsilon}{\varepsilon}

\newcommand{\Aa}{\mathcal{A}}
\newcommand{\Bcal}{\mathcal{B}}
\newcommand{\Cc}{\mathcal{C}}

\newcommand{\Gg}{\mathcal{G}}

\newcommand{\Tt}{\mathcal{T}}

\newcommand{\NN}{\mathbb{N}}

\newcommand{\Nb}{\overbar{N}}
\newcommand{\Ib}{\overline{I}}
\newcommand{\Pb}{\overbar{P}}
\newcommand{\Sb}{\overbar{S}}

\newcommand{\ub}{\overbar{u}}
\newcommand{\wb}{\overbar{w}}
\newcommand{\xb}{\overbar{x}}

\newcommand{\zb}{\overbar{z}}

\newcommand{\Ggb}{\overbar{\Gg}}

\newcommand{\myparagraph}[1]{\subparagraph*{#1}}

\knowledgenewrobustcmd{\PumpConst}{\mathfrak{P}}
\knowledgenewrobustcmd{\subword}{\mathbin{\cmdkl{\preceq}}}
\knowledgenewrobustcmd{\dcl}[1]{\cmdkl{#1\mathord{\downarrow}}}

\makeatletter
\newcommand{\xRightarrow}[2][]{\ext@arrow 0359\Rightarrowfill@{#1}{#2}}
\makeatother

\newcommand{\lang}[1]{\mathsf{L}(#1)}
\knowledgenewrobustcmd{\langIG}[1]{\cmdkl{\mathsf{L}(#1)}}
\knowledgenewrobustcmd{\langSF}[1]{\cmdkl{\mathsf{L}_{\mathsf{SF}}(#1)}}
\knowledgenewrobustcmd{\langX}[2]{\cmdkl{\mathsf{L}_{#2}(#1)}}
\knowledgenewrobustcmd{\langPS}[1]{\cmdkl{\mathsf{L}_{\mathsf{pump,skip}}(#1)}}

\knowledgenewrobustcmd{\sentforms}{\cmdkl{\mathsf{SF}}}

\knowledgenewrobustcmd{\rulepump}{\cmdkl{\to_{\text{\scriptsize{\textsf{pump}}}}}}
\knowledgenewrobustcmd{\ruleskip}{\cmdkl{\to_{\text{\scriptsize{\textsf{skip}}}}}}
\knowledgenewrobustcmd{\ruleskipstar}{\xrightarrow{*}_{\scriptsize{\mathsf{skip}}}}
\knowledgenewrobustcmd{\derivesindexedstar}[1][]{\cmdkl{\mathbin{\xRightarrow{*}_{\text{\scriptsize{$#1$}}}}}}

\knowledgenewrobustcmd{\derivesindexed}[1][]{\cmdkl{\mathbin{\Rightarrow_{\text{\scriptsize{$#1$}}}}}}

\knowledgenewrobustcmd{\derivesindexedps}[1][EMPTY]{%
	\ifthenelse{\equal{#1}{EMPTY}}{%
		\derivesindexed[\mathsf{pump,skip}]%
	}{%
		\derivesindexed[\mathsf{pump,skip},#1]}%
}
\knowledgenewrobustcmd{\derivesindexedstarps}[1][EMPTY]{%
	\ifthenelse{\equal{#1}{EMPTY}}{%
		\derivesindexedstar[\mathsf{pump,skip}]%
	}{%
		\derivesindexedstar[\mathsf{pump,skip},#1]}%
}

\knowledgenewrobustcmd{\Useful}{\cmdkl{\mathsf{U}}}

\knowledgenewrobustcmd{\relTrans}[1]{\cmdkl{\textsf{R}(#1)}}

\newcommand{\grammarcommand}[1]{\color{gray} \emph{#1}}

\knowledgenewrobustcmd{\projection}{\cmdkl{\pi}}

\knowledgenewrobustcmd{\Reach}[1]{\ \cmdkl{\mathsf{R}_{#1}}\ }

\knowledgenewrobustcmd{\alphaG}{\cmdkl{{\alpha}}}
\knowledgenewrobustcmd{\betaG}{\cmdkl{{\beta}}}

\newcommand{\matrixmonoid}{\mathbb{B}^{N \times N}}
\knowledgenewrobustcmd{\prodmonoid}{\cmdkl{\mathbb{M}}}
\newcommand{\neutralM}{\mathbf{1}_{\mathbf{M}}}
\newcommand{\neutral}{\mathbf{1}_{\prodmonoid}}
\newcommand{\zeroPM}{\mathbf{0}_{\prodmonoid}}
\knowledgenewrobustcmd{\morphism}{\cmdkl{\phi}}

\knowledgenewrobustcmd{\Jgreen}{\cmdkl{\mathcal{J}}}
\knowledgenewrobustcmd{\Rgreen}{\cmdkl{\mathcal{R}}}
\knowledgenewrobustcmd{\Lgreen}{\cmdkl{\mathcal{L}}}
\knowledgenewrobustcmd{\Hgreen}{\cmdkl{\mathcal{H}}}
\knowledgenewrobustcmd{\Jleq}{\mathbin{\cmdkl{\leq_{\mathcal{J}}}}}
\knowledgenewrobustcmd{\Lleq}{\mathbin{\cmdkl{\leq_{\mathcal{L}}}}}
\knowledgenewrobustcmd{\Rleq}{\mathbin{\cmdkl{\leq_{\mathcal{R}}}}}
\knowledgenewrobustcmd{\Hleq}{\mathbin{\cmdkl{\leq_{\mathcal{H}}}}}
\knowledgenewrobustcmd{\Jheight}[1]{\cmdkl{\mathsf{depth}(#1)}}

\knowledgenewrobustcmd{\idempotents}[1]{\cmdkl{\mathsf{Idem}(#1)}}
\knowledgenewrobustcmd{\evalmorph}[1]{\cmdkl{\psi_{#1}}}

\knowledgenewrobustcmd{\Jlength}[1]{\cmdkl{\mathcal{J}\mathsf{len}(#1)}}

\knowledgenewrobustcmd{\push}[2]{\cmdkl{\mathsf{push}(#1 \blacktriangleright #2)}} 
\knowledgenewrobustcmd{\pop}[2]{\cmdkl{\mathsf{pop}(#1 \blacktriangleleft #2)}} 
\knowledgenewrobustcmd{\Summaries}{\cmdkl{\mathsf{Summaries}}} 

\knowledgenewrobustcmd{\blaskd}{{A}}

\knowledgenewrobustcmd{\CFG}{\cmdkl{\mathcal{C}_\Gg}}

\newcommand{\inc}{\mathtt{inc}}
\newcommand{\one}{\mathbf{1}}

\newcommand{\duplicate}{D}
\newcommand{\zero}{\mathbf{0}}

\definecolor{green}{HTML}{5EE534}
\definecolor{Red2}{HTML}{CC0400}
\definecolor{Orange2}{HTML}{E6670A}
\definecolor{Violet2}{HTML}{9659d0}
\definecolor{Green3}{HTML}{45A229}
\definecolor{Navy}{HTML}{7973E2}

\definecolor{red1}{HTML}{E6276D}
\definecolor{blue1}{HTML}{1E88E5}
\definecolor{yel1}{HTML}{BF9000}

\definecolor{DarkGamboge}{HTML}{be7c00}

\usetikzlibrary{arrows, arrows.meta,calc,automata,fit,shapes,positioning}
\tikzset{AUT style/.style={>=angle 60,initial text= ,every edge/.append,every state/.style={minimum size=20,inner sep=2}}}

\usetikzlibrary{decorations.pathreplacing}

\usetikzlibrary{decorations.pathmorphing}

\whenFull{\hideLIPIcs}{}

\begin{document}

	\maketitle
	
	\begin{abstract}

	Indexed languages are a classical notion in formal language theory,
which has attracted attention in recent decades due to its role in higher-order
model checking: They are precisely the languages accepted by order-2 pushdown automata.

The downward closure of an indexed language---the set of all
(scattered) subwords of its members---is well-known to be a regular
over-approximation. It is known since 2015 that the downward
closure of a given indexed language is effectively computable.  However, the
algorithm comes with no complexity bounds, and it has remained open
whether a primitive-recursive construction exists.

We settle this question and provide a triply (resp.\ quadruply)
exponential construction of a non-deterministic (resp.\ deterministic)
automaton. We also prove (asymptotically) matching lower bounds.
For the upper bounds, we rely on recent advances in semigroup theory,
which let us compute bounded-size summaries of words with respect to a finite
semigroup.  By replacing stacks with their summaries, we are able to transform
an indexed grammar into a context-free one with the same downward closure, and
then apply existing bounds for context-free grammars.

	\end{abstract}

\section{Introduction}

\myparagraph{Downward closures.}
For finite words $u$ and $v$, we say that $u$ is a \emph{(scattered) subword} of $v$, written $u\subword v$, if $v$ can be obtained from $u$ by inserting letters.
The \emph{downward closure} of a language $L\subseteq\Sigma^*$ is the set $\dcl{L}=\{u\in\Sigma^* \mid u\subword v~\text{for some $v\in L$}\}$ of all subwords of members of $L$.
By the well-known Higman’s lemma~\cite{Higman52}, the downward closure $\dcl{L}$ is regular for \emph{any} language $L$. 

This makes the downward closure a fundamental abstraction---it is a regular overapproximation
that preserves information about 
pattern occurrences and unboundedness behavior. Specifically, in the
verification of complex systems, downward closures are often used to replace
infinite-state components by finite-state ones, which then enables the
algorithmic analysis of the entire system. Examples include asynchronous
programs~\cite{MajumdarTZ22,0001GMTZ23a}, shared-memory systems with dynamic
thread
creation~\cite{AtigBQ11,DBLP:conf/icalp/0001GMTZ23,AiswaryaCBSSZ2026,DBLP:journals/pacmpl/BaumannMTZ21,BaumannMTZ22},
parameterized asynchronous shared-memory
systems~\cite{TorreMW15}, and systems communicating via lossy channels~\cite{DBLP:conf/concur/AtigBT08,AbdullaBB01,DBLP:conf/lics/AiswaryaMS22}.

For these reasons, it is often useful to \emph{compute downward closures}, which means constructing a finite automaton for $\dcl{L}$ when given a description (i.e.\ a grammar or an infinite-state recognizer) for $L$.  This is a notoriously difficult task, as it requires a deep understanding of how $L$ is generated/recognized. Therefore, the problem of computing downward closures has attracted significant attention over recent decades, with work on
context-free languages~\cite{%
	van1978effective,
	Courcelle91},
systems with counting and concurrency~\cite{
	DBLP:conf/icalp/HabermehlMW10,
	DBLP:conf/mfcs/AtigMMS17,
	AiswaryaCBSSZ2026,
	DBLP:conf/lics/AtigCHKSZ16,
	DBLP:conf/stacs/Zetzsche15},
models of higher-order recursion~\cite{
	Zetzsche15,
	DBLP:conf/popl/HagueKO16,
	DBLP:conf/lics/ClementePSW16},
lossy channel systems~\cite{
	AbdullaBB01,
	DBLP:journals/tcs/Mayr03},
general algorithms for broad classes of infinite-state systems~\cite{
	Zetzsche15,
	DBLP:conf/lics/AnandSSZ24},
representation sizes~\cite{
	DBLP:journals/fuin/GruberHK09,
	MajumdarTZ22,
	BaumannMTZ22,
	DBLP:conf/lata/BachmeierLS15},
related algorithmic tasks~\cite{
	Zetzsche16,
	DBLP:conf/fsttcs/Parys17,
	fazekas_et_al:LIPIcs.ISAAC.2024.28,
	DBLP:conf/wia/FazekasKKMMS25,
	DBLP:conf/mfcs/AmarilliMRS25,
	DBLP:journals/corr/abs-1904-10105,
	DBLP:conf/stacs/GanardiSZ24},
and even computability beyond subwords~\cite{
	DBLP:conf/concur/AnandZ23,
	DBLP:conf/lics/AnandSSZ24,
	DBLP:conf/lics/Zetzsche18,
	DBLP:journals/fuin/BarozziniCCP22,
	DBLP:conf/tagt/CourcelleS94}.

\myparagraph{Indexed languages.} A setting
where downward closure computation is particularly challenging is that of
\emph{indexed languages}~\cite{Aho68}, a classical notion in formal language
theory that generalizes context-free languages. Essentially, indexed grammars differ from context-free grammars in that each non-terminal carries a stack, which can be pushed and popped through special rules. These grammars have recently attracted
interest because of their role in higher-order model-checking~\cite{DBLP:conf/lics/Ong15,DBLP:journals/jacm/Kobayashi13}: 
Indexed grammars
are equivalent to order-2 pushdown automata, from the hierarchy of
\emph{higher-order pushdown automata (HOPA)}, which model safe higher-order recursion~\cite{DBLP:conf/fossacs/KnapikNU02,DBLP:conf/popl/Kobayashi09} (see also the survey~\cite{DBLP:conf/lics/Ong15}). Level $k$ of this hierarchy consists of the \emph{order-$k$ pushdown automata ($k$-PDA)}, which have access to \emph{stacks
of (stacks of (...)) stacks}---with nesting depth $k$: In particular, a 1-PDA is an ordinary pushdown automaton, whereas a 2-PDA 
has a stack of stacks, where it can operate on the top-most stack as an
ordinary pushdown automaton; but it can also copy the top-most stack. 

There is a significant gap in our understanding of downward closures of indexed
languages (and HOPA more generally). There is a general approach for computing
downward closures~\cite{Zetzsche15}. Based on this approach, computability of
downward closures has been shown for indexed languages (equivalently,
2-PDA)~\cite{Zetzsche15}, then for general
HOPA~\cite{DBLP:conf/popl/HagueKO16}, and even higher-order recursion
schemes~\cite{DBLP:conf/lics/ClementePSW16}.

However, while downward closure computability is
settled for these models, the complexity has remained a long-standing open problem. This
is because the algorithm from \cite{Zetzsche15} enumerates automata and then
solves instances of the so-called \emph{simultaneous unboundedness problem (SUP)} to
decide whether the current automaton in fact recognizes $\dcl{L}$. Thus,
although the complexity of the SUP itself is well-understood for
HOPA~\cite{DBLP:conf/fsttcs/Parys17,DBLP:conf/icalp/BarozziniPW22}, there are
no complexity upper bounds for the entire task of computing downward closures of indexed languages
(nor for languages recognized by HOPA in general).  In fact, even the existence of a
primitive-recursive (or even hyper-Ackermannian) upper bound remained open\footnote{Given that computing
downward closures involves a well-quasi-ordering (WQO) at its core, and
WQO-based algorithms employing the subword ordering can often be furnished with
hyper-Ackermannian/multiply-recursive upper bounds via length-function
theorems~\cite{DBLP:conf/icalp/SchmitzS11,DBLP:books/hal/Schmitz17}, one might
hope for such a bound here. Unfortunately, it is not clear how to apply
length-function theorems to downward closure computation.}.

\myparagraph{Contribution.} In this work, we settle the complexity of computing
downward closures of indexed languages. We show that given an indexed language
$L$, one can compute a non-deterministic finite automaton (NFA) for $\dcl{L}$
in triply exponential time (and hence of triply exponential size). We also
provide a triply exponential lower bound, improving on the doubly exponential
lower bound in~\cite{Zetzsche16}.

Furthermore, our constructions also provide a tight bound for the computation
of a deterministic finite automaton (DFA) for $\dcl{L}$: It follows that one
can construct a quadruply-exponential-sized DFA for $\dcl{L}$, and we provide a
quadruply-exponential lower-bound.

Moreover, our results settle the complexity of decision problems involving
downward closures of indexed languages. The \emph{downward closure inclusion
problem} asks, whether for given indexed languages $L_1$ and $L_2$, we have
$\dcl{L_1}\subseteq\dcl{L_2}$. Similarly, the \emph{downward closure
equivalence problem} asks whether $\dcl{L_1}=\dcl{L_2}$. We show that
both problems are $\coNEXP[3]$-complete.

Finally, our construction for the downward closure lower bound also settles
another question about indexed languages: It implies that the known triply
exponential upper bound on the pumping threshold for indexed
grammars~\cite{Hayashi1973,DBLP:journals/iandc/Smith17} (the maximal length of
words in a finite indexed language) is in fact asymptotically tight. (In
\cite{DBLP:conf/csl/Kartzow11}, a doubly exponential upper bound was claimed. However, as confirmed by the author of \cite{DBLP:conf/csl/Kartzow11}, that is a miscalculation; see \whenFull{\cref{sec:results,app:results}}{\cref{sec:results} and the full version~\cite{MandelMZ26arxiv}}.)

\myparagraph{Why are the results unexpected?} The complexity results come as a
considerable surprise. This is because (tight) complexity bounds for HOPA are usually towers of exponentials where the height grows linearly with the order. For example, for $k\ge 1$, the emptiness problem for $k$-PDA is $\EXPTIME[(k-1)]$-complete~\cite[comment before Thm.~7.12]{DBLP:journals/iandc/Engelfriet91} (note that $\EXPTIME[0]=\PTIME$). The same is true of the SUP~\cite[Thm.~3]{DBLP:conf/fsttcs/Parys17}. 
Since downward closure inclusion and equivalence are $\coNP$-complete for NFAs (see \cite[Sec.~5]{DBLP:conf/lata/BachmeierLS15} and \cite[Prop.~7.3]{DBLP:journals/tcs/KarandikarNS16}) and $\coNEXP[1]$-complete for $1$-PDAs~\cite[Tab.~1]{Zetzsche16}, and hence $\coNEXP[k]$-complete for $k$-PDAs with $k\le 1$, one would expect $\coNEXP[2]$ rather than $\coNEXP[3]$ for $2$-PDAs. In fact, the best (and only) known lower bound until now has been $\coNEXP[2]$~\cite[Cor.~18]{Zetzsche16}. 
Similarly, since downward closure NFAs are polynomial-sized for given NFAs (sometimes considered $0$-PDAs) and exponential-sized for given $1$-PDAs~\cite[Cor.~6]{DBLP:conf/lata/BachmeierLS15}, one would expect a doubly exponential bound for $2$-PDAs. In fact, the best known lower bound on the NFA size for downward closures had been doubly exponential~\cite[remarks before Cor.~18]{Zetzsche16}.

\myparagraph{Key ingredients.}
Indexed grammars extend context-free grammars by equipping the nodes of derivation trees with pushdown store (which we also call stack). This way, each branch of a derivation tree corresponds to a run of a pushdown automaton. 

Our construction relies on recent advances in finite semigroup theory, which provide succinct “summaries’’ of arbitrarily long words relative to a finite semigroup~\cite{GimbertMT25arxiv}. We use this (for a suitable semigroup) to replace the stack in each derivation tree node by such a summary, each of which uses exponentially many bits. This replacement changes the overall language, but preserves the downward closure. Once the information in each node is bounded, we can transform the grammar into a doubly-exponential-sized context-free grammar. This yields the triply exponential upper bound overall, since for context-free grammars, existing algorithms yield exponential-sized downward closure NFAs~\cite[Corollary 6]{DBLP:conf/lata/BachmeierLS15}.

The aforementioned summaries are similar in spirit to Simon's factorization
forests~\cite{Simon90}. The latter annotate words
by trees: Relative to a morphism $\varphi\colon\Sigma^*\to M$ into a finite
monoid $M$, a factorization forest for $w\in\Sigma^*$ is a tree of height
bounded by a function of $|M|$ that allows evaluating $\varphi$ on infixes of $w$ without
processing the entire infix. Here, a crucial idea is that a sequence of infixes
that all map under $\varphi$ to the same idempotent $e$ must evaluate to $e$. 

Similar to factorization forests, our summaries also exploit repetitions of
idempotents to collapse long infixes in a stack.  However, in contrast to
factorization forests, our summaries reduce the stack word to one of bounded
length (hence losing information).  Also crucially, taking summaries is
compatible with pushing stack symbols: Given a summary of $w$, one can compute a
summary of $aw$.

\myparagraph{Structure of the paper.}
We recall necessary notations and basic results in Section~\ref{sec:prelims}, and state the main results in Section~\ref{sec:results}.
Then in Section~\ref{sec:sound} we show how to make a given indexed grammar \emph{productive}, meaning that every partial computation can be extended to a full one. 
In Section~\ref{sec:monoid} we introduce an object at the core of our construction, the \emph{production monoid}. 
In Section~\ref{sec:pump-skip} we introduce new rules that extend and reduce parts of the stack without altering the downward closure of the language. We then proceed in Section~\ref{sec:summaries} to define \emph{summaries} of stack contents with respect to the production monoid.
In Section~\ref{sec:CFG} we leverage those summaries to compute from an indexed grammar a context-free one with the same downward closure, to which known constructions apply.
Finally, in Section~\ref{sec:lower-bound} we complement our upper bounds with matching lower bounds.

\AP \emph{This paper contains internal links; every occurrence of a "term@@example" is linked to its ""definition"". The reader can click on "terms@@example" (and some notations), or simply hover over them on some pdf readers, to get their definition.}

\whenFull{}{\emph{Due to space constraints, most of the proofs are deferred to the full version~\cite{MandelMZ26arxiv}.}}

\section{Preliminaries}
\label{sec:prelims}

\myparagraph{Words, trees and languages.}
Given a finite alphabet $\Sigma$, we write $\Sigma^*$ for the set of words over $\Sigma$, and $\Sigma^k$ for the set of words of length $k$. Given $w \in \Sigma^*$, we denote $|w|$ its length.
We assume familiarity with basic finite automata theory, see~\cite{Sakarovitch09} for an introduction.

\AP We write $u \intro*\subword v$ if $u$ is a (scattered) ""subword"" of $v$; that is, if $u$ can be obtained from $v$ by removing some of its letters.
If a word $w$ is equal to $uv$, then we call $u$ a ""prefix"", and $v$ a ""suffix"" of $w$.
The ""downward closure"" of a language $L \subseteq \Sigma^*$ is the set of subwords of words of $L$, and is denoted $\intro*\dcl{L}$. An important reason for studying "downward closures" is that they are always regular. This was first shown by Haines~\cite[Theorem 3]{haines1969free}, but also follows from $\subword$ being a well-quasi ordering, which was shown earlier by Higman~\cite[Theorem 4.3]{Higman52}:
\begin{theorem}
	The "downward closure" of any language is regular.
\end{theorem}

\AP A finite ordered $\Sigma$-labeled binary tree is a pair $(\tau, \lambda)$ with $\tau$ a finite prefix-closed subset of $\set{0,1}^*$ such that for all $\nu \in \tau$, $\nu 1 \in \tau$ implies $\nu0 \in \tau$, and $\lambda : \tau \to \Sigma$ a function mapping each node in $\tau$ to a label.
We use the usual terminology for trees, with \emph{node, leaf, branch, subtree, parent, child, ancestor, descendant} etc. retaining their usual meaning. The ""leaf word"" of $\tau$ is the word obtained by concatenating $\lambda(\nu_1)\dots \lambda(\nu_k)$, where $\nu_1, \dots, \nu_k$ are the leaves of $\tau$ read from left to right, i.e., in lexicographic order.

\myparagraph{Indexed grammars.}
\label{sec:defIG}
To simplify definitions and proofs, we use a syntax for "indexed grammars" that resembles the Chomsky normal form used for context-free grammars. 

	\label{def:IG}
	An ""indexed grammar"" is a tuple $\Gg=(N, T, I, P, S)$ with
	\begin{itemize}
		\setlength\itemsep{0em}
		\item $N$ the set of non-terminals, and $S\in N$ the starting non-terminal
		\item $T$ the set of terminal symbols
		\item $I$ the set of index symbols, also called stack symbols
		\item $P$ a set of productions, which are of the following types:
		\begin{itemize}
			\item $A \to w$  with $w \in  T^*$.
			
			\item $A \to B C$  with $A,B,C \in N$
			
			\item $A \to B f$  with $A,B \in N$ and $f \in I$
			
			\item	$A f \to B$  with $A,B \in N$ and $f \in I$
		\end{itemize}
	\end{itemize}	
	We define the ""size@@IG"" of $\Gg$ as $|N|+|P|+\sum_{A \to w \in P}|w|$.	
	A ""context-free grammar"" (or "CFG") is an "indexed grammar" where all productions are of the two first forms.

\AP 
A ""sentential form"" is a word over the (infinite) alphabet $N I^* \cup T$. The set of "sentential forms" is denoted $\intro*\sentforms_{\Gg}$, or simply $\sentforms$ when the grammar is clear from context. We write the elements of $NI^*$ as $A[z]$, with $A\in N$ a non-terminal and $z\in I^*$ interpreted as the content of its stack; such an element is sometimes referred to as a ""term"". If $u\in \sentforms$, we occasionally use the notation $u[z]$ to denote the sentential form obtained by pushing $z$ on top of the stack of every non-terminal in $u$. The derivation relation $\intro*\derivesindexed_{\Gg}$ (or simply $\derivesindexed$) over $\sentforms$ is defined as:
\begin{alignat*}{3}
	&u A [z] v &  ~&\derivesindexed~u B[z] C[z] v&\quad\quad\quad &\text{if } A \to BC \in P,\\  
	&u A [z] v &  &\derivesindexed~u B [fz] v&   &\text{if } A \to Bf \in P,\\
	&u A [fz] v&  &\derivesindexed~u B[z] v&  &\text{if } A f \to B \in P,\\
	&u A [z] v &  &\derivesindexed~u w v&  &\text{if } A \to w \in P.
\end{alignat*}
for all $u,v \in \sentforms$, $A,B,C \in N$, $f \in I$, $z \in I^*$ and $w \in T^*$.
We write $\intro*\derivesindexedstar_{\Gg}$ (or simply $\derivesindexedstar$) for the reflexive transitive closure of $\derivesindexed_\Gg$.

The "subword" relation $\subword$ is extended to $\sentforms$ in the natural way, i.e. $u\subword v$ if $u$ can be obtained from $v$ by deleting "terms". Note that on $\sentforms$, the ordering $\subword$ is not a well-quasi ordering, since any two "terms" $A[z]$ and $A[z']$ are incomparable for $z,z'\in I^*$, $z\ne z'$.

\AP Given a "sentential form" $u$, we write $\intro*\langSF{u}$ for the set  of "sentential forms" derivable from $u$, $\set{v \in \sentforms \mid u \derivesindexedstar v}$.
The ""language@@indexed"" of $\Gg$ is denoted $\intro*\langIG{\Gg}$ and defined as $\langSF{S} \cap T^*$.
Furthermore, for all $X \subseteq N$ and $u \in \sentforms$, the language $\intro*\langX{u}{X}$ is defined as the set $\langSF{u} \cap (X \cup T)^*$ of "sentential forms" derivable from $u$ with all stacks empty, and all non-terminals in $X$. 
In particular, 
$\langX{u}{\emptyset}$ is the set of terminal words which can be derived from $u$.
If the language $\langX{u}{\emptyset}$ is non-empty then we say that $u$ is ""productive@@SF"".

\AP A ""derivation tree"" is a finite ordered tree $\tau$ whose nodes are labeled by elements of $NI^* \cup T^*$, with the following constraints. For each node $\nu$ with label $A[z]$ (where $A\in N$ and $z\in I^*$), exactly one of the following holds 
\begin{itemize}%
	\item $\nu$ is a leaf
	
	\item $\nu$ has one child labeled $w$ for some $A \to w \in P$
	
	\item $\nu$ has two children labeled $B[z]$ and $C[z]$, for some production $A \to B C \in P$ 
	
	\item $\nu$ has one child labeled $B [fz]$ for some $A \to B f \in P$
	
	\item $\nu$ has one child labeled $B [z']$ for some $A f \to B \in P$, with $z = fz'$.
\end{itemize}
A derivation tree whose root is labeled $A[z]$ and whose "leaf word" is equal to $u$ is called a \emph{derivation tree from $A[z]$ to $u$}. If $u \in T^*$ we say that the "derivation tree" is ""complete@@tree"".

\begin{remark}
	An easy induction shows that for all $A[z] \in N I^*$, a "sentential form" $u$ can be derived from $A[z]$ if and only if $u$ is the "leaf word" of a "derivation tree" whose root is labeled $A[z]$. 
	In the forthcoming proofs we will either use "sentential forms" or "derivation trees" depending on what is most convenient.
\end{remark}

\begin{remark}
	When describing examples of "indexed grammars", we will use rules of the form $A f \to u$ and $A \to u$ with $u \in (N \cup T)^*$. 
	This is just syntactic sugar, as we can replace them with rules from the definition of "indexed grammar" while adding a small number of non-terminals, in the spirit of the Chomsky normal form~\cite{Chomsky59}. 
	This transformation incurs a linear size increase of the grammar. 
	Details are left to \whenFull{Appendix~\ref{app:Chomsky}.}{the full version~\cite{MandelMZ26arxiv}.}
\end{remark}

\begin{example}
	We can define the language $\set{\ltr{a}^n \ltr{b}^{n^2} \mid n\in \NN}$ with an indexed grammar:
	\begin{align*}
		&S \to Tg  &\grammarcommand{\#  $g$ is the stack bottom symbol;} \\
		&T \to Tf  &\grammarcommand{\#  we push some number $n$ of $f$;} \\
		&T \to A   &\\
		&Ag \to \epsilon  &\grammarcommand{\#  if $n = 0$ then return $\epsilon$;}\\
		&Af \to C   &\grammarcommand{\#  if $n > 0$  we pop an $f$;} \\
		&C \to \ltr{a} A B   &\grammarcommand{\#  repeat  $n$ times to get $\ltr{a}^nB[g] B[fg] \cdots B[f^{n-1}g]$;} \\
		&Bf \to \ltr{b}\ltr{b}B   &%
        \\
		&Bg \to \ltr{b}  &\grammarcommand{\#the $B$s output $\ltr{b}^{\sum_{i=0}^{n-1}(2i+1)} = \ltr{b}^{n^2}$.}
	\end{align*}

\end{example}

\knowledgenewrobustcmd{\Nlabel}{\cmdkl{\alpha}}

\begin{remark}
	\label{rmk:label-push}
	We can assume without loss of generality that every symbol pushed on the stack carries the information of the production rule used to push it (this property can always be ensured by adding a quantity of new stack symbols and production rules that is at most quadratic in the size of the grammar). Formally, we assume that there are functions $\alpha, \beta : I \to N$ such that for every push rule $A \to B f$ we have $A = \alpha(f)$ and $B = \beta(f)$.
	We also assume that for all $f \in I$ there is a rule of the form $A \to B f$ in $P$. Clearly stack symbols not satisfying this condition can be removed.
\end{remark}

\myparagraph{Complexity.}
We define the functions $\exp_k\colon\NN\to\NN$ inductively by setting
$\exp_0(n)=n$ and $\exp_{k+1}(n)=2^{\exp_k(n)}$. A function $f\colon\NN\to\NN$
is \emph{(at most) $k$-fold exponential} if there is a constant $c>0$ such that
$f(n)\le \exp_k(n^c)$ for almost all $n$. We say that $f$ is \emph{at least
$k$-fold exponential} if there is a constant $c>0$ such that
$f(n)\ge\exp_k(n^c)$ for infinitely many $n$. Instead of $2$-fold, $3$-fold,
$4$-fold exponential, resp., we also say doubly, triply, or quadruply
exponential, resp. For $k\ge 1$, the class $\coNEXP[k]$ consists of the
complements of sets accepted by an $f$-time-bounded non-deterministic
Turing machine, for some $f\colon\NN\to\NN$ that is at most $k$-fold exponential.

\section{Main results}\label{sec:results}

In this section, we present the main results of this work. 

\myparagraph{Non-deterministic automata.}
Our first main
contribution is an algorithm to compute an NFA of at most triply exponential
size for the downward closure of an indexed language: 
\begin{theorem}\label{main-upper-bound}
Given an indexed grammar $\Gg$, one can compute (in triply exponential time) a triply-exponential-sized NFA for $\dcl{\langIG{\Gg}}$.
\end{theorem}

We also provide an asymptotically matching lower bound, which we infer from the following result:

\begin{theorem}\label{construction-lower-bound}
	Given $n\in\NN$ (in unary encoding), we can compute in polynomial time an indexed grammar for the language $\{\ltr{a}^{\exp_3(n)}\}$.
\end{theorem}

An NFA for $\dcl{\{\ltr{a}^{\exp_3(n)}\}}$ clearly requires at least
$\exp_3(n)$ states, implying \cref{main-lower-bound}.

\begin{corollary}\label{main-lower-bound}
	There is a family $(\Gg_n)_{n\ge 1}$ of indexed grammars of size polynomial in $n$
	such that any NFA for $\dcl{\langIG{\Gg_n}}$ requires at least $\exp_3(n)$
	states.
\end{corollary}

\myparagraph{Deterministic automata.}
Of course, \cref{main-upper-bound} implies a quadruply exponential upper bound for deterministic automata:
\begin{corollary}\label{main-upper-bound-dfa}
Given an indexed grammar $\Gg$, one can compute (in quadruply exponential time) a DFA of at most quadruply exponential size for $\dcl{\langIG{\Gg}}$.
\end{corollary}
Here, we have an asymptotically matching lower bound as well:
\begin{theorem}\label{main-lower-bound-dfa}
There is a family $(\Gg_n)_{n\ge 1}$ of indexed grammars of size polynomial in $n$
such that any DFA for $\dcl{\langIG{\Gg_n}}$ requires at least $\exp_4(n)$
states.
\end{theorem}

\myparagraph{Downward closure comparisons.}
\cref{main-upper-bound} and our construction for \cref{construction-lower-bound} also allow us
to settle the complexity of decision problems related to downward closures. The ""downward closure inclusion problem"" (for indexed
languages) asks whether two given indexed languages $L_1,L_2$ satisfy
$\dcl{L_1}\subseteq\dcl{L_2}$.  Similarly, the ""downward closure
equivalence problem"" asks whether $\dcl{L_1}=\dcl{L_2}$.

In \cite[Corollary~18]{Zetzsche16}, it was shown that downward closure
inclusion and equivalence for indexed languages are $\coNEXP[2]$-hard. Here, we
settle their precise complexity:
\begin{theorem}\label{main-completeness}
Downward closure inclusion and downward closure equivalence for indexed
languages are $\coNEXP[3]$-complete.
\end{theorem}
Here, the upper bounds follow from \cref{main-upper-bound} and the fact that
"downward closure inclusion" and "equivalence@downward closure equivalence" are
$\coNP$-complete for NFAs (see \cite[Section 5]{DBLP:conf/lata/BachmeierLS15}
and \cite[Proposition 7.3]{DBLP:journals/tcs/KarandikarNS16}).

\myparagraph{The pumping threshold for indexed languages.}
Our lower bound technique also settles the growth of the "pumping threshold" of indexed grammars. Consider the function
\begin{multline*}
	\intro*\PumpConst(n)=\max\{|w| \mid \text{there is an indexed grammar $\Gg$ of size $\le n$}\\
	\text{with $w\in\langIG{\Gg}$ such that $\langIG{\Gg}$ is finite}\} 
\end{multline*}
We call $\PumpConst(n)$ the ""pumping threshold"" (for size $n$) because
placing an upper bound on $\PumpConst(n)$ usually involves a pumping argument.

An analogous pumping threshold function for NFAs and CFGs is well-understood:
It is linear for NFAs (see~\cite{DBLP:journals/corr/abs-2309-02757} for related
results) and exponential for CFGs (exact bounds can be found
in~\cite{DBLP:conf/birthday/Gruber025}). For indexed grammars, there are two
proofs of a triply exponential upper bound: the pumping lemmas of
Hayashi~\cite[Theorem 5.1]{Hayashi1973} and Smith~\cite[Theorem
1]{DBLP:journals/iandc/Smith17} (Smith's proof mentions the bound explicitly;
Hayashi's proof requires some analysis for this). We complement this by showing
a triply exponential lower bound:
\begin{corollary}
$\PumpConst$ grows at least triply exponentially.
\end{corollary}
This follows from \cref{construction-lower-bound}, because
$\{\ltr{a}^{\exp_3(n)}\}$ is finite but has an indexed grammar of size
polynomial in $n$.  It should be noted that \cite[Section
7]{DBLP:conf/csl/Kartzow11} claims a doubly exponential upper bound for
$\PumpConst$. Unfortunately, as confirmed by the author of
\cite{DBLP:conf/csl/Kartzow11}, that is due to a miscalculation, see
\whenFull{\cref{app:results}}{the full version~\cite{MandelMZ26arxiv}} for a
discussion.
\begin{remark}
	A triply exponential upper bound on $\PumpConst(n)$ also follows from
	our \cref{main-upper-bound}: If an indexed grammar $\Gg$ generates a
	word that is longer than the number of states of an NFA for
	$\dcl{\langIG{\Gg}}$, then $\dcl{\langIG{\Gg}}$ must be infinite, and hence also $\langIG{\Gg}$.
\end{remark}

\myparagraph{Structure of the paper.} In
\cref{sec:sound,sec:monoid,sec:pump-skip,sec:summaries,sec:CFG}, we will
prove \cref{main-upper-bound}, from which all remaining upper bounds follow. In
\cref{sec:lower-bound}, we will then prove all lower bounds.

\section{Productiveness}\label{sec:sound}

For the rest of the paper, we assume that our input grammar $\Gg$ generates a
non-empty language.  This is because emptiness of indexed languages can be decided in
$\EXPTIME$~\cite[Theorem 7.12]{DBLP:journals/iandc/Engelfriet91}, and if
$\langIG{\Gg}$ is empty, an NFA for $\dcl{\langIG{\Gg}}$ is immediate. 

However, we will need to establish the stronger guarantee of \emph{productiveness},
which expresses the absence of deadlocks. This means that if we can produce a "sentential form" $u$, then from $u$ we can derive a terminal word in $T^*$.

\myparagraph{Productive grammars.} \AP We say that $\Gg$ is ""productive@@grammar"" if for all $u \in \langSF{S}$ we have $\langX{u}{\emptyset} \neq \emptyset$, that is, from every "sentential form" obtained from $S$ we can derive a word in $T^*$.

This property is especially useful for "downward closure" computation.
In a "productive@@grammar" indexed grammar, we can observe:
for all $u \in \langSF{S}$, for all $v$ such that $v \subword u$, we have $\langX{v}{\emptyset} \subseteq \dcl{\langX{u}{\emptyset}}$. In other words, we can interleave derivations and deletion of terms without any risk of obtaining ``extra'' terminal words. This property does not hold in general, because deleting "terms" that cannot produce terminal words may allow the derivation of a terminal word that is not in the downward closure.

\begin{example}
	The following grammar $\Gg$ is not "productive@@grammar".
\begin{align*}
	&S \to S \bot& &S \to Sf & &S \to Sg& \\
	&S \to AA& &S\to AAC& & &\\
	&Af \to \ltr{a} A& &Ag \to \ltr{b}&  & &\\
	&Cf \to \ltr{c} C& &Cg \to CAB&  & &\\
	&A\bot \to \epsilon & &C\bot \to \epsilon.& & & 	
\end{align*}
The language of $\Gg$ is $\{w^2 \mid w\in \{\ltr{a},\ltr{b}\}^*  \} \cup \{\ltr{a}^{2n}\ltr{c}^n \mid n > 0\}$. In particular, by setting $u=A[gf\bot]A[gf\bot]C[f\bot]A[f\bot]B[f\bot]$ and $v = C[f\bot]A[f\bot]$ we have $v\subword u$ and $u\in \langSF{S}$, but $\ltr{c}\ltr{a} \in \langX{v}{\emptyset}$ while $ \dcl{\langX{u}{\emptyset}} = \emptyset$ (since $B$ cannot produce a terminal word). Moreover, $ca \not\in \dcl{\langIG{\Gg}}$. In this case, it is easy to modify $\Gg$ to obtain a "productive@@grammar" grammar which generates the same language.
\end{example}

\myparagraph{Tracking productivity of non-terminals.} Before we 
describe our method for achieving "productiveness@@grammar", we observe that tracking the non-terminals which, with a given stack content, can derive a "sentential form" with terms confined to a certain subset, gives rise to a left semigroup action.
\begin{definition}
	For  $X \subseteq N$ and $z \in I^*$, we define
	\[ z \cdot X=\set{A \in N \mid \langX{A[z]}{X} \neq \emptyset}, \]
	i.e.\ $z\cdot X$ is the set of non-terminals $A$ so that the "term" $A[z]$ can derive a "sentential form" consisting of terminals and empty-stack occurrences of non-terminals in $X$.
	Let $\Useful$ be defined as the set of non-terminals which can derive a word in $T^*$, i.e. 	\[\intro*\Useful = \set{A \in N \mid \langX{A}{\emptyset} \neq \emptyset}.\]
\end{definition}
Note that $\langIG{\Gg}$ is empty if and only if $S \notin \Useful$.

We will often rely on the fact that $I^*$ acts (on the left) as a semigroup\footnote{However, it does not act as a monoid, because $\varepsilon\cdot X$ it not necessarily $X$, as every set $z\cdot X$ includes $\Useful$.} on the power set $\powerset{N}$, as we state now:
\begin{restatable}{lemma}{lemSemigroupActionProperties}\label{lem:semigroup-action-properties}
	For every $f\in I$, $z\in I^*$, and $X\subseteq N$, we have $fz\cdot X=f\cdot (z\cdot X)$. Moreover, $z\cdot\Useful=z\cdot\emptyset$.
\end{restatable}
Here, only the inclusion $fz\cdot X\subseteq f\cdot (z\cdot X)$ is not trivial:
It holds because a derivation that eliminates $fz$ from the stack must in each
branch first remove $f$, and then $z$. 
\whenFull{The proof is in
\cref{app:semigroup-action-properties}.}{}

\myparagraph{Achieving productiveness.} \AP We are now ready to present our
construction of a productive equivalent of $\Gg$. Formally, the ""annotated
version"" of $\Gg$, denoted $\Ggb = (\Nb, T, \Ib, \Pb, \Sb)$, is defined as
follows.
The general idea is to integrate in to label every point in the stack with a set of non-terminals, which are the non-terminals $A$ such that if we take the stack content $z'$ from that point, one can derive a terminal word from $A[z']$. By maintaining this information, we avoid producing "terms" that cannot derive a terminal word.

 Recall that we assumed $\langIG{\Gg} \neq \emptyset$, thus $S \in \Useful$.
Otherwise, we have:
\begin{itemize}[leftmargin=5mm]
	\item $\Nb = \set{(A,X) \mid X \subseteq N, A \in X}$, $\Ib = I \times 2^N$, and $\Sb = (S, \Useful)$,
	\item $\Pb$ contains the following rules:
		\begin{itemize}[leftmargin=5mm]
		\item $(A, X) \to w$ for all $A \to w \in P$ with $A\in X$ and $w \in T^*$
		\item $(A, X) \to (B,X) (C,X)$ for all $A \to B C \in P$ with $A,B,C \in X$
		\item $(A, X) \to (B,Y)(f, X)$ for all $A \to B f \in P$ with $Y = f \cdot X$, $A \in X$  and $B \in Y$
		\item $(A, Y) (f, X) \to (B,X) $ for all $A f \to B \in P$ with $Y = f \cdot X$, $A \in Y$ and $B \in X$ 
	\end{itemize}
\end{itemize}

\knowledgenewrobustcmd{\annotate}[2]{\cmdkl{\overline{#1}^{#2}}}
Note that a stack word $\zb = (f_n, X_n) \cdots (f_1, X_1) \in \Ib^*$ appearing as an infix of a stack content in a derivation of $\Ggb$ must be so that $X_i = f_{i-1}\cdot X_{i-1}$ for all $i>1$. A stack word satisfying this property is determined by its projection onto $I^*$ and its last element. 
Given $z = f_n\cdots f_1 \in I^*$ and $X \subseteq N$, define $\intro*\annotate{z}{X}$ as $(f_n, X_n) \cdots (f_1, X_1)$ with $X_1 = X$ and $X_{i+1} = f_i \cdot X_i$ for all $i<n$. Given $A [z]\in NI^+$ with $z=f_n \cdots f_1 $ and $A \in z \cdot X$, the ""$X$-based annotation"" of $A[z]$ is the $\Ggb$ "term" $(A,Y)[\zb] = (A,Y) [(f_n,X_n) \cdots (f_1, X_1)] \in \Nb \,\Ib^+$ such that $X_1 = X$, $Y = f_n \cdot X_n$ and $X_i = f_{i-1}\cdot X_{i-1}$ for all $i>1$.
In particular, by Lemma~\ref{lem:semigroup-action-properties} we have $X_i = f_{i-1} \cdots f_1 \cdot X$ for all $i$ and $Y = z \cdot X$. 
In the case of an empty stack,the "$X$-based annotation" of $A$ is $(A,X)$. 

\myparagraph{Correctness of the construction.} Having defined $\Ggb$, we can
prove that it serves its purpose:
\begin{restatable}{lemma}{lemSoundness}\label{lem-soundness}
	$\Gg$ and $\Ggb$ have the same language, and $\Ggb$ is "productive@@grammar".
\end{restatable}
\whenFull{The mostly straightforward proof can be found in \cref{app:soundness-correctness}.}{}

\begin{remark}\label{rem-noblackbox}
Although we are able to turn any "indexed grammar" into a "productive@@grammar" one with the same language, this transformation incurs an exponential blow-up in the size. Therefore, in order to obtain tight complexity bounds, we do not simply assume productiveness of the input grammar. Rather, we work with "annotated version"s explicitly when needed.
\end{remark}

\section{The stack monoid}\label{sec:monoid}

\myparagraph{Overall goal.}
Another key aspect of our construction is the "stack monoid"---a finite monoid in which we evaluate stack contents.
Essentially, we map a stack $\zb \in \Ib^*$ to a monoid element that encodes\footnote{Notice that there is a dissymmetry between push and pop here (in fact, throughout the construction). This is because a stack symbol is only pushed once, but may be popped on multiple branches.}
\begin{itemize}
	\item the non-terminals from and to which this content was pushed (which are unique thanks to Remark~\ref{rmk:label-push}), and
	
	\item for each pair $A,B$ of non-terminals, whether we can derive a "sentential form" containing $B$ from $A[\zb]$ (i.e., whether we can obtain a $B$ by popping $\zb$ from $A$). 	
\end{itemize}
The purpose of this encoding is that based on this information, we will be able to simplify (or expand) stack contents during derivations, \emph{while preserving the "downward closure"}.
For example, if we have a derivation from $A[z]$ to a "sentential form" $uBv$, then we could just erase $u$ and $v$, thus turning $A[z]$ into $B$.
Even better, if we have a derivation from $A[z]$ to $A$ then, roughly speaking, this means we can turn any term $A[zz']$ into $A[z']$.

Observations like this, based on the stack monoid, will be used in \cref{sec:pump-skip} to define more involved stack manipulations. In \cref{sec:summaries}, these will allow us to reduce, roughly speaking, every stack to one of doubly exponentially many.

\myparagraph{Semigroup terminology.} \AP We begin with some terminology.
A semigroup $(\mathbf{S}, \cdot)$ is a set equipped with an associative product.
We will often identify a semigroup with its set of elements and denote it simply as $\mathbf{S}$, when the product operation is clear.
A monoid $(\mathbf{M}, \cdot, \mathbf{1}_{\mathbf{M}})$ is a semigroup with a neutral element $\mathbf{1}_{\mathbf{M}}$.

\AP For each monoid $(\mathbf{M}, \cdot)$, we define $\intro*\evalmorph{\mathbf{M}} : \mathbf{M}^* \to \mathbf{M}$ to be its ""evaluation morphism"", which maps a sequence of elements of $\mathbf{M}$ to its product, with the empty sequence being mapped to $\mathbf{1}_{\mathbf{M}}$.
An element $x \in \mathbf{M}$ is ""idempotent"" if $x \cdot x = x$.
The set of "idempotent" elements of $\mathbf{M}$ is denoted $\intro*\idempotents{\mathbf{M}}$.

In all that follows we fix an "indexed grammar" $\Gg =(N,T,I,P,S)$. We will mainly work with its "annotated version" $\Ggb$. Note that we cannot use the annotation construction as a black box and simply assume that the input grammar is "productive@@grammar" (see Remark \ref{rem-noblackbox}).

\myparagraph{Formal definition.} Let us now start by describing formally the monoid at hand. It is a bit more involved than what is described above since we have to account for the productiveness issues explained in the previous section. 
\AP We will use the boolean matrix monoid over non-terminals of $\Gg$: This monoid is defined as the set of matrices $\matrixmonoid$, with $\mathbb{B}=\set{\top, \bot}$. The product is the usual product of matrices over the Boolean semiring $(\mathbb{B}, \lor, \land)$, and the neutral element is the matrix with $\top$ on the diagonal and $\bot$ everywhere else.
Given matrices $M_1, M_2$, we write $M_1 M_2$ for their product.

\AP We also define, for all $X \subseteq N$, the ""reachability relation"" $\intro*\Reach{X}$ over $N$.
It relates two non-terminals if from the first one we can produce a "sentential form" containing the second one, with empty stacks. Formally,

\begin{center}
	$A \Reach{X} B$ if and only if there is a derivation $(A,X) ~\derivesindexedstar_{\Ggb}~ u (B, X) v$ with $u,v \in \sentforms$
\end{center}
(note that due to the "productiveness@@grammar" property, requiring stack emptiness is only important for $(B,X)$).
\AP Let $Q$ be the set of tuples $(B,Y, M, A, X)$ with  $A, B \in N$ non-terminals, $X, Y \subseteq N$ sets of non-terminals, and $M \in \matrixmonoid$. Define the ""stack monoid"" as the monoid $(\intro*\prodmonoid, \cdot, \neutral)$ whose elements are $Q \cup \{\neutral\} \cup \{\zeroPM\}$, where $\zeroPM$ satisfies $\zeroPM \cdot x = x \cdot \zeroPM =  \zeroPM$ for all $x \in \prodmonoid$, and whose product operation is defined as follows:

\[ 
(B_2, Y_2, M_2, A_2, X_2) \cdot (B_1, Y_1, M_1, A_1, X_1) =
\begin{cases}
	& \begin{split}(B_2, Y_2, M_1 M_2, A_1 , X_1) \text{ if } &X_2 = Y_1 \\ \text{ and } &B_1 \Reach{X_2} A_2 \end{split}\\
	& \zeroPM \text{ otherwise.}
\end{cases}
\]

\AP Let $\intro*\alphaG, \intro*\betaG$ be the functions described in Remark~\ref{rmk:label-push} for $\Gg$.
\AP We define a morphism $\intro*\morphism : \Ib^+ \to {\prodmonoid}$ as follows. For each letter $(f, X) \in I$, set 
$$\morphism(f,X) = (\betaG(f),f \cdot X, M_{f,X}, \alphaG(f),X),$$ where for all $A,B \in N$,
 	$M_{f,X}(A,B) = \top$ if and only if there exist $u,v \in (X\cup T)^*$ such that $A[f] \derivesindexedstar_{\Gg} u B v$.
 	Note that computing this matrix easily reduces to an emptiness check for an indexed grammar, which can be done in exponential time.

\myparagraph{Feasibility.}
\AP Let us mention some basic properties of stacks that are encoded in their image in $\prodmonoid$. The first concerns whether a given stack content can be pushed:
We say that a non-empty stack content (in either $I^*$ or $\Ib^*$) is ""feasible@@stack"" if, starting from some non-terminal, it can be pushed onto the stack of some non-terminal. In the case of an annotated stack content $\zb = (f_n, X_n)\cdots (f_1, X_1)  \in \Ib^*$, this is equivalent to the existence of a derivation $$(\alphaG(f_1), X_1) ~\derivesindexedstar_{\Ggb}~ u(\betaG(f_n), f_n\cdot X_n)[\zb]v$$ with $u,v \in \sentforms$.
The following lemma says that the "stack monoid" distinguishes the "feasible@@stack" stack contents in $\Ib^*$.

\begin{restatable}{lemma}{lemRightMorphism}
	\label{lem:right-morphism}
	Let  $\zb = (f_n, X_n)\cdots (f_1, X_1)  \in \Ib^*$ be a stack content. The following are equivalent:
	\begin{enumerate}
		\item\label{one} $\zb$ is "feasible@@stack"
		\item\label{two} $\morphism(\zb) \neq \zeroPM$
		\item\label{three} for all $i>1$, $X_{i} = f_{i} \cdot X_{i-1}$ and $\betaG(f_{i-1}) \Reach{X_i} \alphaG(f_i)$.
	\end{enumerate}
\end{restatable}

\whenFull{The proof is given in Appendix~\ref{app:right-morphism}.}{} 
Note that Lemma \ref{lem:right-morphism} shows that all "feasible@@stack" stack contents in $\Ib^*$ can be written as $\annotate{z}{X}$ for some "feasible@@stack" $z\in I^*$ and $X\subseteq N$ (hence, we sometimes assume that a stack content is of this form).

\AP Let $(A,X) \in \Nb$ and $\zb = (f_n, X_n)\cdots (f_1, X_1)  \in \Ib^*$, we say that the term $(A,X) [\zb]$ is ""feasible@@term"" if $\zb$ is "feasible@@stack", $X = f_n \cdot X_n$ and $\betaG(f_n) \Reach{X} A$.

\myparagraph{Pushing between specific non-terminals.}
We will now see that $\prodmonoid$ encodes which non-terminals allow a stack content to be pushed, and which non-terminals can result.
More formally, $\morphism(\zb)$ encodes all pairs of non-terminals $(C,X),(D,Y)\in \Nb$ such that $(C,X)$ can derive a sentential form $u (D,Y)[\zb]v$.

\begin{restatable}{lemma}{lemRightMonoidA}
	\label{lem:right-monoid-A}
	Let $z\in I^+$ a non-empty stack content, let $X \subseteq N$ and let $(B, Y,  M, A, X) =\morphism(\annotate{z}{X})$. Then for all $C \in X$ and $D \in Y$, the following are equivalent:
	\begin{itemize}
		\item $C \Reach{X} A$ and $B \Reach{Y} D$
		\item there exist $u,v \in \sentforms$ such that $(C,X) \derivesindexedstar_{\Ggb} u (D,Y)[\annotate{z}{X}] v$.
	\end{itemize}
\end{restatable}

\whenFull{This lemma is proven in Appendix~\ref{app:right-monoid-A}.}{}

\myparagraph{Popping between specific non-terminals.}
Finally, our monoid even encodes popping behavior. Specifically, the matrix $M$
inside $\varphi(\annotate{z}{X})$ tells us for which non-terminals $D$ and $C$,
it is possible to pop $z$ from $D$ to $C$:
\begin{restatable}{lemma}{lemRightMonoidM}
	\label{lem:right-monoid-M}
	Let $z\in I^+$ a non-empty stack content, let $X \subseteq N$ and let $(B, Y,  M, A, X) =\morphism(\annotate{z}{X})$. The following are equivalent:
	
	\begin{itemize}
		\item $M(D,C) = \top$
		\item $C \in X$, $D \in Y$ and there exist $u,v \in (X \cup T)^*$ such that $$D[z] \derivesindexedstar_{{\Gg}} u C v$$
		\item $C \in X$, $D \in Y$ and there exist $u,v \in \sentforms_{\Gg}$ such that $$(D,Y)[\annotate{z}{X}] \derivesindexedstar_{{\Ggb}} u (C,X) v.$$
	\end{itemize}
\end{restatable}

\whenFull{
This lemma is proven in Appendix~\ref{app:right-monoid-M}.}{}

\section{Pumping and skipping}
\label{sec:pump-skip}

We introduce two new derivation rules on the "terms" of $\Ggb$, and show that they preserve the downward-closure of the resulting language.
Essentially, we show that,  under certain conditions, sequences of more than $2|N|$ contiguous infixes mapping to the same idempotent can be extended and reduced.

This will let us abstract stack contents by forgetting everything but the first and last $N$ elements in such sequences of infixes mapping to the same idempotent.
By adapting a recent construction by Gimbert, Mascle and Totzke~\cite{GimbertMT25arxiv}, we will show that this allows us to obtain "summaries" of stack contents, of size bounded by an exponential function in the size of $\Gg$.

\myparagraph{Pump and skip derivation steps.} \AP Let us formally define the two rules.
The ""pump rule"" is defined as follows:
\[(B,X) [\zb] ~\intro*\rulepump~ (B,X) [z_e \zb]\] for all $(A,X) \in \Nb$, $e = (B,X,M,A,X) \in \idempotents{\prodmonoid}\setminus\set{\zeroPM,\neutral}$, $z_e\in \Ib^+$ with $\morphism(z_e) =e$, and $\zb\in \Ib^*$.

The ""skip rule"" is first defined on stack contents:
\[\zb' u_1 \cdots u_N z_e u_1 \cdots u_N \zb ~\intro*\ruleskip~  \zb'u_1 \cdots u_N \zb\] for all  $\zb',u_1, \dots, u_N, z_e, \zb \in \Ib^*$ and $e \in \idempotents{\prodmonoid}\setminus\set{\zeroPM}$ such that $\morphism(u_1) = \dots = \morphism(u_N) = \morphism(z_e) = e$ (where the $u_i$ are in one-to-one correspondence with $N$; we use the subscript $N$ instead of $|N|$ for convenience).
We then extend it to "terms" naturally: $(A,X) [\zb] \ruleskip (A,X)[\zb']$ whenever $\zb\ruleskip \zb'$. 
Observe that the "skip rule" erases symbols \emph{(potentially deep) inside the stack}, not just at the top. Nevertheless, we will see that allowing the skip rule does not extend the language beyond its "downward closure".

\AP Both the "pump rule" and "skip rules" are extended to "sentential forms" as expected: $u \rulepump u'$ if $u'$ is obtained by applying $\rulepump$ to a term in $u$, and the same goes for $\ruleskip$.
We write $\derivesindexedps[\Ggb]$ for the union of the three relations $\derivesindexed[\Ggb]$, $\rulepump$ and $\ruleskip$, and $\derivesindexedstarps[\Ggb]$ for its reflexive transitive closure. We may thus define
\[ \intro*\langPS{\Ggb}=\set{w \in T^* \mid (S,\Useful) ~\derivesindexedstarps~[\Ggb] w}. \]

	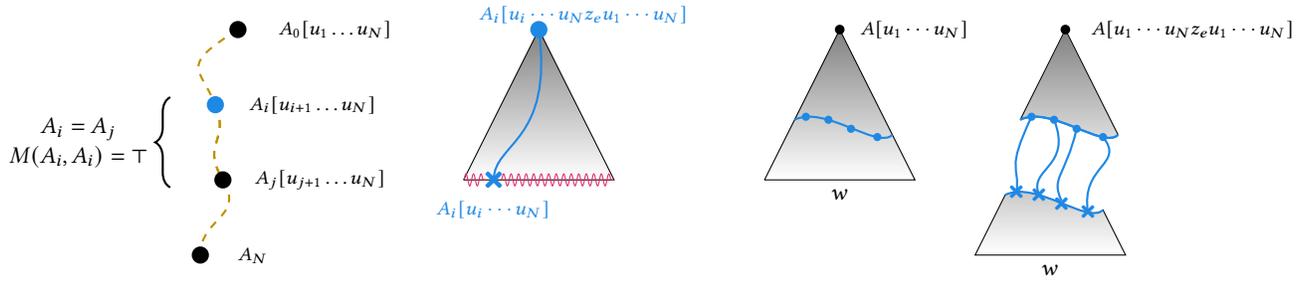
\begin{figure*}
	
	\begin{tikzpicture}[auto,>= triangle 45, scale=0.8]

		\draw[thick, color=yel1, dashed] (0.6,-1).. controls (0.9,-1.5) and (0.3,-1.5)
		..(0.3,-2);
		\draw[thick, color=yel1, dashed] (0.8,1).. controls (0,0.5) and (0.3,0.5)
		..(0.5,0);
		\draw[thick, dashed, color=yel1] (0.5,0).. controls (0.7,-0.5) and (0.3,-0.5)
		..(0.6,-1);
		
		\draw[thick, dashed, color=yel1] (0.5,0).. controls (0.7,-0.5) and (0.3,-0.5)
		..(0.6,-1);

		\draw[thick, color=black, fill=black] (0.8,1) circle (1mm);
		\node (A) at (2.3,1) {\footnotesize $A_0[u_1\dots u_N]$};
		
		\draw[thick, color=blue1, fill=blue1] (0.5,0) circle (1mm);
		\node (A) at (2,0) {\footnotesize $A_i[u_{i+1}\dots u_N]$};

		\path[draw, thick, color=black,decorate,decoration={brace, amplitude=2mm}]  (0.1,-1.1) -- (0.1,0.1)  node[pos=0.5, align=center, left=5pt]{\small$A_i = A_j$\\ \small $M(A_i, A_i)= \top$};

		\draw[thick, color=black, fill=black] (0.6,-1) circle (1mm);
		\node (A) at (2.1,-1) {\footnotesize $A_j[u_{j+1}\dots u_N]$};
		
		\draw[thick, color=black, fill=black] (0.3,-2) circle (1mm);
		\node (A) at (1.2,-2) {\footnotesize $A_N$};

		\draw[color=black, shade] (4,-1) -- (5,1) -- (6,-1) -- (4,-1);
		\draw[thick, color=blue1, fill=blue1] (5,1) circle (1mm);
		\node[color=blue1] (A) at (4.4,-1) {\textbf{\faTimes}};
		
		\node[color=blue1] (O) at (5.6,1.2) {\footnotesize $A_i[u_i \cdots u_Nz_e u_1 \cdots u_N]$};
		\node[color=blue1] (O) at (4.4,-1.4) {\footnotesize $A_i[u_i \cdots u_N]$};
		
		\draw[thick, color=blue1] (5,1).. controls (5.2,-0.5) and (4.4,-0.5)..(4.4,-1);
		
		\draw[color=red1,decorate,decoration={snake, amplitude=0.8mm, segment length=1mm}] (4.5,-1) -- (6,-1);
		\draw[color=red1,decorate,decoration={snake, amplitude=0.8mm, segment length=1mm}] (4,-1) -- (4.3,-1);

		\draw[color=black, shade] (7,-1) -- (8,1) -- (9,-1) -- (7,-1);
		\draw[thick, fill=black] (8,1) circle (0.5mm);
		\node[color=black] (O) at (8.9,1.2) {\footnotesize$A[u_1 \cdots u_N]$};
		\node[color=black] (O) at (8,-1.2) {$w$};
		
		\draw[thick, color=blue1] (7.4,-0.2).. controls (7.8,0) and (8.5,-0.6)
		..(8.7,-0.4);
		
		\draw[thick, color=blue1, fill=blue1] (7.55,-0.16) circle (0.4mm);
		\draw[thick, color=blue1, fill=blue1] (7.85,-0.2) circle (0.4mm);
		\draw[thick, color=blue1, fill=blue1] (8.15,-0.32) circle (0.4mm);
		\draw[thick, color=blue1, fill=blue1] (8.5,-0.43) circle (0.4mm);

		\draw[fill=gray!10!white, color=white,shade] (11.8,0)-- (12.8,-2) -- (10.8,-2) -- (11.8,0);
		
		\draw[white, fill=white] (11.2,-1.2) .. controls (11.6,-1) and (12.3,-1.6)..(12.5,-1.4)  -- (11.8,0);
		
		\draw[color=white, shade] (11,-1) -- (12,1) -- (13,-1) -- (11,-1);
		\draw[fill=white, color=white]  (11.4,-0.2) .. controls (11.8,0) and (12.5,-0.6)..(12.7,-0.4) -- (13,-1) -- (11,-1) -- (11.4,-0.2);

		\draw (11.2,-1.2) .. controls (11.6,-1) and (12.3,-1.6)..(12.5,-1.4) -- (12.8,-2) -- (10.8,-2) --  (11.2, -1.2);
		
		\draw[color=black] (12.7,-0.4) -- (12,1) -- (11.4,-0.2);

		\draw[thick, color=blue1] (11.4,-0.2).. controls (11.8,0) and (12.5,-0.6)
		..(12.7,-0.4);
		\draw[thick, color=blue1] (11.2,-1.2).. controls (11.6,-1) and (12.3,-1.6)
		..(12.5,-1.4);

		\draw[thick, color=blue1, fill=blue1] (11.55,-0.16) circle (0.4mm);
		\draw[thick, color=blue1, fill=blue1] (11.85,-0.2) circle (0.4mm);
		\draw[thick, color=blue1, fill=blue1] (12.15,-0.32) circle (0.4mm);
		\draw[thick, color=blue1, fill=blue1] (12.5,-0.43) circle (0.4mm);
		
		\draw[thick, color=blue1] (11.4,-0.2).. controls (11.8,0) and (12.5,-0.6)
		..(12.7,-0.4);

		\node[color=blue1] (A) at (11.35,-1.16) {\footnotesize \textbf{\faTimes}};
		\node[color=blue1] (A) at (11.65,-1.2) {\footnotesize \textbf{\faTimes}};
		\node[color=blue1] (A) at (11.95,-1.32) {\footnotesize \textbf{\faTimes}};
		\node[color=blue1] (A) at (12.3,-1.43) {\footnotesize \textbf{\faTimes}};

		\draw[thick, color=blue1] (11.55,-0.16).. controls (11.2,-0.9) and (11.4,-0.8)
		..(11.35,-1.16);
		\draw[thick, color=blue1] (11.85,-0.2).. controls (12.1,-0.84) and (11.5,-0.84)
		..(11.65,-1.2);
		\draw[thick, color=blue1] (12.15,-0.32).. controls (12,-0.96) and (11.8,-0.96)
		..(11.95,-1.28);
		\draw[thick, color=blue1] (12.5,-0.43).. controls (12.8,-0.8) and (12.2,-0.8)
		..(12.3,-1.43);

		\draw[thick, fill=black] (12,1) circle (0.5mm);
		\node[color=black] (O) at (12.5,1.2) {\footnotesize$A[u_1 \cdots u_Nz_e u_1 \cdots u_N]$};
		\node[color=black] (O) at (11.8,-2.2) {$w$};
		
	\end{tikzpicture}	
	\caption{The idea behind the "skip rule" is that if we have a derivation from some non-terminal $A$ in which we pop $N$ consecutive infixes $u_1 \dots u_N$ mapping to the same idempotent $e = (B,X,M, A,X)$, then along every branch we must have a node of the form $A_i[u_{i+1} \cdots u_N]$ with $M(A_i, A_i) = \top$.
		This means that for any word mapping to $e$, there is a derivation from $A_i$ popping this word, and with the resulting "sentential form" containing $A_i$. The rest of the sentential form can be erased since we are interested in the "downward closure".
		This non-terminal can be used to ``skip'' the infix $u_{i+1} \cdots u_N z_e u_1 \cdots u_i$, which maps to $e$. Hence, we can turn a derivation from $A[u_1 \cdots u_N]$ into one from $A[u_1 \cdots u_N z_e u_1 \cdots u_N]$ by popping the right infix along every branch.}
	\label{fig:skip}
\end{figure*}

\myparagraph{Pump and skip are harmless.}
We now show that these additional rules do not extend our language beyond its "downward closure":
\begin{restatable}{proposition}{lemPumpSkip}
	\label{prop:eliminate-pump-skip}
	$\langPS{\Ggb} \subseteq \dcl{\langIG{\Ggb}}$ 
\end{restatable}

The full proof is presented in \whenFull{Appendix~\ref{app:eliminate-pump-skip}}{the full version~\cite{MandelMZ26arxiv}}, we sketch the ideas here. 
For the "pump rule" this is not hard. It follows from the definition of the product of $\prodmonoid$ that if $e = (B,X,M,A,X)$ is an idempotent, then $B \Reach{X} A$, i.e., there must be a derivation from $(B,X)$ to some $u (A,X) v$. 
By erasing $u$ and $v$ (thanks to downward closure and "productiveness@@grammar"), from $(B,X)[\zb]$ we can reach $(A,X)[\zb]$, whence we can reach $(B,X)[z_e \zb]$ (where $\phi(z_e) = e$) thanks to Lemma~\ref{lem:right-monoid-A}. 
Hence, any terminal word derived from $(B,X)[z_e \zb]$ is a subword of a terminal word derived from $(B,X)[\zb]$ (note that the "productive@@grammar"ness of $\Ggb$ is essential here).

\myparagraph{Eliminating skip.} Replacing applications of the "skip rule" requires more work. In an "indexed grammar", a symbol is pushed once on the stack, but may be popped on multiple branches of a derivation.
To tackle this issue, we make the following observations, illustrated by Figure~\ref{fig:skip}. When a sequence $u \in \Ib^*$ with $\morphism(u) = (B,Y, M, A, X)$
is popped along a branch starting with a non-terminal $(D,Y)$, the resulting non-terminal $(C,X)$ (after popping $u$) must satisfy $M(D,C) = \top$ (see Lemma~\ref{lem:right-monoid-M}). 

Now suppose we are popping a sequence of infixes $u_1 \cdots u_N$, with $\morphism(u_1) = \dots = \morphism(u_N) = e =(B,X,M,A,X) \in \idempotents{\prodmonoid}$, along a branch starting with the non-terminal $(A_0,X)$. Consider, for $i = 1,\ldots,|N|$, the non-terminal $(A_i, X)$ obtained along that branch right after popping $u_1 \cdots u_i$. 
By Lemma~\ref{lem:right-monoid-M}, it follows that $M(A_i, A_{i+1}) = \top$ for $i=0,\ldots,|N|$.
Since $e$ is idempotent, $M M = M$, and we can then apply the following elementary lemma.

\begin{lemma}
	\label{lem:idempotent-property}
	Let $M \in \matrixmonoid$ a matrix.
	If $M M = M$ and we have terms $A_0, \dots, A_{N} \in N$ such that $M(A_i, A_{i+1}) = \top$ for all $i = 1,\ldots,|N|-1$, then there exists $i$ with $M(A_i, A_i) = \top$.
\end{lemma}

\begin{proof}
	By the pigeonhole principle, there exist $i<j$ such that $A_i = A_j$.
	Since $M$ is "idempotent", $M^{j-i} = M$. Hence, we have
	\[M(A_i, A_i) = M^{j-i}(A_i,A_j) = \bigwedge_{l=i}^{j-1} M(A_l,A_{l+1}) = \top.\qedhere\]
\end{proof}

Thus, one of the $A_i$ is such that $M(A_i, A_i) = \top$ (note that $i$ may depend on the branch).
This means that for all $\zb \in \Ib^*$ with $\morphism(\zb) = e$, we have a derivation from $(A_i,X)[\zb]$ to $u (A_i,X) u'$ for some $u, u'\in \sentforms$ which can be erased since we are only interested in the downward closure (and because of the "productive@@grammar"ness property).
Hence, $(A_i,X)$ can ``erase'' an arbitrary sequence of infixes mapping to $e$ at the top of its stack. The "skip rule" is obtained, roughly speaking, by erasing the sequence $u_{i+1}\cdots u_N z_e u_1\cdots u_i$ at the appropriate $(A_i,X)$ in each branch of a derivation. Since we are sure to encounter, when popping a sequence of $|N|$ infixes mapping to $e$, such a non-terminal on every branch of the derivation (i.e. which can pop such infixes at will), we are able to eliminate all applications of the "skip rule" \whenFull{(see \cref{app:eliminate-pump-skip} for the full proof)}{}.

\section{Summarizing stack contents}
\label{sec:summaries}

We have shown that adding the "pump rule" and "skip rule" to our "annotated" "indexed grammar" does not change its "downward closure". These rules provide flexibility on stack infixes that map to idempotents: the "pump rule" can extend them, and the "skip rule" can reduce them.

We wish to compress stack contents into bounded summaries, where such sequences of infixes are abstracted away, by remembering only the first and last $|N|$ of them.
A natural candidate for this task is Simon's factorization forest theorem~\cite{Simon90}.
It gives us a way to evaluate a word in a finite monoid via a tree of bounded height, using binary products and arbitrary iterations of idempotent elements.
If we only care about the first and last $|N|$ elements in a sequence of idempotent infixes, we can cut out from this tree all the nodes corresponding to the remaining idempotent infixes.

Two problems remain: First, Simon's theorem gives a linear bound for the height of the tree in the monoid size. In our case this would yield summaries of doubly exponential size, and  a quadruply exponential upper bound on an NFA for the "downward closure", while our lower bound is only triply exponential.
To solve this, we dive deeper into the structure of the monoid, utilizing results by Jecker~\cite{Jecker21} which let us cut our monoid into polynomially many layers. Within each layer, we decompose the word by reading it from right to left and compressing sequences of idempotent infixes.

The other problem is that to simulate the "indexed grammar" with a "context-free grammar", we use a version of the push operation on the compressed words: from the compressed version of a word $z$ and a letter $x$, we need to be able to compute the compressed version of $x z$.
This is not a property of Simon's theorem, at least not in its original formulation.
One way to circumvent this problem was through the introduction of \emph{forward Ramsey splits}~\cite{Colcombet-ICALP2007}, a weaker version of factorizations which still detects sequences of idempotent infixes.
In fact, Jecker's results have been applied to improve computational bounds on those splits~\cite{LopezMT25}. It may be interesting to see if one can obtain a form of "summaries" from such splits, by losing information to obtain a bounded object.
Here, however, we do not rely on these, but provide an elementary construction computed deterministically by reading the word right-to-left.

We rely on a recent construction of such summaries presented in~\cite{GimbertMT25arxiv}. Our exposition is self-contained, however, since we use different notation and our summaries need to be constructed in a slightly different manner. Specifically, theirs need to remember only the first and last infixes, while we need the last $|N|$, and theirs are constructed as trees of bounded height, bottom-up, while we need to build ours from right to left along the stack content.

\myparagraph{Green's relations.} \AP We begin with some notions from semigroup theory.
Let $(\mathbf{M}, \cdot, \neutralM)$ be a finite monoid.
We define the usual Green relations $\Jgreen, \Lgreen,\Rgreen, \Hgreen$ on $\mathbf{M}$, starting with the following quasi-orders:
\begin{itemize}
	\item $ x \intro*\Jleq y$ if there exist $a,b \in \mathbf{M}$ such that $x = a \cdot y \cdot b$
	\item $ x \intro*\Lleq y$ if there exist $a \in \mathbf{M}$ such that $x = a \cdot y$
	\item $ x \intro*\Rleq y$ if there exist $b \in \mathbf{M}$ such that $x = y \cdot b$
	\item $ x \intro*\Hleq y$ if $x \Lleq y$ and $x \Rleq y$	
\end{itemize}

The relations $\intro*\Jgreen, \intro*\Lgreen, \intro*\Rgreen, \intro*\Hgreen$ are the equivalence relations induced by those quasi-orders: 
for each $\mathcal{X} \in \set{\mathcal{J}, \mathcal{L}, \mathcal{R}, \mathcal{H}}$ we define $\mathcal{X} = \mathord{\leq_{\mathcal{X}}} \cap \mathord{\geq_{\mathcal{X}}}$, as well as $\mathord{<_{\mathcal{X}}} = \mathord{\leq_{\mathcal{X}}} \setminus \mathord{\geq_{\mathcal{X}}}$.
We do not include the monoid $\mathbf{M}$ in the notation since it will always be clear from the context.

\myparagraph{$\mathcal{J}$-length and $\mathcal{J}$-depth.} 
\AP The ""regular $\mathcal{J}$-length""
\footnote{Called regular $\mathcal{D}$-length in~\cite{Jecker21}. 
	For finite monoids, $\mathcal{D} = \mathcal{J}$, and we use $\mathcal{J}$ here since it is more common. Furthermore, it is defined differently in~\cite{Jecker21}, but a proof that both definitions are equivalent can be found in the long version of that paper~\cite[Appendix B]{Jecker21arxiv}.} of $\mathbf{M}$, denoted $\intro*\Jlength{\mathbf{M}}$, is defined as 
\[\Jlength{\mathbf{M}} = \sup\set{k \in \NN \mid \exists e_1 >_{\mathcal{J}} \cdots >_{\mathcal{J}} e_k \in \idempotents{\mathbf{M}}}. \]

\begin{definition}
	The ""regular $\mathcal{J}$-depth"" (or just "depth") of an element $x \in \mathbf{M}$ is the maximal number $d$ such that there are $e_1,\ldots, e_d \in \idempotents{\mathbf{M}}$ with $e_1 >_{\mathcal{J}} \dots >_{\mathcal{J}} e_d >_{\mathcal{J}} x$.
	
	We write $\intro*\Jheight{x}$ for the "regular $\Jgreen$-depth" of an element $x \in \mathbf{M}$.
	We extend this notation to sequences of elements:
	For $u \in \mathbf{M}^*$, we write $\Jheight{u}$ instead of $\Jheight{\evalmorph{\mathbf{M}}(u)}$.
\end{definition}

\begin{remark}\label{rem:product-depth}
	Note that for all $x,y \in \mathbf{M}$, one has $\Jheight{x \cdot y} \geq \max(\Jheight{x}, \Jheight{y})$. 
	This is because of the definition of $\Jgreen$, since $x \cdot y \Jleq x$ and $x \cdot y \Jleq y$.
\end{remark}

\begin{remark}
	Observe that in the case of the "stack monoid" $\prodmonoid$, we have $x<_{\Jgreen} \neutral$ for every $x \in \prodmonoid \setminus \set{\neutral}$, because if $y,z\in\prodmonoid\setminus\set{\neutral}$, then $y\cdot z\ne\neutral$.
	As a consequence, $\Jheight{\neutral} = 0$, whereas
	$\Jheight{x}$ is in $[1, \Jlength{\prodmonoid}]$ for all $x \in \prodmonoid \setminus \set{\neutral}$.
	The upper bound follows from the definition of $\Jlength{\prodmonoid}$.
	Positivity of $\Jheight{x}$ is due to $x <_{\Jgreen} \neutral$.  
\end{remark}

\myparagraph{Summaries.}
We now define the main object of this section, \emph{summaries}. As mentioned above, they can be viewed as compressions of stack words (with some information loss). Syntactically, these are sequences of (sequences of (...)) sequences, where the nesting depth depends on the depth of $\morphism(w)$, where $w$ is the stack word to be compressed. These sequences contain letters from $\Ib$, but also a special letter $e^+$ for each "idempotent" $e\in\idempotents{\prodmonoid}\setminus\set{\neutral}$. Intuitively, $e^+$ represents a sequence of infixes that evaluate to $e$.

Summaries will have a well-defined image under $\morphism$, and the summary of $w\in\Ib^*$ will agree with $w$ under $\morphism$. To this end, we set $\morphism(e^+) = e$ for all such $e$. 
We also extend $\morphism$ to $(\Ib^*)^*$ by defining the image of a sequence of words $z_1 \dots z_n \in (\overline{I}^*)^*$ as the image of their concatenation, and so on for sequences of (sequences of (...)) sequences.
To simplify notation, given a sequence of stack symbols $z \in \Ib^*$, we also write $\Jheight{z}$ for $\Jheight{\morphism(z)}$, and we extend this notation to sequences of sequences of (sequences of (...)) sequences naturally.

Formally, a $0$-summary is just the empty word, and a $(d+1)$-summary is a sequence of (i)~elements of $\Ib \cup \set{e^+ \mid e\in \idempotents{\prodmonoid}}$ and of (ii)~$d$-summaries.
As explained above, $\morphism(\sigma)$ and $\Jheight{\sigma}$ are well-defined for all "summaries" $\sigma$.
\begin{definition}
	A ""$0$-summary"" is simply the empty word $\epsilon$.
	Let $d \in [1, \Jlength{\prodmonoid}]$.
	\begin{itemize}
		\item A ""$d$-atom"" is a word $(f,X) \sigma$ with $(f,X) \in \Ib$ and $\sigma$ a "$d'$-summary" for some $d' < d$, such that $\Jheight{(f,X) \sigma} =d$.
		
		\item A ""$d$-block"" is a sequence of the form $u_1 \cdots u_N e^+ v_1 \cdots v_N w$ where $u_1, \dots, u_N, v_1, \dots, v_N$ and $w$ are sequences of "$d$-atoms" and $e \in \idempotents{\prodmonoid}\setminus \set{\neutral}$, such that $\morphism(u_i)=\morphism(v_i) = e$ for all $i$ and $\Jheight{u_1 \cdots u_N e^+ v_1 \cdots v_N w} = d$.
		
		\item A ""$d$-summary"" is a sequence of the form $\sigma' u B_1 \dots B_k$ where $\sigma'$ is a "$d'$-summary" for some $d'<d$, $u$ is a word of "$d$-atoms", and $B_1, \dots, B_k$ are "$d$-blocks" with $\Jheight{\sigma' u B_1 \dots B_k} = d$. 
	\end{itemize}
	A ""summary"" is a "$d$-summary" for some $d$. 
	Observe that by definition, a "$d$-summary" $\sigma$ has $\Jheight{\sigma}=d$.
	The set of "summaries" is denoted $\intro*\Summaries$.
\end{definition}

	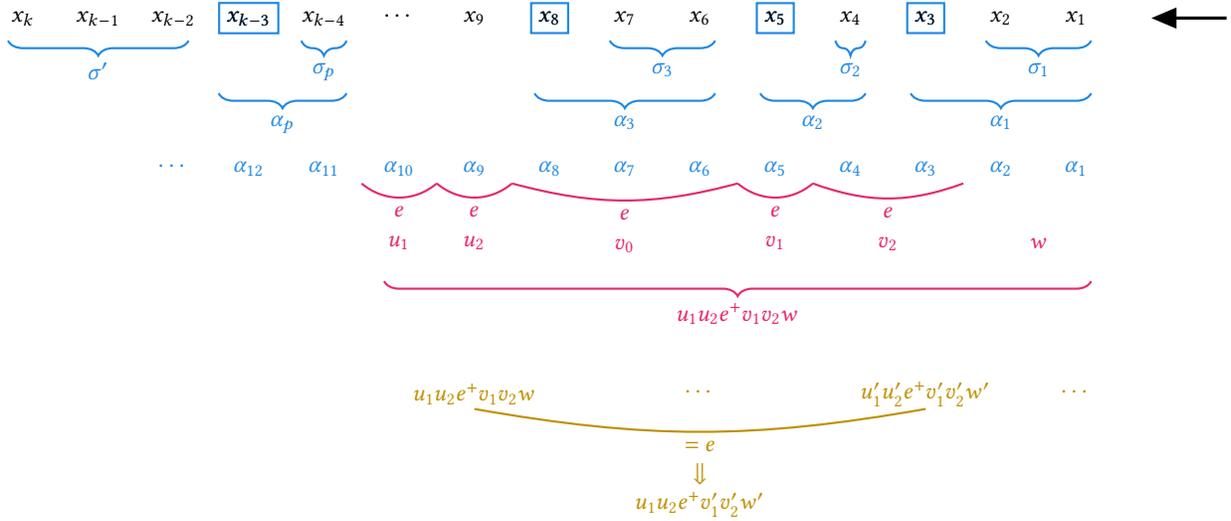
\begin{figure*}[ht]
	
	\begin{tikzpicture}[auto,>= triangle
		45, scale=0.85]

		\path[draw, thick, color=blue1,decorate,decoration={brace, amplitude=2mm}] (-0.8,-0.3) -- (-2.2,-0.3) node[pos=0.5,below=5pt]{$\sigma_1$};
		\path[draw, thick, color=blue1,decorate,decoration={brace, amplitude=2mm}]  (-3.8,-0.3) -- (-4.2,-0.3) node[pos=0.5,below=5pt]{$\sigma_2$};
		\path[draw, thick, color=blue1,decorate,decoration={brace, amplitude=2mm}]  (-5.8,-0.3) -- (-7.2,-0.3) node[pos=0.5,below=5pt]{$\sigma_3$};
		\path[draw, thick, color=blue1,decorate,decoration={brace, amplitude=2mm}] (-10.7,-0.3)-- (-11.3,-0.3) node[pos=0.5,below=5pt]{$\sigma_p$};
		\path[draw, thick, color=blue1,decorate,decoration={brace, amplitude=2mm}]  (-12.8,-0.3) -- (-15.2,-0.3)  node[pos=0.5,below=5pt]{$\sigma'$};
		
		\path[draw, thick, color=blue1,decorate,decoration={brace, amplitude=2mm}] (-0.8,-1) -- (-3.2,-1) node[pos=0.5,below=5pt]{$\alpha_1$};
		\path[draw, thick, color=blue1,decorate,decoration={brace, amplitude=2mm}] (-3.8,-1) -- (-5.2,-1) node[pos=0.5,below=5pt]{$\alpha_2$};
		\path[draw, thick, color=blue1,decorate,decoration={brace, amplitude=2mm}] (-5.8,-1) -- (-8.2,-1) node[pos=0.5,below=5pt]{$\alpha_3$};
		
		\path[draw, thick, color=blue1,decorate,decoration={brace, amplitude=2mm}] (-10.7,-1) -- (-12.4,-1) node[pos=0.5,below=5pt]{$\alpha_p$};
		
		\node[draw, color=blue1, thick] at (-3,0) {$x_3$};
		\node[draw, color=blue1, thick] at (-5,0) {$x_5$};
		\node[draw, color=blue1, thick] at (-8,0) {$x_8$};
		\node[draw, color=blue1, thick, inner xsep=2pt] at (-12,0) {$x_{k-3}$};

		\foreach \x in {1,...,9}
		{
			\node (x\x) at (-\x,0) {$x_{\x}$};
		}

		\foreach \x in {1,...,4}
		{
			\pgfmathparse{int(15-\x)}
			\node (\x4) at (-\pgfmathresult,0) {$x_{k-\x}$};
		}

		\node (dots) at (-10,0) {$\cdots$};
		\node (dots) at (-15,0) {$x_k$};

		\draw[ very thick] (1, 0) edge[->] (0, 0);

		\foreach \x in {1,...,12}
		{
			\node[color=blue1] (a\x) at (-\x,-2) {$\alpha_{\x}$};
		}

		\node[color=blue1] (w) at (-13,-2) {$\cdots$};

		\draw[bend left=20, thick, color=red1, align=center] (-2.5,-2.2) edge node[below] {$e$\\$v_2$} (-4.5,-2.2);
		\draw[bend left=40, thick, color=red1, align=center] (-4.5,-2.2) edge node[below] {$e$\\$v_1$} (-5.5,-2.2);
		\draw[bend left=15, thick, color=red1, align=center] (-5.5,-2.2) edge node[below] {$e$\\$v_0$} (-8.5,-2.2);
		\draw[bend left=40, thick, color=red1, align=center] (-8.5,-2.2) edge node[below] {$e$\\$u_2$} (-9.5,-2.2);
		\draw[bend left=40, thick, color=red1, align=center] (-9.5,-2.2) edge node[below] {$e$\\$u_1$} (-10.5,-2.2);
		
		\node[color=red1] (w) at (-1.5,-3) {$w$};

		\path[draw, thick, color=red1,decorate,decoration={brace, amplitude=2mm}] (-0.8,-3.5) -- (-10.2,-3.5) node[pos=0.5,below=5pt]{$u_1 u_2 e^+ v_1 v_2 w$};

		\node[color=yel1] (w) at (-9,-5) {$u_1 u_2 e^+ v_1 v_2 w$};
		\node[color=yel1] (w) at (-3,-5) {$u'_1 u'_2 e^+ v'_1 v'_2 w'$};
		\node[color=yel1] (w) at (-1,-5) {$\cdots$};
		\node[color=yel1] (w) at (-6,-5) {$\cdots$};

		\draw[bend left=10, thick, color=yel1, align=center] (-3,-5.2) edge node[below] {$= e$\\$\Downarrow$ \\$u_1u_2 e^+ v'_1 v'_2 w'$} (-9,-5.2);
	\end{tikzpicture}	
	\caption{A visual of how "summaries" are constructed. Say we are given a word $z$ of "regular $\Jgreen$-depth" $d$.
		We read it from right to left. We cut it into "$d$-atoms" by iteratively taking the smallest suffix of "regular $\Jgreen$-depth" $d$. Its summary consists of its leftmost letter and a "$d'$-summary" for its "tail" of depth $d'<d$.
		In parallel, we read the resulting sequence of atoms $ \cdots \alpha_3 \alpha_2 \alpha_1$ from right to left. Every time we find an infix with a "prefix" made of $2|N|+1$ infixes mapping to the same idempotent, we turn it into a "$d$-block" $u_1 u_2 e^+ v_1 v_2$.
		Finally, whenever we have two blocks corresponding to an idempotent $e$ and such that the infix between their middle parts also evaluates to $e$, we merge them into one "$d$-block".}\label{fig:summary}
\end{figure*}

\myparagraph{Intuition.} Let us give some intuition on the definition of
summaries.  Note that when processing a string from right to left, the "depth" of
the growing suffix increases monotonically (\cref{rem:product-depth}). We read the word from right to left since this is how our stacks are built.
Here, a
"$d$-atom" represents a suffix-minimal sequence of "depth" $d$: It has depth
$d$, but when removing the left-most letter, the remainder has lower "depth".
Hence, that remainder is given as a "$d'$-summary" for some $d' <d$. 
A "$d$-block" can be thought of as a sequence of atoms where some of them (which mapped to $e$ under $\phi$) were (lossily) compressed
into a single letter $e^+$: This compression is only allowed if the
compressed part had been surrounded by $|N|$ sequences of $d$-atoms (on each
side) that also map to $e$, which are still present in
$u_1,\ldots,u_N,v_1,\ldots,v_N$. Finally, a "$d$-summary" consists of
$d$-blocks $B_1,\ldots,B_k$, some remaining $d$-atoms (that did not permit
compression into $d$-blocks), and a lower-depth part represented by a
$d'$-summary.

\myparagraph{Compressing stacks into summaries.} 
Let us describe how a stack $\zb\in\Ib^*$ is compressed into its (unique) "summary".
For a sequence (i.e.~word over a finite alphabet or a "summary") of length $\ge 1$, its ""tail"" is obtained by removing its leftmost element.
Given a word $\zb$ with $\Jheight{\zb} = d$, we proceed in three stages, illustrated in \cref{fig:summary}. 

	\emph{Stage I: Splitting into $d$-atoms.} First we split the word into
	suffix-minimal infixes of "depth" $d$, from right to left,  and a
	residual "prefix" of some "depth" $d'<d$. By suffix-minimality, the "tail" of each
	of these "depth"-$d$ infixes has "depth" $< d$.  We thus turn each of
	these "depth"-$d$ infixes into a "$d$-atom" by computing recursively a
	"summary" for its lower-"depth" "tail".  Afterwards, the residual "prefix"
	of "depth" $d'<d$ is recursively turned into a "summary".

	\emph{Stage II: Compression into blocks.} Then, we look at the sequence
	of "$d$-atoms" and their values in $\prodmonoid$. Going from right to
	left, we look for sequences of $2|N|+1$ infixes $u_1 \cdots u_N v_0 v_1
	\cdots v_N$ all evaluating to the same idempotent. Whenever we find
	such a pattern we turn the current infix (i.e. the pattern with, possibly, a "suffix" $w$) into a "$d$-block" by
	replacing the middle part $v_0$ by $e^+$.  There may be a remainder $u$
	at the left end of the "$d$-atom" sequence with no such pattern; we write it explicitly
	in the "summary".
	
	\emph{Stage III: Merging blocks.} Finally, we look at the sequence of
	resulting "$d$-blocks".  We go again from right to left, this time
	looking for pairs of blocks corresponding to the same idempotent $e$,
	and so that the infix between their $e^+$ markers also evaluates to
	$e$.  This lets us merge them into a single block, obtained by
	abstracting their $e^+$ markers and everything in between in a single
	$e^+$. 

\myparagraph{Updating a summary.} 
A key feature of summaries is that given the "summary" $\sigma$ of a stack $\zb\in\Ib^*$ and an additional letter $\xb\in\Ib$, we can compute the "summary" of $\xb\zb$.
This is needed when simulating pushes. We now describe the relevant operation, 
\[ \push{\_}{\_}: \Ib^* \times \Summaries\rightarrow \Summaries. \]
\AP 
Given a depth $d \in [0,\Jlength{\prodmonoid}]$, an element $(f,X) \in \Ib$ and an "$d$-summary" $\sigma$, we define $\intro*\push{(f,X)}{\sigma}$ as follows.
If $d=0$ then $\sigma = \epsilon$ and $\push{(f,X)}{\sigma} = (f,X)$.
If $d> 0$ then $\sigma$ is of the form  $\sigma' u B_1 \dots B_k$.
\begin{enumerate}[leftmargin=5mm]
	\item If $\Jheight{(f,X) \sigma}> d$ then let $d_+ = \Jheight{(f,X) \sigma}$. We set $\push{(f,X)}{\sigma}$ as the "$d_+$-summary" made of a single  "$d_+$-atom" $(f,X)\sigma$
	
	\item Otherwise $\Jheight{(f,X) \sigma} \leq d$. In fact, since $\Jheight{(f,X)\sigma}$ is at least $\Jheight{\sigma} =d$, we even know $\Jheight{(f,X) \sigma} = d$.
	\begin{enumerate}
		\item If $\Jheight{(f,X) \sigma'} <d$ then\\ $\push{(f,X)}{\sigma} = (\push{(f,X)}{\sigma'}) u B_1 \dots B_k$
		
		\item Otherwise, $\Jheight{(f,X) \sigma'}\ge d$. Then we even know $\Jheight{(f,X)\sigma'}=d$, because $\Jheight{(f,X)\sigma'}$ is at most $\Jheight{(f,X)\sigma} \leq d$. Hence, $(f,X) \sigma'$ is a "$d$-atom". 
		\begin{enumerate}
			\item If the sequence of "$d$-atoms" $((f,X) \sigma') u$ is of the form \\$u_1 \dots u_N v_0 v_1 \dots v_N w$ with $\morphism(v_0)= \morphism(u_i)= \morphism(v_i) = e$ for all $i \geq 1$, for some $e \in \idempotents{\prodmonoid}$, then define the "$d$-block" $B = u_1 \dots u_N e^+ v_1 \dots v_N w$ (note that $e^+$ replaces the middle infix $v_0$).
			\begin{enumerate}
				\item Suppose there is $j$ such that $B_j$ is of the form \[u'_1 \dots u'_Ne^+ v'_1 \dots v'_N w'\] and the infix $v_1 \dots v_N w B_1 \dots B_{j-1} u'_1 \dots u'_N $ also evaluates to $e$ under $\morphism$. 
					In this case, we merge everything up to $B_j$ into a single block: We define $\push{(f,X)}{\sigma}$ to be $B'  B_{j} B_{j+1} \dots B_k$, with $B' = u_1 \dots u_Ne^+ v'_1 \dots v'_N w'$.
				When there are multiple such $j$, we pick the maximal one.
				
				\item Otherwise $\push{(f,X)}{\sigma} = B B_{1} \dots B_k$
			\end{enumerate} 
			\item Otherwise $\push{(f,X)}{\sigma} = u' B_1 \dots B_k$ with $u' = ((f,X) \sigma')u$
		\end{enumerate}
	\end{enumerate}
\end{enumerate}

The definition is extended inductively to $\Ib^*$ by
$$\push{(f,X)\zb}{\sigma} = \push{(f,X)}{\push{\zb}{\sigma}}.$$
Intuitively, the $\push{\_}{\_}$ operation simultaneously executes stages I, II and III described above.
Upon adding a new letter $(f,X)$, it checks whether it yields a new atom (Stage~I). Case (1) happens when a new atom is generated because the depth of the summary is increased by adding $(f,X)$, case (a) when no new atom is generated at depth $d$.
In case (b), we have a new atom $(f,X)\sigma'$. We thus have to check whether this new atom yields a new block (Stage II). If it does, we are in case (i), and we then have to apply Stage III, i.e., see if we can merge this new block with another one. In case (A) we can; in case (B) we cannot.

\myparagraph{Popping from summaries.}
We also define the inverse relation: for $\zb \in \Ib^*$ and $\sigma, \sigma'\in \Summaries$, we have $\sigma \in \intro*\pop{\zb}{\sigma'}$ if and only if $\sigma' = \push{\zb}{\sigma}$. Note that $\push{\_}{\_}$ is a function while $\pop{\_}{\_}$ is a relation.
Intuitively, $\push{\_}{\_}$ makes compressions, by creating and merging blocks, thereby losing information. Meanwhile, $\pop{\_}{\_}$ may revert those compressions, and thus requires non-determinism.

\myparagraph{Bounding summaries.}
\AP The ""size@@summary"" of a "$d$-summary" (resp. "$d$-atom", "$d$-block") is defined recursively as the sum of the "sizes@@summary" of its elements, the "size@@summary" of a single letter being $1$.
Although "summaries" can have a priori unbounded "sizes@@summary", the ones we will use in our context-free grammar will be of the form $\push{\zb}{\epsilon}$ for some $\zb\in\Ib^*$. For such "summaries", we can prove a "size@@summary" bound:

\begin{restatable}{lemma}{ExponentialSummaries}
	\label{lem:exp-summaries}
	For each $\zb \in \Ib^*$, the "size@@summary" of $\push{\zb}{\epsilon}$ is at most exponential in $|N|$.
\end{restatable}
To prove this, we rely on two results of Jecker~\cite{Jecker21}, recalled below. \knowledgenewrobustcmd{\Ramsey}[1]{\cmdkl{\mathcal{R}_{#1}}}

Define the ""Ramsey function"" of $\mathbf{M}$ as follows: for all $k \in \NN$, $\intro*\Ramsey{\mathbf{M}}(k)$ is the minimal $n$ such that for every word of length $n$ there exists $e \in \idempotents{\mathbf{M}}$ and $u_1,\ldots,u_k\in \mathbf{M}^+$ such that the word contains an infix $u_1 \cdots u_k$ with $\evalmorph{\mathbf{M}}(u_i) = e$ for all $i$.

The first one guarantees that in a sequence of exponential length in $\Jlength{\prodmonoid}$, we can find large sequences of consecutive infixes that all map to the same idempotent~\cite[Theorem 1]{Jecker21}:
\begin{theorem}[Jecker]\label{thm:Ismael-main}
	For all finite monoid $\mathbf{M}$, for all $k \in \NN$: 
	\[ \Ramsey{\mathbf{M}}(k) \leq (k |\mathbf{M}|^4)^{\Jlength{\mathbf{M}}}. \]
\end{theorem}

The other one bounds the "regular $\mathcal{J}$-length" of a Boolean matrix monoid by a polynomial in its dimension~\cite[Theorem 2]{Jecker21}:
\begin{theorem}[Jecker]\label{thm:Ismael-height}
	The "regular $\mathcal{J}$-length" of $\matrixmonoid$ is $\Jlength{\matrixmonoid} =\frac{N^2+N+2}{2}$.
\end{theorem}

For Lemma \ref{lem:exp-summaries}, we show that when viewing  when viewing $d$-atoms and "$d'$-summaries", for $d'<d$, as individual letters, the size of a "$d$-summary" is bounded by an exponential in $|N|$. The overall bound on "summaries" then results from the product of those (polynomially many) exponential functions.

\section{Building the context-free grammar}
\label{sec:CFG}

We now construct a "context-free grammar" $\CFG$ that over-approximates $\langIG{\Gg}$, but remains within its "downward closure". It will thus satisfy $\dcl{\langIG{\CFG}}=\dcl{\langIG{\Gg}}$ and enable us to compute an NFA for the latter.

\myparagraph{Definition of the grammar.}
Essentially, $\CFG$ is obtained from $\Ggb$ by replacing stack contents with their "summaries".

\knowledgenewrobustcmd{\FT}{\cmdkl{\textsf{FT}}}

\AP 
We first restrict the set of "summaries" to those that can actually result from a derivation of the indexed grammar.
A summary $\sigma$ is ""feasible@@summary"" if $\sigma=\push{\zb}{\epsilon}$ for some "feasible@@stack" $\zb\in\Ib^*$. Note that this implies $\morphism(\sigma) \neq \zeroPM$ by \Cref{lem:right-morphism}. The non-terminals will be triples $(A,X,\sigma)$ where $(A,X) \in \Nb$, and $\sigma$ is a "feasible@@summary" "summary". We call such triples ""feasible@@triple"". 
The set of "feasible@@triple" triples is denoted $\intro*\FT$.

Define the following context-free grammar $\intro*\CFG$: Its set of non-terminals is $\FT$, with $(S, \Useful, \epsilon)$ the initial one. The set of terminal symbols is $T$. The productions directly mimic the productions in $\Ggb$, except that push and pop productions are simulated by the $\push{\_}{\_}$ and $\pop{\_}{\_}$ relations on summaries:
\begin{itemize}
	\item If $(A,X) \to w \in \Pb$ then $(A,X,\sigma) \to w$, for all "feasible@@summary" $\sigma$.

	\item If $(A,X) \to (B,X)(C,X) \in \Pb$ then\\ $(A,X, \sigma) \to (B,X, \sigma) (C,X, \sigma)$, for all "feasible@@summary" $\sigma$.
	
	\item If $(A,X) \to (B, Y) (f, X) \in \Pb$ then\\ $(A,X, \sigma) \to (B, Y, \push{(f, X)}{\sigma})$, for all "summary" $\sigma$ so that $\sigma$ and $\push{(f,X)}{\sigma}$ are "feasible@@summary".
	
	\item If $(A,Y) (f, X) \to (B, X)  \in \Pb$ then\\ $(A,Y, \sigma) \to (B, X, \sigma')$, whenever $\sigma,\sigma'$ are "feasible@@summary" and $\sigma' \in \pop{(f, X)}{\sigma}$.
\end{itemize}
Abusing terminology slightly, we call production rules of the third and fourth type \emph{pushes} and \emph{pops}, respectively (even though they are standard context-free productions).

\myparagraph{Correctness of the construction.}
The key property of $\CFG$ is that it has the same "downward closure" as $\langIG{\Gg}$.
\begin{theorem}\label{correctness-cfg}
	$\dcl{\langIG{\CFG}}=\dcl{\langIG{\Gg}}$.
\end{theorem}

One of the directions is quite easy: we can simply show that the language of $\CFG$ contains that of $\Ggb$. This is natural as $\CFG$ is built as an over-approximation of $\Ggb$.
We can turn a derivation of $\Ggb$ into one of $\CFG$ by replacing every stack content with its "summary".
\whenFull{The formal proof is presented in Appendix~\ref{app:Ggb-in-CFG}.}{}

\begin{restatable}{proposition}{GgbInCFG}
	\label{prop:Ggb-in-CFG}
	$\langIG{\Ggb} \subseteq \langIG{\CFG}$.  
\end{restatable}

\myparagraph{Simulating $\CFG$ with pumps and skips.}
For the inclusion $\dcl{\langIG{\CFG}}\subseteq\dcl{\langIG{\Gg}}$, we will show that every derivation in $\CFG$ can be simulated by pumps and skips:
\begin{restatable}{proposition}{propCFGinPumpSkip}
	\label{prop:CFG-in-pump-skip}
	$\langIG{\CFG} \subseteq \langPS{\Ggb}$.  
\end{restatable}	
Indeed, by \cref{prop:eliminate-pump-skip}, this implies that
$\langIG{\CFG}\subseteq\dcl{\langIG{\Ggb}}$ and thus
$\dcl{\langIG{\CFG}}\subseteq\dcl{\langIG{\Ggb}}=\dcl{\langIG{\Gg}}$, establishing \cref{correctness-cfg}.

We shall prove \cref{prop:CFG-in-pump-skip} by simulating "summaries" by actual stacks. 
Recall that in the summary $\push{\zb}{\varepsilon}$ that compresses the stack $\zb$, the letters $e^+$ represent a sequence of "$e$-words", where we call $\wb\in\Ib^*$ an ""$e$-word"" if $\morphism(\wb)=e$. The other letters in $\sigma$ are taken directly from $\zb$. Therefore, the unfolding reverses this: It replaces $e^+$ by $e$-words and leaves the other letters unchanged. 

Intuitively, our strategy for replacing $e^+$ is as follows. We replace $e^+$ by a sequence of all "$e$-words" that might be needed to sustain the remainder of the derivation (specifically, all its pop operations). In a particular branch of the derivation, those "$e$-words" that are not needed can always be cleared using the skip rule. On the other hand, the pump rule allows us to justify introducing all these "$e$-words".

\myparagraph{Unfoldings.} Instead of tailoring the stack $\zb$ simulating $\sigma$ to a specific derivation, we will construct a ``canonical'' stack word $\zb$ for $\sigma$ that only depends on the number of pops in that derivation. We will call this canonical stack word the ``$p$-unfolding'' of $\sigma$, where $p$ is the number of pops it is designed to sustain. This means, intuitively, each $e^+$ is replaced by a large enough concatenation of "$e$-words" so that any sequence of $p$ pops can be executed.

\knowledgenewrobustcmd{\ordersm}{\cmdkl{\unlhd}}
\knowledgenewrobustcmd{\unfold}[2]{\cmdkl{\mathsf{unf}_{#1}(#2)}}

Just to choose the order in which those "$e$-words" appear, we impose an
arbitrary total order (e.g.\ a length-lexicographical ordering) on the set of
"summaries", which we denote $\intro*\ordersm$. 

\begin{definition}
	Let $p \in \NN$ and $\sigma$ be a "$d$-summary". The ""$p$-unfolding"" of $\sigma$, denoted $\intro*\unfold{p}{\sigma}$ is defined inductively w.r.t.\ $d$ and $p$ as follows.
	
	For $d = 0$, the "$p$-unfolding" of a "$0$-summary" (i.e., $\epsilon$) is  $\epsilon$.
	\begin{itemize}
		\item The "$p$-unfolding" of a "$d$-atom" $(f,X) \sigma'$ is $(f,X) \unfold{p}{\sigma'}$.
		
		\item The "$p$-unfolding" of a sequence of "$d$-atoms" $\alpha_1 \cdots \alpha_m$ is defined as $\unfold{p}{\alpha_1} \cdots \unfold{p}{\alpha_{m}}$.
		
		\item The "$p$-unfolding" of a "$d$-block" $B = u_1 \dots u_N e^+ v_1 \dots v_N w$ is
				\[			z^u  z^v ((f,X) \unfold{p-1}{\sigma_1}) \cdots ((f,X) \unfold{p-1}{\sigma_r}) z^v \unfold{p}{w},  \]
			where:
		\begin{itemize}
			\item $(f,X)$ is the first symbol of $B$, that is, the stack symbol such that the first "$d$-atom" of $u_1$ is of the form $(f,X) \sigma'$.
			
			\item $z^u = \unfold{p}{u_1 \dots u_N}$ and $z^v = \unfold{p}{v_1 \dots v_N}$
			
			\item if $p>0$, then  $(\sigma_i)_{1 \leq i \leq r}$ is the family of all "feasible@@summary" "summaries" $\sigma'$ for which we have the equality $\push{(f,X)}{\sigma'}=u_1 \dots u_N e^+ v_1 \dots v_N$, ordered according to $\ordersm$.
			
			\item if $p=0$, then $r=0$, i.e., the "$p$-unfolding" of  $B$ is simply
			$z^u z^v \unfold{p}{w}$.
		\end{itemize}
		 
	\item The "$p$-unfolding" of a "$d$-summary" $\sigma' u B_1 \dots B_m$ is defined as \[ \unfold{p}{\sigma'} \unfold{p}{u} \unfold{p}{B_1} \cdots \unfold{p}{B_m}. \]
	\end{itemize}
\end{definition}

Since the unfolding is obtained by replacing each $e^+$ in a "summary" by a concatenation of words with image $e$ (and all other letters are unchanged), the unfolding has the same image under $\varphi$ as $\sigma$:
\begin{remark}
	For every "summary" $\sigma$ and every $p \in \NN$, we have $\morphism(\unfold{p}{\sigma}) = \morphism(\sigma)$.
\end{remark}

\myparagraph{Removing excess "$e$-words".}
When choosing a stack word to simulate a given "summary" $\sigma$, we pick the
$p$-unfolding, where $p$ is the total number of pops in the entire derivation.
However, some branches will apply less than $p$ pops. The
following lemma is therefore crucial: It allows us to get rid of excess
"$e$-words" that are not needed on less pop-heavy branches, while maintaining the
invariant that we have unfoldings on the stack:
\begin{restatable}{lemma}{lemMonotoneUnfolding}\label{lem:monotone-unfolding}
	For each $\sigma$ and $p \geq 1$: $\unfold{p}{\sigma} \ruleskipstar \unfold{p-1}{\sigma}$.
\end{restatable}
Note that it is not possible to simply skip "$e$-words" at will, since that
requires equality of some infixes in the stack. However, unfoldings are
carefully constructed to allow \cref{lem:monotone-unfolding}.  For
example, the fact that we always follow a uniform order $\ordersm$ on
"summaries" is key. 
\whenFull{A full proof can be
found in \cref{app:monotone-unfolding}.}{}

\myparagraph{Simulating pushes using unfoldings.}
Let us now show how unfoldings are used to mimic derivations of $\CFG$ in $\Ggb$, with pumps and skips.
First, note that all productions of $\CFG$ that are not pushes and pops have direct counterparts in $\Ggb$. Suppose we want to simulate a push, say a rule $(A,X,\sigma_1)\to(B,Y,\sigma_2)$ of $\CFG$, induced by a push rule $(A,X)\to (B,Y)(f,X)$ in  $\Ggb$. If $(A,X,\sigma_1)$ is simulated by $(A,X)[\unfold{p}{\sigma_1}]$, then we can use $(A,X)\to (B,Y)(f,X)$ first in $\Ggb$. But then the stack is $(f,X)\unfold{p}{\sigma_1}$, rather than an unfolding of $\sigma_2$. The following lemma tells us that, using pump and skip, we can replace $(f,X)\unfold{p}{\sigma_1}$ by $\unfold{p}{\sigma_2}$, which will then enable us to continue the simulation.
\begin{restatable}{lemma}{lemPushRuleAux}
	\label{lem:push-rule-aux}
	Let $p\in \NN$, let $(A,X, \sigma_1) \to (B,Y, \sigma_2)$ a rule of $\CFG$ with $\sigma_2 = \push{(f,X)}{\sigma_1}$.
	We have  \[(B,Y)[(f,X)\unfold{p}{\sigma_1}] ~~\derivesindexedstarps~~ (B,Y)[\unfold{p}{\sigma_2}].\]
\end{restatable}
\whenFull{This lemma is proved in Appendix~\ref{app:push-rule-aux}; in fact, only $\rulepump$ and $\ruleskip$ are needed.}{In fact, only $\rulepump$ and $\ruleskip$ are needed to prove this lemma.}

\myparagraph{Simulating pops using unfoldings.}
Pop steps are more difficult. In a production $(A,X,\sigma_1)\to(B,Y,\sigma_2)$ in $\CFG$ induced by a pop rule $(A,X)(f,Y)\to (B,Y)$, $\sigma_2$ is the summary of a word obtained by removing the first letter of a word compressed by $\sigma_1$. This removal might break a block centered around some $e\in\idempotents{\prodmonoid}$, in $\sigma_1$. This means, the symbol $e^+$ is replaced by a concatenation of "summaries". 
However, $p$-unfoldings are designed so that $\unfold{p}{\sigma_1}$ contains enough "$e$-words" for each $e^+$ so that using "skip rules", we can remove a subset of them so that the resulting stack is precisely $(f,Y)\unfold{p-1}{\sigma_2}$. This is shown in the following lemma:
\begin{restatable}{lemma}{lemPopRuleAux}
	\label{lem:pop-case}
	Let $p \geq 1$, let $(A,X, \sigma_1) \to (B,Y, \sigma_2)$ a rule of $\CFG$ with $\sigma_2 \in \pop{(f,Y)}{\sigma_1}$. We have 
	\[(A,X)[\unfold{p}{\sigma_1}] ~~\ruleskipstar~~ (A,X)[(f,Y) \unfold{p-1}{\sigma_2}].\]
\end{restatable}
\whenFull{This is shown in Appendix~\ref{app:pop-rule-aux}.}{} Thus, to simulate this pop rule, we can first invoke \cref{lem:pop-case} and then apply the production $(A,X)(f,Y)\to (B,Y)$.

\myparagraph{Simulating the whole derivation.}
With \cref{lem:push-rule-aux,lem:pop-case} in hand, \cref{prop:CFG-in-pump-skip} is now easy to show. Indeed, it is straightforward to simulate an entire derivation of $\CFG$ in $\Ggb$ with pump and skip rules: Simulating pushes and pops is as explained in \cref{lem:push-rule-aux,lem:pop-case}, and the other productions are immediate. 
\whenFull{The full proof can be found in \cref{app:CFG-in-pump-skip}.}{}

We have thus completed \cref{prop:CFG-in-pump-skip} and hence \cref{correctness-cfg}.

\myparagraph{Putting it all together.}
With \cref{correctness-cfg}, we are prepared to prove \cref{main-upper-bound}.
By Lemma~\ref{lem:exp-summaries}, the "size@@summary" of a "summary" is bounded by an exponential in the "size@@IG" of $\Gg$. 
As a consequence, the size of $\CFG$ is at most doubly exponential in the size of $\Gg$.
Since for a given "context-free grammar", one can compute an exponential-sized NFA for its language's "downward closure"~\cite[Corollary 6]{DBLP:conf/lata/BachmeierLS15}, this yields a tripy exponentially sized NFA for $\dcl{\langIG{\Gg}}$, as desired in \cref{main-upper-bound}.

\section{Lower bounds}\label{sec:lower-bound}

In this section, we prove the lower bounds in our main results.

\myparagraph{NFA Lower bound: Overview.} We begin with \cref{construction-lower-bound}. The overall goal
is to have a unique "complete@@tree"
"derivation tree" that is a full binary tree of doubly exponential depth:
clearly, such a tree must have triply exponentially many leaves. Here, the
challenge is to ensure that the paths have doubly exponential length. 

A standard construction
can enforce \emph{singly} exponentially long paths: Use a height-$n$ stack over
the alphabet $\set{\ltr{0},\ltr{1}}$, and then count up from $\ltr{0}^n$ to
$\ltr{1}^n$, resulting in $2^n-1$ steps. This works because when
incrementing a binary expansion of the form $\ltr{1}^m\ltr{0}w$ (the least
significant digit being on the left), we must replace the prefix
$\ltr{1}^m 0$ with $\ltr{0}^m\ltr{1}$. Here, we store the number
$m\le n$ in the non-terminal, so as to restore the stack height $n$ when pushing
$\ltr{0}^m\ltr{1}$.  

Doing the same for stack height $2^n$ is not so easy: To restore a stack height
of $2^n$, we would need to remember (in the non-terminal) a number $m\le 2^n$.
In fact, enforcing a single run of length $2^{2^n}$ in a pushdown automaton of
polynomial size is not possible: A pushdown automaton that accepts any
word must also accept a word of at most exponential length.

Instead, we exploit the fact that an indexed grammar can simulate a pushdown
automaton \emph{with alternation}: We implement binary counting on a stack of
height $2^n$; in order to replace a prefix $\ltr{1}^m\ltr{0}$
with $\ltr{0}^m\ltr{1}$, we non-deterministically push some number of
$\ltr{0}$'s, but then use alternation to ensure that the stack 
height is exactly $2^n$.

\myparagraph{Step I: Checking the stack via alternation.}
To this end, we introduce syntactic sugar. We will use rules of the form 
\begin{equation} A\xrightarrow{\Aa_1,\ldots,\Aa_r}B\,, \label{check-rule} \end{equation}
where $\Aa_1,\ldots,\Aa_r$ are DFAs over the stack alphabet $I$. The
rule has the same effect as $A\to B$, but it can only be applied to a
term $A[z]$ if for each $i=1,\ldots,r$, the stack $z$ has a prefix in
$\lang{\Aa_i}$. Such rules can be implemented with only polynomial overhead:
Introduce non-terminals $C_i,D_i$ for each $i=1,\ldots,r$ and also $E_q$ for
each state $q$ in the DFAs $\Aa_1,\ldots,\Aa_r$ (we assume the state sets are
disjoint). Then we can simulate \eqref{check-rule} using $A\to D_1C_1$, $D_i\to
D_{i+1}C_{i+1}$ for $i=1,\ldots,r-1$, and $D_r\to B$, which split the term
$A[z]$ into terms $B[z]$ and $C_1[z],\ldots,C_r[z]$. We then run $\Aa_i$ using
rules $C_i\to E_{q_i}$, where $q_i$ is the initial state of $\Aa_i$, for each
$i$. The non-terminals $E_q$ simulate the DFAs: for each transition $(p,f,q)$,
we have $E_p f\to E_q$. To check acceptance, we have $E_q\to\varepsilon$ for
each final state $q$.

\myparagraph{Step II: Implementing a binary counter in DFAs.}
We want to use rules \eqref{check-rule} to check that the current stack height is $2^n$, for which we construct
automata $(\Aa_i)_{1 \leq i \leq n}$ over some alphabet $\Sigma_n$ such that:
	(i)~Each $\Aa_i$ has two states $\zero_i$ and $\one_i$, with $\zero_i$ being initial and $\one_i$ the only final state,
	(ii)~$|\Sigma_n| = n$, and
		(iii)~the intersection $\bigcap_{i=1}^n\lang{\Aa_i}$ contains a single word of length $2^n$.
The construction is simpler if we do this for $2^n-1$ instead of $2^n$, which 
suffices: We can introduce a fresh letter $\ltr{\#}$ and build automata
$\Aa'_i$ with $\lang{\Aa'_i}=\lang{\Aa_i}\ltr{\#}$.

The idea is simply to use $\Sigma_n = \set{\inc_1, \ldots, \inc_{n}}$. Each letter is an increment operation over an $n$-bit binary counter: $\inc_i$ should be read as ``flip the $i$-th bit from $0$ to $1$ and all lower bits from $1$ to $0$''.
Each automaton $\Aa_i$ keeps track of the value of the $i$-th bit throughout that sequence of instructions.

Formally, $\Aa_i = (\set{\zero_i, \one_i}, \Sigma_n, \delta_i, \zero_i, \set{\one_i})$ where $\delta_i(\zero_i, \inc_j)$ is defined as $\one_i$ if $j=i$, it is $\zero_i$ if $j<i$, and it is undefined if $j>i$.
Meanwhile, $\delta_i(\one_i, \inc_j)$ is $\zero_i$ if $j>i$, it is $\one_i$ if $j < i$ and it is undefined if $i=j$.
It is easy to check that there is a unique word accepted by those automata, corresponding to the only correct sequence of instructions to increment a $n$-bit binary counter from $0$ to $2^n-1$, which enforces a single string of length $2^n-1$, as desired.

\myparagraph{Step III: Constructing the indexed grammar.}
Let $\Aa_1,\ldots,\Aa_n$ the DFAs over $\Sigma_n$ built above. Since they will check for exponential stack height, but we also need to store the binary digits on the stack, we modify them slightly. For each $i$, the automaton $\Bcal_i$ will work over the alphabet $\{\bot\}\cup (\Sigma_n\times\{\ltr{0},\ltr{1}\})$ and accept exactly the words of the form $(\alpha_1,b_1)\cdots (\alpha_m,b_m)\bot$ where $\alpha_1\cdots\alpha_m\in \lang{\Aa_i}$.
The $\bot$ letter is used to mark the bottom of the stack.

Consider the following grammar: $\Gg_n = (N_n, T, I_n, P_n, S)$ with
		$ N_n = \set{S, A, B, D, F, Z} $,
		$T = \set{\ltr{a}}$,
		$ I_n = \{\bot\}\cup (\Sigma_n \times \set{\ltr{0},\ltr{1}}) $, and
		$P_n$ contains the following rules: 
		\begin{align*}
			&S\to Z\bot & & \duplicate \to A A &&B\to Z(\alpha,\ltr{1})\\
			&Z\to Z(\alpha,\ltr{0}) && A(\alpha,\ltr{1})\to A && A\bot \to F \\
			& Z\xrightarrow{\Bcal_1,\ldots,\Bcal_n}\duplicate & &A(\alpha,\ltr{0})\to B && F\to\ltr{a} 
		\end{align*}
for each $\alpha\in\Sigma_n$.
The grammar works as follows. Initially, it places $\bot$ on the stack and switches to $Z$. A non-terminal $Z$ will then fill the stack with $\ltr{0}$'s, each of which is non-deterministically annotated with some $\alpha\in\Sigma_n$. After pushing these, it verifies that the stack height is $2^n$, by using $Z\xrightarrow{\Bcal_1,\ldots,\Bcal_n}D$. This $D$ splits into two $A$'s, where an increment is performed on the number encoded on the stack: It removes the prefix of the form $\ltr{1}^m\ltr{0}$ and switches to $B$. After this, it has to put back $\ltr{0}^m\ltr{1}$: To this end, it pushes a single $\ltr{1}$ and then using $Z$ pushes $\ltr{0}$'s non-deterministically. It then uses $Z\xrightarrow{\Bcal_1,\ldots,\Bcal_n}D$ to verify that the stack height is $2^n$. All this repeats until in each branch, all stack contents encode the number $2^{2^n}-1$. This means, all terms are of the form $A[z\bot]$, where $z$ has length $2^n$ and all its digits are $\ltr{1}$'s.
Each such $A[z\bot]$ is then rewritten to $F$, and then to $\ltr{a}$. Since the terms are duplicated before each increment (using $D\to AA$), the final number of $\ltr{a}$ letters is $\exp_3(n)$, deriving $\ltr{a}^{\exp_3(n)}$. It is also straightforward to check that $\ltr{a}^{\exp_3(n)}$ is the only derivable word. 
\whenFull{Details are in \cref{app:main-lower-bound}.}{}

\myparagraph{Computational hardness.} The lower bounds for "downward closure
inclusion" and "equivalence@downward closure equivalence problem" now follow
easily from \cref{construction-lower-bound} and results in \cite{Zetzsche16}. In
\cite{Zetzsche16}, the $\Delta(f)$ property of language classes is introduced.
Roughly speaking, it requires simple closure properties and that for given
$n\in\NN$, one can construct in polynomial time the language
$\set{\ltr{a}^{f(n)}}$. Under additional mild assumptions, \cite[Theorem
15]{Zetzsche16} shows that "downward closure inclusion" and
"equivalence@downward closure equivalence" are $\coNTIME(f)$-hard for
$\Delta(f)$ classes. Since all assumptions besides a small grammar for
$\{\ltr{a}^{\exp_3(n)}\}$ are easy to observe, we may conclude that the indexed
languages are $\Delta(\exp_3)$ and the two problems are $\coNEXP[3]$-hard. 
\whenFull{See
\cref{app:computational-hardness} for details.}{}

\myparagraph{DFA size.}
For \cref{main-lower-bound-dfa}, we adapt an idea from \cite[Theorem
7]{DBLP:conf/lata/BachmeierLS15}, which shows a doubly exponential lower bound
for downward closure DFAs for CFL. It is not difficult
to translate the grammar $\Gg_n$ for $\set{\ltr{a}^{\exp_3(n)}}$ into one for
$L_n = \set{uv \mid u,v\in\{\ltr{0},\ltr{1}\}^* \mid |u|=|v|=\exp_3(n),~u\ne v}$.
It is easy to see that a DFA for $\dcl{L_n}$ requires $\exp_4(n)$ states: For
distinct $u,v\in\set{\ltr{0},\ltr{1}}^*$ with $|u|=|v|$, the DFA must accept
$uv$ and $vu$, but reject $uu$ and $vv$. It therefore must enter distinct
states after reading $u$ and $v$. 
\whenFull{See \cref{app:lower-bound-dfa} for details.}{}

\section{Conclusion}

We have established (asymptotically) tight bounds on the size of an automaton for the "downward closure" of an indexed language. We rely on an algebraic abstraction of stack contents to translate indexed grammars into context-free ones while preserving the "downward closure".
\label{beforebibliography}
\newoutputstream{pages}
\openoutputfile{main.pages.ctr}{pages}
\addtostream{pages}{\getpagerefnumber{beforebibliography}}
\closeoutputstream{pages}
\bibliography{biblio.bib}

\onlyFull{
\newpage
\appendix

\section{Normalising indexed grammars}
\label{app:Chomsky}

When describing indexed grammars we sometimes use production rules of the form not allowed by our definition of "indexed grammar"
\begin{enumerate}
	\item[(i)] $A f \to u$ with $u \in (N \cup T)^* \setminus N$
	\item[(ii)] $A \to u$ with $u \in (N \cup T)^* \setminus (N^2\cup T^*)$.
\end{enumerate}

We now formally define how these rules should be eliminated to obtain an "indexed grammar" as in Definition~\ref{def:IG}.

We start by eliminating rules of the first type: we replace each rule $A f \to u$ with $u \in (N \cup T)^* \setminus N$ by two rules $A f \to A'$ and $A' \to u$, with $A'$ a fresh non-terminal.

It remains to eliminate rules of the form $A \to u$ with $u \in (N \cup T)^* \setminus (N^2 \cup T^*)$.
If $u = B \in N$ then replace the rule with $A \to BC$ and $C\to \epsilon$ with $C$ a fresh non-terminal.
Otherwise, decompose $u$ as $u = w_0 A_1 w_1 \dots A_k w_k$ with $A_1, \dots, A_k \in N$ and $w_0, \dots, w_k \in T^*$. 
Introduce fresh non-terminals $B_1, \cdots, B_k, C_1, \dots, C_{k-1}$.
We replace $A \to u$ with rules 
\begin{itemize}
	\item $A \to W_0 B_1$,
	
	\item $W_i \to w_i$ for all $i \in \set{0,\ldots,k}$,
	
	\item $B_i \to A_i C_i$ for all $i \in \set{1,\ldots,k-1}$,
	
	\item $C_i \to W_i B_{i+1}$ for all $i \in \set{1,\ldots,k-1}$,
	
	\item $B_k \to A_k W_k$
\end{itemize}

Note that $\sum_{i=0}^{k}|w_i| \leq |u|$ and $k \leq |u|$.
Hence, each such rule $A \to u$ is replaced by a set of at most $3|u|+1$ rules, introducing at most $2|u|-1$ non-terminals whose lengths sum up to at most $8|u|+6$.
As a consequence, each rule of the form $A f \to u$ is replaced by a set of at most $3|u|+2$ rules whose lengths sum up to at most $8|u| + 8$.

\section{Additional material from Section~\ref{sec:results}}\label{app:results}
The conclusion section (Section 7) of \cite{DBLP:conf/csl/Kartzow11} claims that the "pumping threshold"  $\PumpConst$ (see \cref{sec:results} for the definition) grows at most doubly exponentially. Here, we briefly explain the mistake in this claim.

The corresponding results are Theorems 25, 32, and 33 in \cite{DBLP:conf/csl/Kartzow11}.
Each of them provides a bound $\ell$ (in terms of the number of states, the input alphabet, and a target configuration) such that in a (collapsible) order-2 pushdown automaton, if there is an accepting path of length $\ge \ell$, then there are infinitely many. These bounds are in the form of functions $f_3$ (for Thm.~25) and $f_6$ (for Thms.~32 and 33). However, both $f_3$ and $f_6$ grow at least triply exponentially in the number of states of the pushdown system.
To see this, we track the functions $f_0,\ldots,f_6$, which are defined across the paper, in \cref{growth-table}. 
\begin{table}
\begin{tabular}{cll}
	function & defined where in \cite{DBLP:conf/csl/Kartzow11} & growth w.r.t.\ $|Q|$ \\\hline
$f_0$ & Lemma~14 on p.~329 & constant \\
$f_1$ & Theorem~22 on p.~331 & exponential \\
$f_2$ & Corollary~23 on p.~331 & exponential \\
$f_3$ & Theorem~25 on p.~332 & triply exponential \\
$f_4$ & Lemma~26 on p.~333 & exponential \\
$f_5$ & Corollary~27 on p.~333 & exponential \\
$f_6$ & Theorem~32 on p.~334 & triply exponential
\end{tabular}
	\caption{Growth of functions in the paper \cite{DBLP:conf/csl/Kartzow11}}\label{growth-table}
\end{table}

\section{Additional material from Section~\ref{sec:sound}}\label{app:soundness}
\subsection{Proof of Lemma~\ref{lem:semigroup-action-properties}}\label{app:semigroup-action-properties}
In this subsection, we prove:
\lemSemigroupActionProperties*

We will prove this in the two separate lemmas below.
\begin{lemma}
	\label{lem:left-action}
	For all $f\in I$, $z\in I^*$ and $X\subseteq N$, we have 
	\[fz\cdot X = f\cdot (z\cdot X).\]
\end{lemma}
\begin{proof}
	If $A\in f\cdot (z\cdot X)$ then (by definition) $A[f]\derivesindexedstar u$, with $u\in (z\cdot X \cup T)^*$. 
	Let $u[z]$ be the "sentential form" obtained by replacing every non-terminal $B$ in $u$ with $B[z]$ (i.e. pushing $z$ onto every stack).
	Since all those non-terminals are in $z \cdot X$, there exists $v \in (X\cup T)^*$ such that $u[z]\derivesindexedstar v$,
	implying that $A[fz]\derivesindexedstar v$ and so $A\in fz\cdot X$.
	
	To show the other inclusion, suppose that $A\in fz\cdot X$ and consider a derivation tree from $A[fz]$ to some $v\in (X\cup T)^*$. 
	Along every branch there is a first node with a label either in $T^*$ or of the form $B[z]$.	
	In the latter case we have $B \in z \cdot X$ (since the tree from this node is a "derivation tree" from $B[z]$ to an element of $(X \cup T)^*$). 
	After deleting everything below these nodes and removing the $z$ suffix from the stack in each label, we obtain a "derivation tree" from $A [f]$ to an element of $(z\cdot X \cup T)^*$, completing the proof. 
\end{proof}
Since non-terminals in $\Useful$ derive terminal words, we have:
\begin{lemma}
	For all $z \in I^*$, $z \cdot \Useful = z\cdot \emptyset = \set{A \in N \mid \langX{A[z]}{\emptyset} \neq \emptyset}$.
\end{lemma} 
\begin{proof}
	Note that the second equality is simply the definition of $z\cdot \emptyset$. We proceed by induction on $z$. For $z=\epsilon$, the statement is equivalent to the definition of $\Useful$. Assuming the statement holds for $z$, then two applications of Lemma \ref{lem:left-action} yield
	$$fz\cdot \Useful = f\cdot(z\cdot \Useful) = f\cdot(z\cdot \emptyset) = fz\cdot \emptyset,$$
	completing the proof.
\end{proof}

\subsection{Proof of Lemma~\ref{lem-soundness}}\label{app:soundness-correctness}

In this subsection, we prove:
\lemSoundness*
We prove \cref{lem-soundness} below, after establishing a preliminary result, which details the relation between derivations in $\Gg$ and $\Ggb$.
\knowledgenewrobustcmd{\project}{\cmdkl{\pi}}
\AP Define $\intro*\project: \Nb\,\Ib^* \cup T \to NI^* \cup T$ to be the function projecting each $(A,X) \in \Nb$ to $A$, each $(f,X) \in \Ib$ to $f$ and each $a \in T$ to itself.
We naturally extend it to a morphism from $(\Nb\,\Ib^* \cup T)^*$ to $(NI^* \cup T)^*$

\begin{lemma}\label{lem-consistent}
	Let $(A,Y)[\zb]\in \Nb\,\Ib^*$ be the "$X$-based annotation" of $A[z]\in NI^*$ for some $X\subseteq N$. 
	If $u\in \langX{A[z]}{X}$, then there exists $\ub\in (X\times\{X\}\cup T)^*$ such that $$\pi(\ub) = u \mbox{ and } (A,Y)[\zb]\derivesindexedstar[\Ggb]\ub.$$
	In particular, if $u\in T^*$ then $A[z]\derivesindexedstar[\Gg] u$ implies that $(A,Y)[\zb]\derivesindexedstar[\Ggb] u$ as well.
\end{lemma}
\begin{proof}
	We proceed by induction on the length of the derivation $A[z] \derivesindexedstar_{\Gg} u$, and distinguish cases according to the production rule used in the first step.
	\begin{itemize}
		\item If the rule is of the form $A \to w\in T^*$, then $u=w$ and the statement is immediate.
		
		\item If the  rule is of the form $A \to BC$, then $u = u_B u_C$ where $B [z] \derivesindexedstar_{\Gg} u_B$ and $C [z] \derivesindexedstar_{\Gg} u_C$. 
		By our assumption, we have $Y = z\cdot X$, and since $A[z]$, $B[z]$ and $C[z]$ can all produce a word in $(X\cup T)^*$, we must have $A,B,C \in Y$ (also, $A\in Y$ by definition). 
		Consequently, we may apply the corresponding derivation in $\Ggb$: $(A,Y) [\zb] \derivesindexed[\Ggb] (B,Y)[\zb] (C,Y)[\zb]$.
		Note that $(B,Y)[\zb]$ and $(C,Y)[\zb]$ are "$X$-based annotations" of $B[z]$ and $C[z]$ respectively, 
		so by the induction hypothesis there are $\ub_B, \ub_C\in (X\times\{X\}\cup T)^*$ such that $\project(\ub_B) = u_B$, $\project(\ub_C) = u_C$, $(B,Y)[\zb] \derivesindexedstar[\Ggb] \ub_B$ and $(C,Y)[\zb] \derivesindexedstar[\Ggb] \ub_C$. 
		Setting $\ub = \ub_B\ub_C$, we thus have $\pi(\ub) = u$ and $(A,Y)[\zb] \derivesindexedstar[\Ggb] \ub$, as desired.
		
		\item If the  rule is of the form $A \to B f$, then $B[fz] \derivesindexedstar_{\Gg} u$. We have $Y = z \cdot X$, and $A \in Y$ by definition. Let $Y' = f\cdot Y = fz \cdot X$. Since $B[fz]$ can produce a word in $(X\cup T)^*$, it must be the case that $B \in Y'$. 
		Hence, we have the corresponding derivation $(A,Y)[\zb] \derivesindexed[\Ggb] (B,Y')[(f,Y)\zb]$.
		Note that $(B,Y')[(f,Y)\zb]$ is an  "$X$-based annotation" of $B[fz]$, so from the induction hypothesis we obtain $\ub\in (X\times\{X\}\cup T)^*$ such that $\pi(\ub)=u$ and $(B,Y')[(f,Y) \zb] \derivesindexedstar[\Ggb] \ub$. Hence, $(A,Y)[\zb] \derivesindexedstar[\Ggb] \ub$.

		\item If the  rule is of the form $A f \to B$, then we have $z = f z'$ for some $z'$ such that $B [z'] \derivesindexedstar_{\Gg} u$. 
		Let $\zb'$ be such that $\zb = (f,X')\zb'$. 
		Since $(A,Y)[\zb]$ is the "$X$-based annotation" of $A[z]$, we must have $X' = z'\cdot X$, and since $B[z']$ derives a word in $(X\cup T)^*$ it follows that $B\in X'$. Hence, there is a corresponding derivation $(A,Y)[(f,X')\zb']\derivesindexed[\Ggb](B,X')[\zb']$. 
		Since $(B,X')[\zb']$ is easily seen to be the "$X$-based annotation" of $B[z']$, the induction hypothesis yields a derivation $(B,X')[\zb'] \derivesindexedstar[\Ggb] \ub$ for some $\ub\in (X\times\{X\}\cup T)^*$ with $\project(\ub)=u$. 
		Hence, $(A,Y)[\zb] \derivesindexedstar[\Ggb] \ub$ as desired.
		
	\end{itemize}
	This concludes the proof.
\end{proof}

\begin{proof}[Proof of Lemma \ref{lem-soundness}]
	The inclusion $\langIG{\Ggb} \subseteq \langIG{\Gg}$ is obtained as follows. For all $w \in \langIG{\Ggb}$, we have a "derivation tree" for $\Ggb$ from $(S, \Useful)$ to $w$. Since $\pi$ maps the rules in $\Pb$ to rules in $P$, it is easy to check that the tree obtained by applying $\pi$ to each node is a "derivation tree" from $S$ to $w$ for $\Gg$, implying that $w \in \langIG{\Gg}$.
	
	To obtain $\langIG{\Gg} \subseteq \langIG{\Ggb}$, suppose that $w \in \langIG{\Gg}$. Then there is a derivation $S \derivesindexedstar_{\Gg} w$, and since $(S,\Useful)$ is a "$\Useful$-based annotation" of $S$ it follows from Lemma \ref{lem-consistent} that $(S, \Useful) \derivesindexedstar[\Ggb] w$, implying $w \in \langIG{\Ggb}$.

	\AP	It remains to prove "productiveness@@grammar".
	We wish to show that every $\ub$ such that $(S,\Useful) \derivesindexed[\Ggb] \ub$ is "productive@@SF".
	
	First observe that for all $\ub$ such that  $(S,\Useful) \derivesindexed[\Ggb] \ub$, every "term" $(A,X)[\zb]$ of $\ub$ is a "$\Useful$-based annotation" of some $A[z] \in NI^*$.
	This follows from the definition of $\Ggb$ and an easy induction on the derivation.
	Then, it is enough to prove that every such "term" $(A,X)[\zb]$ is "productive@@SF", since a "sentential form" can produce a terminal word if and only if all its "terms" can. 
	
	As $(A,X)[\zb]$ is a "$\Useful$-based annotation" of some $A[z]$, we have $X = z \cdot \Useful$.
	Since $A \in X$ by definition of $\Nb$, there is a derivation from $A[z]$ to a word of $T^*$.
	As a result, by Lemma \ref{lem-consistent} there is a derivation from $(A,X)[\zb]$ to a word in $T^*$.
\end{proof}

\section{Additional material from Section~\ref{sec:monoid}}\label{app:monoid}
\subsection{Proof of Lemma~\ref{lem:right-morphism}}
\label{app:right-morphism}

\lemRightMorphism*

\begin{proof}
	All three properties are clearly true for $\zb = \epsilon$. We now focus on non-empty stacks.
	
	From the definitions of $\prodmonoid$ and $\phi$, one sees that property \ref{three} is necessary and sufficient to ensure that no product of two consecutive infixes in $\phi(\zb) = \prod_{i=1}^n \phi(f_i,X_i)$ is equal to $\zeroPM$. It is also clear from the definitions that $\prod_{i=1}^n \phi(f_i,X_i) = \zeroPM$ if and only if two consecutive infixes multiply to equal $\zeroPM$, so we immediately obtain $\ref{two}\Leftrightarrow\ref{three}$.
	
	We now show that $\ref{one} \Rightarrow \ref{three}$. Let $\zb$ be "feasible@@stack", so by definition we have a derivation \begin{align}\label{derivation}
		(\alphaG(f_1), X_1) \derivesindexedstar_{\Ggb} u(\betaG(f_n), f_n\cdot X_n)[\zb]v.\tag{$*$}
	\end{align} We proceed by induction on the derivation length. If \eqref{derivation} has length one, then it must be of the form $(\alphaG(f_1), X_1) \derivesindexed_{\Ggb} (\betaG(f_1), f_1\cdot X_1)[(f_1,X_1)]$, whence \ref{three} holds trivially. We now assume that the length is greater than one, and that the first operation is of the form $(\alphaG(f_1), X_1)\to (B,X_1)(C,X_1)$. In this case, we may assume without loss of generality that $(B,X_1)$ derives $u'(\betaG(f_n), f_n\cdot X_n)[\zb]v'$ for some $u',v'\in \sentforms$. Since $(f_1,X_1)$ must eventually be pushed (and $(\alphaG(f_1),X_1)$ is the only nonterminal which allows this), it follows that $B\Reach{X_1}\alphaG(f_1)$, and that there is a derivation  $(\alphaG(f_1),X_1)\derivesindexedstar_{\Ggb} u''(\betaG(f_n),f_n\cdot X_n)[\zb]v''$ for some $u'',v''\in \sentforms$ that is strictly shorter than \eqref{derivation}. Hence, \ref{three} holds by induction.
	
	Now let us assume that the length of \eqref{derivation} is greater than one, and that the first operation is a push (the only remaining possibility). The push operation must be of the form $$(\alphaG(f_1), X_1) \derivesindexed_{\Ggb}(\betaG(f_1),f_1\cdot X_1)[(f_1,X_1)].$$ If $f_1\cdot X_1 \neq X_2$, then it is easy to see that the right hand side cannot derive any term which pushes $(f_2,X_2)$. Hence, we have
	\begin{align}\label{eq:cond3}
		X_2 = f_1\cdot X_1 \mbox{ and }\betaG(f_1)\Reach{X_2}\alphaG(f_2),\tag{$**$}
	\end{align} and there is a derivation $(\alphaG(f_2),X_2)[(f_1,X_1)]\derivesindexedstar_{{\Ggb}}u'(\betaG(f_n),f_n\cdot X_n)[\zb]v'$. Letting $\zb' = (f_n,X_n)\cdots (f_2,X_2)$, this implies that $\zb'$ is "feasible@@stack" with a derivation $$(\alphaG(f_2),X_2)\derivesindexedstar_{{\Ggb}}u'(\betaG(f_n),f_n\cdot X_n)[\zb']v'$$ that is strictly shorter than \eqref{derivation}. Hence, $\zb'$ satisfies \ref{three} by induction, and with \eqref{eq:cond3} we immediately obtain \ref{three} for $\zb$.
	
	Finally, we show that $\ref{three}\Rightarrow \ref{one}$. 
	Let us assume that $\zb$ satisfies \ref{three}. 
	Recall that we assumed that for all $f \in I$ there is a rule pushing $f$.
	Then for $i=1,\ldots,n-1$ we have
	\[(\alphaG(f_i),X_i)\derivesindexed[\Ggb] (\betaG(f_i), X_{i+1})[(f_i,X_i)]\] and \[(\betaG(f_i),X_{i+1})\derivesindexedstar_{{\Ggb}} u_i(\alphaG(f_{i+1}),X_{i+1})v_i\] for some $u_i,v_i\in \sentforms$, as well as \[(\alphaG(f_n),X_n)\derivesindexed[\Ggb](\betaG(f_n),f_n\cdot X_n)[(f_n,X_n)].\] It is clear that we may combine these derivations to obtain \[(\alphaG(f_1), X_1) \derivesindexedstar_{\Ggb} u(\betaG(f_n), f_n\cdot X_n)[\zb]v,\] proving that $\zb$ is "feasible@@stack".
\end{proof}

\subsection{Proof of Lemma~\ref{lem:right-monoid-A}}
\label{app:right-monoid-A}

\lemRightMonoidA*

\begin{proof}
	From the first condition, we have derivations 
	$$(C,X)\derivesindexedstar_{{\Ggb}}u_C(A,X)v_C \mbox{ and }(B,X)\derivesindexedstar_{{\Ggb}}u_D(D,X)v_D$$
	for some $u_C,v_C,u_D,u_D\in \sentforms$. Let $z = f_n\cdots f_1$. From the definition of $\phi$ (and the fact that $\phi(\annotate{z}{X})\neq \zeroPM$) it follows easily that $A = \alphaG(f_1)$, $B = \betaG(f_n)$ and $Y = z\cdot X$. Hence, Lemma \ref{lem:right-morphism} ensures that there is a derivation $$(A,X)\derivesindexedstar_{{\Ggb}} u'(B,Y)[\zb^X]v'$$ for some $u',v'\in \sentforms$, which we combine with the derivations above to obtain $$(C,X)\derivesindexedstar_{{\Ggb}}u(D,Y)[\annotate{z}{X}]v,$$ proving one direction.
	
	For the other implication, suppose such a derivation exists. Eventually, $(f_1,X)$ must be pushed onto an empty stack, and since $(A,X)$ is the only non-terminal which facilitates this operation, it follows that $C\Reach{X}A$. Similarly, the derivation must eventually push the topmost symbol in $\annotate{z}{X}$, and the only non-terminal which can result from this operation is $(B,Y)$. This implies that $(B,Y)[\annotate{z}{X}]\derivesindexedstar_{{\Ggb}}u(D,Y)[\annotate{z}{X}]v$, hence $B\Reach{Y}D$, completing the proof.
\end{proof}

\subsection{Proof of Lemma~\ref{lem:right-monoid-M}}
\label{app:right-monoid-M}

\lemRightMonoidM*
\begin{proof}
	Equivalence of the second and third statements follows easily from Lemma \ref{lem-consistent}. Hence, it suffices to prove equivalence of the first and second statements.
	We proceed by induction on $z$.
	If $|z| =1$ then we have $\annotate{z}{X} = (f, X)$ with $\morphism(f,X) = (B, Y,  M, A, X)$, and the equivalence holds simply by definition.
	
	If $|z| >1$, then let $w$ be such that $z = f w$. Let $Z = w \cdot X$, so that $\annotate{z}{X} = (f,Z)\annotate{w}{X}$, and let us write $(B_f, Y,  M_f, A_f, Z) = \morphism(f,Z)$ and $(B_w, Z, M_w, A_w, X) = \morphism(\annotate{w}{X})$. Note that $M = M_f M_w$. We prove the two directions separately.
	
	\noindent
	$\Rightarrow$: \ \ \
	Suppose $M(D,C) = \top$. Then there exists $E$ such that $M_f(D,E) = M_w(E,C) = \top$. From the induction hypothesis applied to $f$, it follows that $D \in Y$, $E \in Z$ and there exist $u_1,v_1 \in (Z \cup T)^*$ such that $D[f] \derivesindexedstar_{{\Gg}} u_1 E v_1$. On the other hand, applying the induction hypothesis to $w$ yields $C \in X$ and $u_2,v_2 \in (X \cup T)^*$ such that $E[w]\derivesindexedstar_{{\Gg}} u_2 Cv_2$.
	
	Since $Z = w \cdot X$ and $u_1,v_1 \in (Z \cup T)$, there must be $u_3, v_3 \in (X\cup T)^*$ such that $u_1[w] \derivesindexedstar_{\Gg} u_3$ and $v_1[w] \derivesindexedstar_{\Gg} v_3$. 
	Combining the facts above, we get \[D[z] \derivesindexedstar_{{\Gg}} u_1[w] E[w]v_1[w] \derivesindexedstar_{\Gg} u_3 E[w] v_3 \derivesindexedstar_{\Gg} u_3 u_2 C v_2 v_3\]
	as desired. 
	
	\noindent
	$\Leftarrow$: \ \ \
	Suppose $D \in Y$, $C \in X$ and there exist $u, v \in (X \cup T)^*$ such that $D[z] \derivesindexedstar_{{\Gg}} u C v$. Then we must have $D[f] \derivesindexedstar_{\Gg} u_1 U v_1$, $U[w] \derivesindexedstar_{\Gg} u_2 C v_2$,  $u_1[w]\derivesindexedstar_{\Gg} u_3$ and $v_1[w] \derivesindexedstar_{\Gg} v_3$ for some $U\in N$ and "sentential forms" $u_1, u_2, u_3, v_1, v_2, v_3$ such that $u = u_3 u_2$ and $v= v_2 v_3$.
	As a consequence, we have $u_2, u_3, v_2, v_3 \in (X \cup T)^*$.
	Since $w \cdot X = Z$, we get that $u_1, v_1 \in (Z \cup T)^*$, and since $C \in X$, the  derivation $U[w] \derivesindexedstar_{\Gg} u_2 C v_2 \in (X \cup T)^*$ proves that $U \in w\cdot X = Z$. 
	By the induction hypothesis, we have $M_w(D,U) = M_f(U,C) = \top$, implying that $M(D,C) = \top$.
\end{proof}

\section{Additional material from Section~\ref{sec:pump-skip}}\label{app:pump-skip}
\subsection{Proof of Proposition~\ref{prop:eliminate-pump-skip}}
\label{app:eliminate-pump-skip}

We prove the following statement:

\lemPumpSkip*
The following two auxiliary lemmas are required.
The first one shows that we can eliminate "pump rules", the second one that we can eliminate "skip rules".

\AP Call a "sentential form" $u$ of $\Ggb$ ""reachable@@SF"" if $u \in \langSF{(S,\Useful)}$. 
Call a "term" $(B,X)[\zb]$ ""reachable"" if it appears in a derivation from $(S,\Useful)$.

\begin{lemma}\label{lem:eliminate-pumps}
	Let $e = (B,X,M,A,X) \in \idempotents{\prodmonoid} \setminus \set{\zeroPM, \neutral}$, let $(B,X) [\zb]$ be a "reachable" "term" of $\Ggb$ and let $z_e \in \Ib^*$ be such that $\morphism(z_e) =e$. Then $$\langX{(B,X) [z_e \zb]}{\emptyset} \subseteq \dcl{\langX{(B,X)[\zb]}{\emptyset}}.$$
\end{lemma}

\begin{proof}
	Let $w \in \langX{(B,X)[z_e\zb]}{\emptyset}$, so there is a derivation $$(B,X)[z_e\zb] \derivesindexedstar[\Ggb] w.$$
	Since $\morphism(z_e) = e \neq \zeroPM$, by Lemma~\ref{lem:right-morphism}, $z_e$ is "feasible@@stack", so there is a derivation $(A,X) \derivesindexedstar[\Ggb] u (B,X)[z_e] v$ with $u,v \in \sentforms$.
	Furthermore, since $e$ is "idempotent", by definition of the product in $\prodmonoid$ we have $B \Reach{X} A$, i.e., there is a derivation $(B,X) \derivesindexedstar[\Ggb] u' (A,X) v'$. In total, we obtain \[(B,X)[\zb] \derivesindexedstar[\Ggb] u' (A,X)[\zb] v' \derivesindexedstar[\Ggb] u' u (B,X)[z_e \zb] v v' \derivesindexedstar[\Ggb] u' u w v v'.\] 

	Moreover, since $\Ggb$ is "productive@@grammar" and $(B,X)[\zb]$ is a "reachable" "term" of $\Ggb$, it follows that $u' u w v v'$ is "reachable" as well, and can thus derive a terminal word $w'$. Since $w'$ necessarily contains $w$ as a subword, the proof is complete.
\end{proof}

\begin{lemma}\label{lem:eliminate-skips}
	Let $(A,Y) [z' u_1 \dots u_N z_e u_1 \dots u_N \zb]$ be a "reachable" "term" of $\Ggb$, and let $e=(B,X, M, A,X) \in \idempotents{\prodmonoid}$ be such that $\morphism(u_1) = \dots = \morphism(u_N) = \morphism(z_e) =e$. Then 
	$$\langX{(A,Y) [z' u_1 \dots u_N \zb]}{\emptyset} \subseteq \dcl{\langX{(A,Y)[z' u_1 \dots u_N z_e u_1 \dots u_N \zb]}{\emptyset}}.$$
\end{lemma}

\begin{proof}
	Let $w \in \langX{(A,Y) [z'u_1 \dots u_N \zb]}{\emptyset}$, and let $\tau$ a derivation tree for $(A,Y) [z'u_1 \dots u_N \zb] \derivesindexedstar[\Ggb] w$.
	
	Consider the set of nodes $\nu$ such that either
	\begin{enumerate}
		\item $\nu$ is a leaf with a label in $T^*$ and $\zb$ is a suffix of the stack content of each of its ancestors, or
		
		\item $\nu$ is labeled $(C, X)[u_{i+1} \dots u_N \zb]$ for some $i \in \set{0, \dots, |N|}$ (where $u_{N+1}$ is the empty string) and $ u_{i+1} \dots u_N \zb $ is a suffix of the stack content of each of its ancestors.
	\end{enumerate}
	A node of the second type which also satisfies $M(C,C) = \top$ is called a \emph{special node}.
	\begin{claim}
		Every branch of $\tau$ contains either a node of the first type or a special node.
	\end{claim}
	
	\begin{claimproof}
		Consider a branch of $\tau$. If the leaf of this branch is not of the first type, then the prefix $z' u_1\dots u_N$ must be fully popped. 
		Clearly there must be distinct nodes $\nu_0, \dots, \nu_{N}$ of the second type along this branch, where $\nu_i$ is labeled $(A_i, X_i)[u_{i+1} \dots u_N \zb]$ for some $A_i\in N$, and each of its ancestors has a stack content with suffix $u_{i+1} \dots u_N \zb$. Consequently, for $i = 0,\ldots,|N|-1$ we can define the "derivation tree" obtained by restricting $\tau$ to $\nu_i$ and its descendants, and then removing all nodes (and their descendants) where $u_{i+1} \dots u_N \zb$ is not a suffix of the stack, as well as all descendants of $\nu_{i+1}$. 
		By construction, $\nu_{i+1}$ is a leaf of the resulting tree, and so we obtain a derivation $$(A_i, X_i)[u_i \dots u_N \zb] \derivesindexedstar[\Ggb] w_i (A_{i+1}, X_{i+1})[u_{i+1} \dots u_N \zb] w_i'$$ with $w_i, w_i' \in \sentforms$. Since $\morphism(u_i) = e$, we must have $X_i = X_{i+1} = X$ and $M(A_i, A_{i+1}) = \top$. Thus, it follows from Lemma~\ref{lem:idempotent-property} that there is some $i$ such that $M(A_i,A_i) = \top$, so $\nu_i$ is a special node. 
	\end{claimproof}
	
	Let $\nu$ be a special node, labeled $(C, X)[u_{i+1} \dots u_N \zb]$ with $M(C,C) = \top$. Since $\morphism(u_1) = \dots = \morphism(u_N) = \morphism(z_e) = e$, it follows that $\morphism(u_{i+1} \dots u_N z_e u_1 \dots u_{i}) = e$.
	As a consequence, since $M(C,C) = \top$, we have a derivation $(C,X) [u_{i+1} \dots u_N z_e u_1 \dots u_{i}] \derivesindexedstar[\Ggb] w_-(C,X)w_+$ with $w_-, w_+ \in \sentforms$ (see Lemma~\ref{lem:right-monoid-M}).
	
	For the following construction we refer to Figure~\ref{fig:skip} for a visual presentation.
	
	Consider the set $V$ of minimal nodes (for the ancestor relation) which are either of the first type or special.	
	Notice that $V$ intersects every branch exactly once: at most once by the minimality requirement, at least once by the claim above. 
	We can thus define the subtree whose root is the same as $\tau$ and whose leaves are $V$. 
	By construction, every node in this tree is labeled with either a terminal word or a term whose stack has $\zb$ as a suffix. Define $\tau'$ the tree obtained by replacing each suffix $\zb$ with $z_e u_1 \dots u_N \zb$, and notice that the resulting tree $\tau'$ is still a valid derivation tree. 
	
	For each leaf $\nu$ of $\tau'$ that was a special node of $\tau$, with a label $(C,X) [u_{i+1} \dots u_N z_e u_1 \dots u_N \zb]$ in $\tau'$ for some $(C,X)$ and $i$, we append a derivation tree $\tau_{\nu}$ for 
	\[(C,X) [u_{i+1} \dots u_N z_e u_1 \dots u_N\zb] \derivesindexedstar[\Ggb] w_-(C,X)[u_{i+1} \dots u_N \zb]w_+,\] where $w_-, w_+ \in \sentforms$. Now take the subtree of $\tau$ rooted at $\nu$, and append it at the leaf of $\tau_{\nu}$ labeled with $(C,X)[u_{i+1} \dots u_N \zb]$.
	
	The result is a derivation tree from $(A,Y)[z' u_1 \dots u_N z_e u_1 \dots u_N \zb]$ to a "sentential form" $\wb$ of which $w$ is a subword. Since $\Ggb$ is "productive@@grammar" and $(A,Y)[z' u_1 \dots u_N z_e u_1 \dots u_N \zb]$ is a "reachable" "term", there is a word $w' \in T^*$ such that $\wb \derivesindexedstar[\Ggb] w'$, and since $w \subword \wb$ we have $w \subword w'$, completing the proof.
\end{proof}

\begin{proof}[Proof of Proposition~\ref{prop:eliminate-pump-skip}]
	We show the stronger statement that for all "reachable" "sentential form" $u$ in $\Ggb$, for all terminal word $w \in T^*$ such that $u \derivesindexedstarps[\Ggb] w$, there exists $w' \in T^*$ such that $u \derivesindexedstar[\Ggb] w'$ and $w \subword w'$. 
	The result then follows by taking $u = (S,\Useful)$.	

	We proceed by induction on the derivation $u \derivesindexedstarps[\Ggb] w$, distinguishing cases according to the first step. Note that the base case, where $u = w$, is trivial.
	\begin{itemize}
		\item If the first step is a production rule in $\Ggb$, say $u \derivesindexed[\Ggb] u'$, then $u' \derivesindexedstarps[\Ggb] w$ with a shorter derivation. By the induction hypothesis, $u' \derivesindexedstar[\Ggb] w'\in T^*$ with $w \subword w'$, hence $u \derivesindexed[\Ggb] u' \derivesindexedstar[\Ggb] w'$.
		
		\item If the first step is a "pump rule", say 
		\[u = u_- (B,X)[ \zb] u_+ \rulepump u_- (B,X)[z_e \zb] u_+ = u',\]
		with $\morphism(z_e) =e$ for some idempotent $e = (B,X,M,A,X)$.
		then $u' \derivesindexedstarps[\Ggb] w$, and the induction hypothesis implies that $u' \derivesindexed[\Ggb] w'\in T^*$ with $w \subword w'$. 
		Let us write $w' = w_- w_B w_+$, where
		\begin{itemize}
			\item $u_- \derivesindexedstar[\Ggb] w_-$,
			
			\item $u_+ \derivesindexedstar[\Ggb] w_+$, and
			
			\item $(B,X)[z_e \zb] \derivesindexedstar[\Ggb] w_B$.
		\end{itemize}
		
		By Lemma~\ref{lem:eliminate-pumps} we have that $u \derivesindexedstar[\Ggb] w_- w'_B w_+$ with $w_B \subword w'_B$. 
		The desired result follows, since $w \subword w' \subword w_- w'_B w_+$.
		
		\item If the first step is a "skip rule", say
		\begin{multline*}
			u = u_- (A,X)[z' u_1 \dots u_N z_e u_1 \dots u_N \zb] u_+ \\
			\ruleskip u_- (A,X)[z' u_1 \dots u_N \zb] u_+ = u',
		\end{multline*}
		then 
		$u' \derivesindexedstarps[\Ggb] w$, and by the induction hypothesis $u' \derivesindexed[\Ggb] w'\in T^*$ with $w \subword w'$. 
		Let us write $w' = w_- w_A w_+$ where 
		\begin{itemize}
			\item $u_- \derivesindexedstar[\Ggb] w_-$,
			
			\item  $u_+ \derivesindexedstar[\Ggb] w_+$, and
			
			\item $(A,X)[z' u_1 \dots u_N \zb] \derivesindexedstar[\Ggb] w_A$.
		\end{itemize}
		
		Then by Lemma~\ref{lem:eliminate-skips} we obtain $$(A,X)[z' u_1 \dots u_N z_e u_1 \dots u_N \zb] \derivesindexedstar[\Ggb] w'_A$$ with $w_A \subword w'_A$. We thus have $u \derivesindexedstar[\Ggb] w_- w'_A w_+$, and since $w \subword w' \subword w_- w'_A w_+$, the result follows and the proof is complete.
	\end{itemize}

\end{proof}
\section{Additional material from Section~\ref{sec:summaries}}\label{app:summaries}
\subsection{Proof of Lemma~\ref{lem:exp-summaries}}
\label{app:exp-summaries}

In this section we prove the following statement.

\ExponentialSummaries*

Its proof is very similar to the one of Theorem~26 in~\cite{GimbertMT25arxiv}.
We start by recalling some classical facts on Green relations.
For a more in-depth introduction to those, see for instance~\cite{McNaughton89}, or~\cite{Colcombet11}.

\begin{lemma}
	\label{lem:Green1}
	In a finite monoid $\mathbf{M}$, every $\Hgreen$-class contains at most one idempotent.	
\end{lemma}

\begin{lemma}
	\label{lem:Green2}
	In a finite monoid $\mathbf{M}$, for all $x,y \in \mathbf{M}$, if $x \Jgreen y$ and $x \Lleq y$ (resp. $x \Rleq y$) then $x \Lgreen y$ (resp. $x \Rgreen y$).
\end{lemma}

\begin{theorem}\label{thm:Jheight}
	$\Jlength{\prodmonoid} \leq \frac{(N^2+N+2)}{2}+2$.
\end{theorem}

\begin{proof}
	Let us start by using another definition of the "regular $\mathcal{J}$-length".
	By~\cite[Appendix B]{Jecker21arxiv}, the "regular $\mathcal{J}$-length" of $\prodmonoid$ is the largest $m$ such that there is an injective homomorphism from the max monoid $(\set{1,\dots, m}, \max, m)$ to $\prodmonoid$.
	
	Let $m \in \NN$, let $\theta : \set{1,\dots, m} \to \prodmonoid$ be such a homomorphism.
	We must show that  $m \leq \frac{(N^2+N+2)}{2}+2$.
	
	If $\zeroPM$ is in the image of $\theta$ then $\theta(m) = \zeroPM$.
	If $\neutral$ is in the image of $\theta$ then $\theta(1) = \neutral$.
	Since all $i$ in $\set{2,\dots, m-1}$ are idempotents, the image of each $i$ by $\theta$ must also be an idempotent. 
	As a result, we can set $\theta(i) = (B_i, X_i, M_i, A_i , X_i)$ for all $1< i <m$, with $M_i^2 = M_i$.
	
	Furthermore, for all $1< i < j< m$, since $\max(i,j) =j$, we must have 
	\begin{align*}
		& (B_j, Y_j, M_j, A_j , X_j)\cdot (B_i, X_i, M_i, A_i , X_i)\\
		= & (B_i, X_i, M_i, A_i , X_i) \cdot (B_j, Y_j, M_j, A_j , X_j)\\
		= & (B_j, X_j, M_j, A_j , X_j)
	\end{align*} 
	We infer that  $X_j = X_i$, $B_j = B_i$ and $A_j = A_i$ for all $i<j$.
	Since $\theta$ is injective, $M_2, \dots, M_{m-1}$ must be distinct.
	
	Define the function $\theta' : \set{1,\dots, m-2}$ mapping each $i$ to $M'_{i+1}$.
	It suffices to observe that $\theta_p$ is an injective homomorphism from the max monoid of size $m-2$ to $\matrixmonoid$.
	As a consequence, we have $m-2 \leq \Jlength{\matrixmonoid}$.
	By Theorem~\ref{thm:Ismael-height}, we have $\Jlength{\matrixmonoid} \leq \frac{N^2 + N + 2}{2}$.
	As a result, $m \leq \frac{(N^2+N+2)}{2}+2$. 
\end{proof}

Observe that the size of $\prodmonoid$ is bounded by $N^2 2^{N^2+2N}$.
Define \[K = ((2N+1)N^8 2^{4N^2+8N})^{\frac{(N^2+N+2)}{2}+2}.\] By combining the results above, we obtain the following corollary.

\begin{corollary}
	\label{cor:bound-idempotent}
	Let $z \in \prodmonoid^*$ with $|z| \geq K$, there exist $u_1, \dots, u_N \in \prodmonoid^*$ and $v_0, \ldots, v_N \in \prodmonoid^*$ and $e \in \idempotents{\prodmonoid}$ such that
	\begin{itemize}
		\item $u_1 \dots u_N v_0 \dots v_N$ is an infix of $z$
		
		\item $\morphism(u_1) = \dots = \morphism(u_N) = \morphism(v_0) = \cdots = \morphism(v_N) = e$
	\end{itemize}
\end{corollary}

In what follows we distinguish the "size@@summary" of a "$d$-summary"/"$d$-block" from its ""length@@summary"", which is simply its length as a word of "$d$-atoms", $e^+$ letters and $d'$-summaries for various $d'<d$.

\begin{lemma}
	\label{lem:length-blocks}
	For all $\zb \in \Ib^*$, $\push{\zb}{\epsilon} = \sigma' u B_1 \dots B_k$ is such that $u$ has "length@@summary" at most $K$ and all $B_j$ at most $2K$.
\end{lemma}

\begin{proof}
	We first show it for $u$.
	If $\push{\zb}{\epsilon} = \sigma' u B_1 \dots B_k$ then $u$ cannot have an infix of the form $u_1 \dots u_N v_0 \dots v_N$ with $\morphism(u_i) =\morphism(v_i) = \morphism(v_0)= e$ for all $i$, for any $e \in \idempotents{\prodmonoid}$. This is because when such a pattern appears $u$ is turned into a "$d$-block".
	As a consequence, by Corollary~\ref{cor:bound-idempotent}, $u$ has length at most $K-1$.
	
	For the blocks, we show that in a block $u_1 \cdots u_N e^+ v_1 \cdots v_N w$ appearing in $\push{\zb}{\epsilon}$, the "lengths@@summary" of $u_1 \dots u_N e^+$ and $v_1 \dots v_N w$ are always at most $K$.
	We show this by induction on $|\zb|$. For $\zb = \epsilon$ this is trivial.
	Now suppose $|\zb|> 0$, let $(f,X) \zb' = \zb$.
	By induction hypothesis $\push{\zb'}{\epsilon}$ has the property.
	
	Let $\sigma'' u' B'_1 \dots B'_m = \push{\zb'}{\epsilon}$
	We show the property on $\push{\zb}{\epsilon}$ (which is equal to  $\push{(f,X)}{\push{\zb'}{\epsilon}}$ by definition) by following the cases of the definition of $\push{\_}{\_}$.
	
	In cases (1), (a) and (ii) blocks remain the same, thus the property still holds.
	In case (i), we define a new block $B = u_1 \cdots u_N e^+ v_1 \cdots v_N w$ by concatenating $(f,X)$ with $u'$. Since we showed that $u'$ has "length@@summary" at most $K-1$, $B$ has "length@@summary" at most  $K$, hence so do $u_1 \cdots u_N e^+$ and $v_1 \cdots v_N w$.
	In case (B), we obtain the property immediately.
	In case (A), we form a new block $u_1 \cdots u_N e^+ v'_1 \cdots v'_N w'$ by merging $B$ with one of the $B'_j = u'_1 \cdots u'_N e^+ v'_1 \cdots v'_N w'$.
	Since $u_1 \cdots u_N e^+$ and $v'_1 \cdots v'_N w'$ both have length at most $K$, the property is maintained.

\end{proof}

\begin{lemma}
	\label{lem:nb-blocks}
	For all $\zb \in \Ib^*$, $\push{\zb}{\epsilon}$ has at most $|\prodmonoid|^{4\Jlength{\prodmonoid}+1}$ blocks.
\end{lemma}

\begin{proof}
	Let $\sigma = \sigma' u B_1 \dots B_m$ be a "$d$-summary". 
	For each $i$ let $B_i =u_{i,1} \dots u_{i,N} e_i^+ v_{i,1}, \ldots, v_{i,N} w_i$ and for each  $i<j$ define $\alpha_{i,j} = \morphism(v_{i,1} \dots v_{i,N} w_i B_{i+1} \cdots B_{j-1} u_{j,1} \dots u_{j,N})$.
	
	Note that since $\morphism(v_{i,1}) =e_i$ and $e_i \in \idempotents{\prodmonoid}$ for all $i$ we have \[\alpha_{i,j}=  e_i \cdot \morphism(v_{i,1} \dots v_{i,N} w_i B_{i+1} \cdots B_{j-1} u_{j,1} \dots u_{j,N})\] and thus 
	\[\alpha_{i,j} \cdot \alpha_{j,k} = \alpha_{i,k} \text{ for all } i<j<k\]
	
	Suppose by contradiction that $m> |\prodmonoid|^{4\Jlength{\prodmonoid}+1}$, then by pigeonhole principle there exist $e \in \idempotents{\prodmonoid}$ and $i_1 <\dots < i_p \in \set{1,\dots, K}$ such that 
	$p > (|\prodmonoid|)^{4\Jlength{\prodmonoid}}$ and $e_{i_k} = e$ for all $k \in \set{1,\dots,p}$.
	
	Then by Theorem~\ref{thm:Ismael-main}, there exist $k, \ell$ such that $\morphism(\alpha_{i_k,i_\ell})$ is an idempotent $e'$.
	
	Since $\sigma$ is a "$d$-summary", we have $\Jheight{e'} = d = \Jheight{e}$. 
	In consequence, since $e' \Jleq e$ and they are both idempotent, we must have $e \Jgreen e_i$.
	Furthermore, since \[e' \Rleq \alpha_{i_k} \Rleq \phi_{\mathbb{M}}(\pi(v_{i_k,1})) = e,\] we have $e' \Rleq e$ and thus $e \Rgreen e'$ by Lemma~\ref{lem:Green1}. 
	Similarly, since $e' \Lleq \alpha_{i_\ell} \Lleq e$ we have $e' \Lleq e$ and thus $e \Lgreen e'$.
	
	We obtain $e \Hgreen e'$. By Lemma~\ref{lem:Green2}, this implies $e = e'$.
	This is a contradiction since then the blocks from $B_{i_k}$ to $B_{i_\ell}$ should have been merged when $B_{i_k}$ was created.
	As a result, we must have $m \leq |\prodmonoid|^{4\Jlength{\prodmonoid}+1}$.
\end{proof}

We now have all necessary tools to show Lemma~\ref{lem:exp-summaries}.

\begin{proof}[Proof of Lemma~\ref{lem:exp-summaries}]
	We show that for all $\zb$ and $d$, if $\push{\zb}{\epsilon}$ is a "$d$-summary" then it has "size@@summary" at most $(8|\prodmonoid|^{4\Jlength{\prodmonoid}+1} K )^d$. 
	
	We do an induction  on $d$. The property trivially holds for $d=0$.
	Let $d>0$, suppose the property holds for $d-1$. 
	
	Let $\push{\zb}{\epsilon} = \sigma' u B_1 \dots B_k$.
	By Lemma~\ref{lem:nb-blocks} we must have $k \leq |\prodmonoid|^{4\Jlength{\prodmonoid}+1}$.
	By Lemma~\ref{lem:length-blocks} we have $|u| \leq K$ and $|B_i| \leq 2K$ for all $i$.
	
	Therefore $\push{z}{\epsilon}$ has "length@@summary" at most $(|\prodmonoid|^{4\Jlength{\prodmonoid}+1}+1) K +1$, which is bounded by $4|\prodmonoid|^{4\Jlength{\prodmonoid}+1} K$.
	
	Let $a\sigma''$ be a $d$-atom appearing in $\zb$, with $\sigma''$ a "$(d-1)$-summary".
	Note that there must be a (strict) infix $\zb''$ of $\zb$ such that $\push{\zb''}{\epsilon} = \sigma''$.
	By induction hypothesis $\sigma''$ has size at most $(8|\prodmonoid|^{4\Jlength{\prodmonoid}+1}K)^{d-1}$.
	Thus every "$d$-atom" has size at most $(8|\prodmonoid|^{4\Jlength{\prodmonoid}+1} K)^{d-1}+1 \le 2(8|\prodmonoid|^{4\Jlength{\prodmonoid}+1} K)^{d-1}$.
	With the bound on its "length@@summary", we conclude that $\push{z}{\epsilon}$ has "size@@summary" at most 
	\begin{align*}
		& 2(8|\prodmonoid|^{4\Jlength{\prodmonoid}+1} K)^{d-1} (4|\prodmonoid|^{4\Jlength{\prodmonoid}+1} K ) \\
		\leq & 	(8(|\prodmonoid|^{4\Jlength{\prodmonoid}+1} K )^{d}
	\end{align*}
	
	We get the result by applying this bound with $d = \Jlength{\prodmonoid}$, which is at most $\frac{(N^2+N+2)}{2}+2$ by Theorem~\ref{thm:Jheight}. 
\end{proof}

\section{Additional material from Section~\ref{sec:CFG}}
\subsection{Proof of Proposition~\ref{prop:Ggb-in-CFG}}
\label{app:Ggb-in-CFG}

\GgbInCFG*

\begin{proof}
	Let $w \in \langIG{\Ggb}$.
	There is a "derivation tree" from $(S,\Useful)$ to $w$.
	Let $\tau$ its tree structure and $\lambda: \tau \to \Nb \Ib^* \cup T^*$ its labeling.
	
	We now define $\mu: \tau \to \FT \cup T^*$ a labeling of $\tau$ with non-terminals and (words of) terminals of $\CFG$ such that for all $\nu \in \tau$, 
	\begin{itemize}
		\item if $\lambda(\nu) \in T^*$ then $\mu(\nu) = \lambda(\nu)$
		
		\item if $\lambda(\nu) = (A,X)[\zb] \in \Nb \Ib^*$ then $\mu(\nu) = (A,X,\push{\zb}{\epsilon})$
	\end{itemize}
	
	We show that the resulting labeled tree is a "derivation tree" from $(S,\Useful,\epsilon)$ to $w$ in $\CFG$, thereby showing the lemma.
	
	Clearly the "leaf word" of $\tau, \mu$ is $w$.
	It remains to show that this is a "derivation tree". 
	Since $\tau, \lambda$ is a "derivation tree" from $(S, \Useful)$ to $w$, all stack contents appearing in it must be "feasible@@stack".
	As a consequence, $\mu$ maps all nodes of $\tau$ to non-terminals of $\CFG$. 
	Let $\nu$ an internal node of $\tau$, and $(A,X)[\zb] = \lambda(\nu)$. One of the following cases holds.
	\begin{itemize}
		\item $\nu$ has one child labeled $w' \in T^*$, and there is a rule $(A,X) \to w'$ in $\Ggb$.
		Then $(A,X, \push{\zb}{\epsilon}) \to w'$ is a rule of $\CFG$.
		
		\item $\nu$ has two children labeled $B[\zb]$ and $C[\zb]$, and there is a rule $(A,X) \to (B,X) (C,X)$ in $\Ggb$.
		Then $(A,X, \push{\zb}{\epsilon}) \to (B,X, \push{\zb}{\epsilon}) (C,X, \push{\zb}{\epsilon})$ is a rule of $\CFG$.
		
		\item $\nu$ has a child labeled $B[(f,X)\zb]$, and there is a rule $(A,X) \to (B,Y) (f,X)$ in $\Ggb$.
		Then, since $\push{(f,X)}{\push{\zb}{\epsilon}} = \push{(f,X)\zb}{\epsilon}$, it must be that
		
		$(A,X, \push{\zb}{\epsilon}) \to (B,Y, \push{(f,X)\zb}{\epsilon})$ is a rule of $\CFG$.
		
		\item $\nu$ has a child labeled $B[\zb_-]$ with $\zb = (f,Y) \zb_-$, and there is a rule $(A,X) (f,Y) \to (B,Y)$ in $\Ggb$.
		Then, since by definition $\push{(f,X)}{\zb_-} = \push{\zb}{\epsilon}$, we have $\zb_- \in \pop{(f,Y)}{\zb}$. 
		
		As a result, 
		$(A,X, \push{\zb}{\epsilon}) \to (B,Y, \push{\zb_-}{\epsilon})$ is a rule of $\CFG$.
	\end{itemize}
	We have shown that every node satisfies the requirements of a "derivation tree" for $\CFG$. 
\end{proof}

\subsection{Proof of Proposition~\ref{prop:CFG-in-pump-skip}}\label{app:CFG-in-pump-skip}
\propCFGinPumpSkip*
\begin{proof}
	A ""pop step"" is simply a derivation step of $\CFG$ where the production rule applied is of the form $(A,X,\sigma) \to (B,Y,\sigma')$ with $\sigma=\push{(f,Y)}{\sigma'}$.
	We prove the following statement:
	
	For all derivation with $p$ "pop steps" from $(A,X,\sigma) \derivesindexedstar_{\CFG} w$ with $w \in T^*$,  
	there is a derivation \emph{with pump and skip} from $(A,X)[\zb]$ to $w'$ in $\Ggb$ with $w \subword w'$, and $\zb$ the "$p$-unfolding" of $\sigma$.
	
	We show this by induction on the derivation, and distinguish cases according to the rule used in its first step.
	\begin{itemize}
		\item If the rule is of the form $(A,X, \sigma) \to w$ then we apply the corresponding rule $(A,X) \to w$ from $(A,X)[\zb]$, hence $(A,X)[\zb] \derivesindexed[\Ggb] w$.
		
		\item If the rule is of the form $(A,X, \sigma) \to (B,X, \sigma)(C,X, \sigma)$ then there exist $w_B, w_C$ such that $w = w_B w_c$ and $ (B,X, \sigma) \derivesindexedstar_{\CFG} w_B$ and  $ (C,X, \sigma) \derivesindexedstar_{\CFG} w_C$.
		
		We apply the corresponding rule $(A,X) \to (B,X)(C,X)$ of $\Ggb$ from $(A,X)[\zb]$ to obtain $(B,X) [\zb] (C,X) [\zb]$. 
		By induction hypothesis, we obtain  $(B,X) [\zb] \derivesindexedstarps[\Ggb] w'_B$, as well as  $(C,X)[\zb] \derivesindexedstarps[\Ggb] w'_C$ with $w_B \subword w'_B$ and $w_C \subword w'_C$.
		As a consequence,  \[(A,X)[\zb] \derivesindexedstarps[\Ggb] w'_B w'_C\] which yields the result since $w = w_B w_C \subword w'_B w'_C$. 
		
		\item If the rule is of the form $(A,X, \sigma) \to (B,Y, \push{(f,X)}{\sigma})$, then, 
		by Lemma~\ref{lem:push-rule-aux}, the "$p$-unfolding" $\zb'$ of $\push{(f,X)}{\sigma}$ satisfies $(B,Y)[(f,X) \zb] \derivesindexedstarps (B,Y)[\zb']$.

		By induction hypothesis, there exists $w' \in T^*$ such that $(B,Y)[\zb'] \derivesindexedstarps[\Ggb] w'$ and $w \subword w'$. 
		As a consequence, \[(A,X)[\zb] \derivesindexed[\Ggb] (B,Y)[(f,X)\zb] \derivesindexedstarps[\Ggb] w'.\]
		
		\item If the rule is of the form $(A,X, \sigma) \to (B,Y, \sigma')$, with $\sigma = \push{(f,Y)}{\sigma'}$, then by Lemma~\ref{lem:pop-case} the "$(p-1)$-unfolding" $\zb'$ of $\sigma'$ is such that $(A,X)[\zb] \ruleskipstar (A,X) [(f,Y)\zb']$. 
		By induction hypothesis, there exists $w' \in T^*$ such that \[(A,X)[(f,Y)\zb'] \derivesindexedstarps[\Ggb] w'.\]
		As a consequence, \[(A,X)[\zb] \derivesindexed[\Ggb] (A,X) [(f,Y) \zb'] \derivesindexedstarps[\Ggb] w'.\]
	\end{itemize}
	We have proven the induction. To obtain the lemma, let $w \in \langIG{\CFG}$.
	There is a  derivation in $\CFG$ from $(S,\Useful, \epsilon)$ to $w$. Let $p$ be its number of "pop steps".
	Since $\epsilon$ is the "$p$-unfolding" of $\epsilon$, there is a derivation from $(S,\Useful)$ to some $w'$ with $w \subword w'$.
	As a result, $w \in \dcl{\langIG{\Ggb}}$.
\end{proof}

\subsection{Proof of Lemma~\ref{lem:monotone-unfolding}}\label{app:monotone-unfolding}

\lemMonotoneUnfolding*
\begin{proof}
	By induction on the depth $d$ of $\sigma$.
	If $d=0$ then $\unfold{p}{\sigma} = \unfold{p-1}{\sigma} = \epsilon$.
	If $d>0$, we distinguish cases according to the shape of $\sigma$. 
	If $\sigma$ is a "$d$-atom" $(f,X) \sigma'$ then we  simply apply the induction hypothesis.
	The same goes for a sequence of "$d$-atoms": we apply the previous case to each one of them.
	In the case of a "$d$-block", we have 
	\begin{multline*}\unfold{p}{\sigma} =  z^u  z^v ((f,X) \unfold{p-1}{\sigma_1}) \cdots \\ \cdots ((f,X) \unfold{p-1}{\sigma_r}) z^v \unfold{p}{w}\end{multline*} with $z^u, z^v, \sigma_1, \dots, \sigma_r$ as in the definition.
	\begin{itemize}
		\item if $p>1$, then by the previous cases, $z^u$, $z^v$, and $\unfold{p}{w}$ reduce to their "$(p-1)$-unfolding" counterparts. By induction hypothesis, each $\unfold{p-1}{\sigma_i}$ reduces to $\unfold{p-2}{\sigma_i}$.
		Hence $\unfold{p}{\sigma}$ reduces to $\unfold{p-1}{\sigma}$.

		\item If $p=1$, then we  have \[\unfold{p-1}{\sigma} = \unfold{0}{u_1 \cdots u_N}  \unfold{0}{v_1 \cdots v_N}  \unfold{0}{w}.\] 
		Further, by definition of a "$d$-block",  $\morphism(z^u_i) = \morphism(z^v_i) = e$ for all $i$. Moreover, for all $j$ we have $\morphism((f,X) \unfold{0}{\sigma_j}) = e$, since $\push{(f,X) \unfold{0}{\sigma_j}}{\epsilon}= \push{(f,X)}{\sigma_j} = u_1 \dots u_N e^+ v_1 \dots v_N$ and $\morphism(u_1 \dots u_N e^+ v_1 \dots v_N) = e$.
		As a consequence, we have 
		\begin{align*}
			\unfold{p}{\sigma} = &z^u z^v ((f,X) \unfold{p-1}{\sigma_1}) \cdots ((f,X) \unfold{p-1}{\sigma_r}) z^v \unfold{p}{w}\\
			&\ruleskip z^u z^v z^w \ruleskipstar \unfold{p-1}{\sigma}
		\end{align*}
	\end{itemize}
	Finally, for a "summary" we can simply apply the previous cases to each of the components.
\end{proof}
\subsection{Proof of Lemma~\ref{lem:push-rule-aux}}
\label{app:push-rule-aux}

\lemPushRuleAux*

\begin{proof}
	We prove this statement by induction on the depth of $\sigma_2$.
	If $\sigma_2$ has depth $0$ then it is $\epsilon$, contradicting $\sigma_2 = \push{(f,X)}{\sigma_1}$.
	
	If $\sigma_2$ has depth $d>0$ then we distinguish cases according to its shape.
	In cases (1) and (ii) we will simply show that $(f,Y) \unfold{p}{\sigma_1}$ is the "$p$-unfolding" of $\sigma_2$.
	In case (a) we will use the induction hypothesis, and in cases (A) and (B) we will actually apply a "pump rule" and "skip rules" to obtain $\unfold{p}{\sigma_2}$. 
	
	We decompose $\sigma_1$ as $\sigma_1 = \sigma' u B_1 \dots B_k$. 
	By definition  we have \[\unfold{p}{\sigma_1}= \unfold{p}{\sigma'} \unfold{p}{u} \unfold{p}{B_1} \cdots \unfold{p}{B_k}.\]
	We follow the cases in the definition of $\push{(f,X)}{\sigma_1}$.
	
	\begin{enumerate}
		\item If $\Jheight{(f,X) \sigma_1}>d$ then $\sigma_2 = (f,X)\sigma_1$.
		Then by definition $\unfold{p}{\sigma_2} = (f,X)\unfold{p}{\sigma_1}$.
		
		\item Otherwise, we have $\Jheight{(f,X) \sigma_1}= d$
		\begin{enumerate}
			\item if $\Jheight{(f,X) \sigma'}<d$ then \[\sigma_2 = (\push{(f,X)}{\sigma'})u B_1 \dots B_k.\]
			By induction hypothesis, we have \[(B,Y)[(f,X)\unfold{p}{\sigma'}] \derivesindexedstarps (B,Y) [z'']\] with $z''= \unfold{p}{\push{(f,X)}{\sigma'}}$.
			Therefore,
			\begin{align*}
				& (B,Y)[\unfold{p}{\sigma'} \unfold{p}{u} \unfold{p}{B_1} \cdots \unfold{p}{B_k}]\\
				& \derivesindexedstarps (B,Y) [z'' z^u z_1 \dots z_k]\\
				& = (B,Y) [\sigma_2].
			\end{align*}
			
			\item Otherwise, $\Jheight{(f,X) \sigma'}=d$ and $(f,X) \sigma'$ is a "$d$-atom".
			\begin{enumerate}
				\item If $((f,X) \sigma') u$ is of the form $u_1 \dots u_N v_0 v_1 \dots v_N w$ with $\morphism(u_i) = \morphism(v_i) = \morphism(v_0) = e$ for all $i\geq 1$, for some  $e \in \idempotents{\prodmonoid}$, 
				then 
				\[
					(f,X) \unfold{p}{\sigma'} = \unfold{p}{u_1} \dots \unfold{p}{u_N} \unfold{p}{v_0} \dots \unfold{p}{v_N} \unfold{p}{w}.	
				\]

				Furthermore since $\sigma_2$ is obtained by pushing $(f,X)$ we must have $\betaG(f) = B$ and $Y = f \cdot X$.
				Furthermore, since $(B,Y, \sigma_2)$ is a non-terminal of $\CFG$, $\sigma_2$ must be "feasible@@summary", hence $e \neq \zeroPM$. This means that $e = (B,Y, M, C,Y)$ for some $C$ and $M$.
				We have two cases.

				\begin{enumerate}
					\item If there exists $j$ such that $B_j$ is of the form \[u'_1 \dots u'_N e^+ v'_1 \dots v'_N w'\] and $\morphism(v_1 \dots v_N w B_1 \dots B_{j-1} u'_1 \dots u'_N) = e$. Then this implies \[\sigma_2=(u_1 \dots u_N e^+ v'_1 \dots v'_N w') B_{j+1} \dots B_k. \]
					In that case, we also have
					\begin{align*}
						\unfold{p}{B_j} =& 
						\unfold{p}{u_1'} \dots \unfold{p}{u_N'} z'_e\\ 
						&\unfold{p}{v_1'} \dots \unfold{p}{v_N'} \unfold{p}{w'}
					\end{align*}  
					with $\morphism(z'_e) = e$.
					
					Let $\tilde{z}$ be the "$p$-unfolding" of the "summary" $u_1 \dots u_N e^+ v'_1 \dots v'_N w'$.
					Then this must be of the form $z''  \unfold{p}{v'_1} \dots \unfold{p}{v'_N} \unfold{p}{w'}$ with $\morphism(z'') = e$.
					Since $e = (B,Y, M, C,Y)$ is idempotent, we can apply "pump rules" and "skip rules" as follows:
					
					\begin{align*}
						&(B,Y)&&[(f,X) \unfold{p}{\sigma_1}] \\ 
						=&(B,Y)&&[ \unfold{p}{u_1} \dots \unfold{p}{u_N}\\
						&&&\unfold{p}{v_0} \dots \unfold{p}{v_N}  \unfold{p}{w} \\
						&&& \unfold{p}{B_1} \dots \unfold{p}{B_{j-1}} \\
						&&& \unfold{p}{u_1'} \dots \unfold{p}{u_N'} z'_e\\
						&&& \unfold{p}{v_1'} \dots \unfold{p}{v_N'} \unfold{p}{w'}\\
						&&& \unfold{p}{B_{j+1}} \dots \unfold{p}{B_k}]\\
						\rulepump &(B,Y)&&[\unfold{p}{v_1'} \dots \unfold{p}{v_N'}\\
						&&&\unfold{p}{u_1} \dots \unfold{p}{u_N}\\
						&&&\unfold{p}{v_1} \dots \unfold{p}{v_N}  \unfold{p}{w} \\
						&&& \unfold{p}{B_1} \dots \unfold{p}{B_{j-1}} \\
						&&& \unfold{p}{u_1'} \dots \unfold{p}{u_N'} z'_e\\
						&&& \unfold{p}{v_1'} \dots \unfold{p}{v_N'} \unfold{p}{w'}\\
						&&& \unfold{p}{B_{j+1}} \dots \unfold{p}{B_k}]\\
						\ruleskip &(B,Y)&&[\unfold{p}{v_1'} \dots \unfold{p}{v_N'} \unfold{p}{w'}\\
						&&& \unfold{p}{B_{j+1}} \dots \unfold{p}{B_k}]\\
						\rulepump &(B,Y)&&[z''\unfold{p}{v_1'} \dots \unfold{p}{v_N'} \unfold{p}{w'}\\
						&&& \unfold{p}{B_{j+1}} \dots \unfold{p}{B_k}]\\
						= &(B,Y)&&[\unfold{p}{\sigma_2}]
					\end{align*}
					
					\item Otherwise we have the "summary" 
					\[\sigma_2 = (u_1 \dots u_N e^+ v_1 \dots v_N w) B_{1} \dots B_k. \] 
					
					We define $B = u_1 \dots u_N e^+ v_1 \dots v_N w$. Then the "$p$-unfolding" $\unfold{p}{B}$ is of the form \[ z''  \unfold{p}{v_1} \dots \unfold{p}{v_N}\unfold{p}{w} \] with $\morphism(z'') = e$.
					Since $e = (B,Y, M, C,Y)$ is idempotent, we can apply "pump rules" and "skip rules" as follows:
					
					\begin{align*}
						&(B,Y)&&[(f,X) \unfold{p}{\sigma_1}] \\ 
						=&(B,Y)&&[\unfold{p}{u_1} \dots \unfold{p}{u_N}\\
						&&&\unfold{p}{v_0} \unfold{p}{v_1} \dots \unfold{p}{v_N} \unfold{p}{w}\\ 
						&&&\unfold{p}{B_1} \dots \unfold{p}{B_k}]\\
						\rulepump &(B,Y)&&[\unfold{p}{v_1} \dots \unfold{p}{v_N} \\
						&&&\unfold{p}{u_1} \dots \unfold{p}{u_N}\\
						&&&\unfold{p}{v_1} \dots \unfold{p}{v_N}  \unfold{p}{w} \\
						&&& \unfold{p}{B_1} \dots \unfold{p}{B_{k}}]\\
						\ruleskip &(B,Y)&&[\unfold{p}{v_1} \dots \unfold{p}{v_N} \\
						&&& \unfold{p}{B_{1}} \dots \unfold{p}{B_k}]\\
						\rulepump &(B,Y)&&[z''\unfold{p}{v_1} \dots \unfold{p}{v_N} \\
						&&& \unfold{p}{B_{j+1}} \dots \unfold{p}{B_k}]\\
						= &(B,Y)&&[\unfold{p}{\sigma_2}]
					\end{align*}	
					The resulting stack content is the "$p$-unfolding" of $\sigma_2$.
				\end{enumerate} 						
				\item Otherwise, we have $\sigma_2 = ((f,X) \sigma')u B_1 \dots B_k$ and thus $(f,X) \unfold{p}{\sigma_1} = \unfold{p}{\sigma_2}$. 
			\end{enumerate}
		\end{enumerate}
	\end{enumerate}
\end{proof}

\subsection{Proof of Lemma~\ref{lem:pop-case}}
\label{app:pop-rule-aux}

\lemPopRuleAux*

\begin{proof}[Proof of Lemma~\ref{lem:pop-case}]
	We prove this statement by induction on the depth of $\sigma_1$.
	If $\Jheight{\sigma_1}= 0$ then it is $\epsilon$ and $\pop{(f,Y)}{\sigma_1}$ is empty.

	If $\sigma_1$ has depth $d>0$ then we distinguish cases according to its shape. 
	Recall that by Lemma~\ref{lem:monotone-unfolding} a "$p$-unfolding" of a "summary" $\sigma$ always reduces to its "$(p-1)$-unfolding" through "skip rules". We will use this fact often throughout the proof.
	
	Let $\sigma_2 = \sigma' u B_1 \dots B_k$. 
	
	\begin{enumerate}
		\item If $\Jheight{(f,Y) \sigma_2}>d$ then $\sigma_1 = (f,Y)\sigma_2$.
		Then by definition $\unfold{p}{\sigma_1} = (f,Y)\unfold{p}{\sigma_2}$, hence  \[ (A,X)[\unfold{p}{\sigma_1}] \ruleskipstar (A,X)[(f,X)\unfold{p-1}{\sigma_2}]\] by Lemma~\ref{lem:monotone-unfolding}.
		
		\item Otherwise, we have $\Jheight{(f,Y) \sigma_2}= d$
		\begin{enumerate}
			\item if $\Jheight{(f,Y) \sigma'}<d$ then \[\sigma_1 = (\push{(f,Y)}{\sigma'})u B_1 \dots B_k.\] 
			Then we have 
			\begin{align*}
				&\unfold{p}{\sigma_1} =\\
				& \unfold{p}{\push{(f,Y)}{\sigma'}} \unfold{p}{u} \unfold{p}{B_1} \dots \unfold{p}{B_k}.
			\end{align*}
			Since $(A,X, \sigma_1) \to (B,Y, \sigma_2)$ is a rule of $\CFG$,  $\sigma_1$ and $\sigma_2$ must be "feasible@@summary", thus so are $\sigma'$, and $\push{(f,Y)}{\sigma'}$.
			Thus, by construction of $\CFG$, we must have the rule $(A,X,\push{(f,Y)}{\sigma'}) \to (B, Y, \sigma')$ in $\CFG$.
			Hence, by induction hypothesis, \[(A,X)[\unfold{p}{\push{(f,Y)}{\sigma'}}] \ruleskipstar (A,X)[(f,Y) \unfold{p-1}{\sigma'}].\] 
			As a result, we have 
			\begin{align*}
				&(A,X)[\unfold{p}{\sigma_1}]\\
				& \ruleskipstar (A,X)[(f,Y)\unfold{p-1}{\sigma'} \unfold{p}{u} \unfold{p}{B_1} \dots \unfold{p}{B_k}]
			\end{align*}
			By Lemma~\ref{lem:monotone-unfolding}, we thus have 
			
			\begin{align*}
				& (A,X)[\unfold{p}{\sigma_1}]  \\
				\ruleskipstar & (A,X)[(f,Y)\unfold{p-1}{\sigma'} \unfold{p-1}{u} \unfold{p-1}{B_1} \dots \unfold{p-1}{B_k}]\\
				= & (A,X)[(f,Y)\unfold{p-1}{\sigma_2}]
			\end{align*}

			\item Otherwise, $\Jheight{(f,Y) \sigma'}=d$ and $(f,Y) \sigma'$ is a "$d$-atom".
			\begin{enumerate}
				\item If $((f,Y) \sigma') u$ is of the form $u_1 \dots u_N v_0 v_1 \dots v_N w$ with $\morphism(u_i) = \morphism(v_i) = \morphism(v_0)= e$ for all $i \geq 1$, for some $e \in \idempotents{\prodmonoid}$,
				we have two cases.
				\begin{enumerate}
					\item If there exists $j$ such that $B_j$ is of the form $u'_1 \dots u'_N e^+ v'_1 \dots v'_N w'$ and \[\morphism(v_1 \dots v_N w B_1 \dots B_{j-1} u'_1 \dots u'_N) = e\] then we pick the maximal such $j$. 
					We have $\sigma_1 = B B_{j+1} \dots B_k$, where \[ B = u_1 \dots u_N e^+ v'_1 \dots v'_N w'. \]
					In that case, \[\unfold{p}{\sigma_1} = \unfold{p}{B} \unfold{p}{B_{j+1}} \dots \unfold{p}{B_k}.\]
					
					Let $B_j' = u'_1 \dots u'_N e^+ v'_1 \dots v'_N$, that is, $B_j$ without the $w'$ suffix.					
					Let us also define $z'= \unfold{p-1}{\sigma' u B_1 \dots B_{j-1} B'_j}$. Then $(f,Y) z'$ is of the form 
					\[\unfold{p-1}{u_1} \cdots \unfold{p-1}{u_N}  z'_e \unfold{p-1}{v'_1} \cdots \unfold{p-1}{v'_N}\]
					for some $z_e'$ such that $\morphism(z_e') = e$.

					By definition of the "$p$-unfolding", since \[\push{(f,X)}{\sigma' u B_1 \dots B_{j-1} B'_j} = B\] and $p \geq 1$, $\unfold{p}{B}$ is of the form \[\unfold{p}{u_1} \dots \unfold{p}{u_N} z_-  (f,Y) z'  z_+ \unfold{p}{v'_1} \cdots \unfold{p}{v'_N} \unfold{p}{w'}\] 
					with $\morphism(z_-) = \morphism(z_+) =e$.

					We can use "skip rules": 
					\begin{align*}
						& (A,X)&&[\unfold{p}{\sigma_1}]\\ 
						= &(A,X)&&[\unfold{p}{B} \unfold{p}{B_{j+1}} \dots \unfold{p}{B_k}]\\
						= &(A,X)&&[\unfold{p}{u_1} \dots \unfold{p}{u_N} z_-  (f,Y) z' z_+\\ &&&\unfold{p}{v'_1} \cdots \unfold{p}{v'_N} \\ 
						&&& \unfold{p}{w'} \unfold{p}{B_{j+1}} \dots \unfold{p}{B_k}] 	\\
						\ruleskipstar &(A,X)&&[\unfold{p-1}{u_1} \dots \unfold{p-1}{u_N} z_-  \\ 
						&&&(f,Y) z' z_+ \unfold{p-1}{v'_1} \cdots \unfold{p-1}{v'_N}\\ 
						&&& \unfold{p-1}{w'} \unfold{p-1}{B_{j+1}} \dots \unfold{p-1}{B_k}]\\
						=& (A,X)&&[\unfold{p-1}{u_1} \dots \unfold{p-1}{u_N} z_- \\ 
						&&& \unfold{p-1}{u_1} \cdots \unfold{p-1}{u_N}  z'_e\\ 
						&&&\unfold{p-1}{v'_1} \cdots \unfold{p-1}{v'_N}\\ 
						&&&  z_+ \unfold{p-1}{v'_1} \cdots \unfold{p-1}{v'_N} \unfold{p-1}{w'}\\ 
						&&& \unfold{p-1}{B_{j+1}} \dots \unfold{p-1}{B_k}]\\
						\ruleskip^2 & (A,X)&&[\unfold{p-1}{u_1} \dots \unfold{p-1}{u_N}  z'_e\\ &&&\unfold{p-1}{v'_1} \cdots \unfold{p-1}{v'_N}\\ 
						&&& \unfold{p-1}{w'} \unfold{p-1}{B_{j+1}} \dots \unfold{p-1}{B_k}]\\
						= &(A,X)&&[(f,Y)z' \unfold{p-1}{B_{j+1}} \dots \unfold{p-1}{B_k}]\\
						= &(A,X)&&[(f,Y)\unfold{p-1}{\sigma_2}]
					\end{align*}
					\item Otherwise our "summary" is of the form $\sigma_1 = (u_1 \dots u_N e^+ v_1 \dots v_N w) B_{1} \dots B_k$. 
					In particular, its "$p$-unfolding" is the stack content $\unfold{p}{\sigma_1} = \unfold{p}{B} \unfold{p}{B_1} \cdots \unfold{p}{B_k}$.

					Since $(f,Y)\sigma' u = u_1 \dots u_N v_0 \dots v_N w$, the "$(p-1)$-unfolding" of $u_1 \dots u_N v_0 \dots v_N$ must be of the form $(f,Y) z'$ for some $z'$. 
					Hence \[(f,Y) z'= \unfold{p-1}{u_1} \cdots \unfold{p-1}{u_N} \unfold{p-1}{v_0} \cdots \unfold{p-1}{v_N}.\]
					
					By definition of the "$p$-unfolding", $ \unfold{p}{B}$ is of the form \[\unfold{p}{u_1} \cdots \unfold{p}{u_N} z_-  (f,Y)z'  z_+ \unfold{p}{v_1} \cdots \unfold{p}{v_N} \unfold{p}{w}\] with $\morphism(z_-) = \morphism(z_+) =e$.

					We can use "skip rules": 
					\begin{align*}
						&(A,X)&&[\unfold{p}{\sigma_1}]\\
						= &(A,X)&&[\unfold{p}{B}  \unfold{p}{B_1} \cdots \unfold{p}{B_k}]\\
						= &(A,X)&&[\unfold{p}{u_1} \cdots \unfold{p}{u_N} z_-  \\ 
						&&& (f,Y)z'z_+ \unfold{p}{v_1} \cdots \unfold{p}{v_N} \\
						&&& \unfold{p}{w} \unfold{p}{B_1} \cdots \unfold{p}{B_k}]\\
						= &(A,X)&&[\unfold{p}{u_1} \cdots \unfold{p}{u_N} z_-\\ 
						&&&  \unfold{p-1}{u_1} \cdots \unfold{p-1}{u_N}\\
						&&& \unfold{p-1}{v_0} \unfold{p-1}{v_1} \cdots \unfold{p-1}{v_N}  \\ 
						&&&z_+ \unfold{p}{v_1} \cdots \unfold{p}{v_N} \unfold{p}{w}\\
						&&& \unfold{p}{B_1} \cdots \unfold{p}{B_k}]\\
						\ruleskipstar & (A,X)&&[\unfold{p-1}{u_1} \cdots \unfold{p-1}{u_N} z_-\\ &&&\unfold{p-1}{u_1} \cdots \unfold{p-1}{u_N}\\
						&&& \unfold{p-1}{v_0} \unfold{p-1}{v_1} \cdots \unfold{p-1}{v_N} \\\ 
						&&&z_+ \unfold{p-1}{v_1} \cdots \unfold{p-1}{v_N} \unfold{p-1}{w}\\ &&&\unfold{p-1}{B_1} \cdots \unfold{p-1}{B_k}]\\
						\ruleskip^2 & (A,X)&&[(f,Y)z' \unfold{p-1}{w} \unfold{p-1}{B_1} \cdots \unfold{p-1}{B_k}]\\
						=& (A,X)&&[\unfold{p-1}{u_1 \dots u_N v_0 \dots v_N}\\ 
						&&&\unfold{p-1}{w} \unfold{p-1}{B_1} \cdots \unfold{p-1}{B_k}]\\
						=& (A,X)&&[\unfold{p-1}{\sigma_2}]
					\end{align*}					
				\end{enumerate} 						
				\item Otherwise, we have the "summary" \[ \sigma_1 = ((f,Y) \sigma')u B_1 \dots B_k \] and thus $\unfold{p}{\sigma_1} =(f,Y) \unfold{p}{\sigma_2}$.
				According to Lemma~\ref{lem:monotone-unfolding}, we obtain the derivation \[ (A,X)[\sigma_1] \ruleskipstar (A,X)[(f,Y)\unfold{p-1}{\sigma_2}]. \] 
			\end{enumerate}
		\end{enumerate}
	\end{enumerate}
\end{proof}

\section{Additional material from Section~\ref{sec:lower-bound}}\label{app:lower-bound}
\subsection{Proofs for NFA lower bound}\label{app:main-lower-bound}

\begin{lemma}
	There is an derivation of $\Gg_n$ that produces the word $\ltr{a}^{\exp_3(n)}$
\end{lemma}

\begin{proof}
	Let $\alpha_1\cdots \alpha_m \in \Sigma_n^*$ be the unique word accepted by all $\Aa_i$, with $m = 2^n$.
	For all $b_1, \ldots, b_m \in \set{\ltr{0},\ltr{1}}$, we write $\overline{b_1\cdots b_m}$ for the number in $[0,2^m-1]$ whose binary representation over $m$ bits is $b_1\cdots b_m$, where $b_1$ is the least significant digit.
	
	We show that for all $b_1, \ldots, b_m \in \set{\ltr{0},\ltr{1}}$, if $M= \overline{b_1\cdots b_m}$ then
	$Z[ (\alpha_1, b_1) \cdots (\alpha_m, b_m) \bot]$ produces $\ltr{a}^{2^{2^{m} - M}}$, by induction on $2^m-M$.
	
	We can apply $Z \xrightarrow{\Bcal_1,\ldots,\Bcal_n} D$ and $D \to AA$ to obtain two copies of $A [(\alpha_1, b_1) \cdots (\alpha_m, b_m) \bot]$. This is because
	$\alpha_1\cdots \alpha_m$ is accepted by all $\Aa_i$. We are left with 
	\[ A[(\alpha_1, b_1) \cdots (\alpha_m, b_m) \bot] A [(\alpha_1, b_1) \cdots (\alpha_m, b_m) \bot]. \]
	
	If $b_i = \ltr{1}$ for all $i$, then we can apply $A (\alpha_i, \ltr{1}) \to A$ for each $i$ and then $A \bot \to F$ and $F\to \ltr{a}$ to obtain $\ltr{a}\ltr{a}$, which is what we want since we would then have $M = 2^m -1$ and thus $\ltr{a}^{2^{2^{m} - M}} = \ltr{a}^2$.
	
	Otherwise, let $j$ be the least index such that $b_j=\ltr{0}$.
	It suffices to show that the term $A [(\alpha_1, b_1) \cdots (\alpha_m, b_m) \bot]$ produces $\ltr{a}^{2^{2^{m} - M -1}}$. 
	We have $M+1 = \overline{\ltr{1}^{j-1} \ltr{0} b_{j+1} \cdots b_m} +1 = \overline{\ltr{0}^{j-1} \ltr{1} b_{j+1} \cdots b_m}$.

	We can apply:
	\begin{itemize}
		\item $A (\alpha_i, \ltr{1}) \to A$ for each $i < j$ until we get\\$A [(\alpha_j, \ltr{0})(\alpha_{j+1}, b_{j+1}) \cdots (\alpha_m, b_m) \bot]$,
		
		\item then apply $A (\alpha_j, \ltr{0}) \to B$ and $B \to Z(\alpha_j,\ltr{1})$ to get \\ $Z[ (\alpha_j, \ltr{1})(\alpha_{j+1}, b_{j+1}) \cdots (\alpha_m, b_m) \bot]$,
		
		\item and $Z \to Z (\alpha_i, \ltr{0})$ for each $i<j$, in decreasing order, until we obtain\\ $Z [(\alpha_1, \ltr{0})\cdots (\alpha_{j-1},\ltr{0}) (\alpha_j,\ltr{1}) (\alpha_{j+1}, b_{j+1}) \cdots (\alpha_m, b_m) \bot]$, which produces $\ltr{a}^{2^{2^{m} - M -1}}$ by induction hypothesis.
	\end{itemize} 
	
	The induction is proved. To obtain the lemma, it suffices to start with
	$S$, apply $S \to Z \bot$ and then $Z \to Z (\alpha_i, \ltr{0})$ for
	each $i \in [1,n]$ in decreasing order. We get  $Z (\alpha_1,
	\ltr{0})\cdots (\alpha_{n},\ltr{0})$, which produces $a^{\exp_3(n)}$
	by applying the induction with $M=0$.
\end{proof}

\begin{lemma}
	For all $E \in N_n$ and $z \in I_n^*$ the language $\langX{Ez}{\emptyset}$ contains at most one word.
	In particular, $\langIG{\Gg_n}$ is empty or a singleton.
\end{lemma}
\begin{proof}
	We prove this for each $E \in N_n$, one by one.
	
	\begin{enumerate}

		\item We first observe that from a configuration $Zz$ there is always at most one rule which can lead to a "complete@@tree" "derivation tree": if $z$ does not contain $\bot$ then $\langX{Zz}{\emptyset} = \emptyset$. If $z = (\alpha_1, b_1) \cdots (\alpha_m,b_m) \bot z'$ then: 
		\begin{itemize}
			\item if $m \neq 2^n$ we cannot apply $Z \xrightarrow{\Bcal_1,\ldots,\Bcal_n} D$ since $z$ has no prefix accepted by all $\Bcal_i$. Furthermore, if $m> 2^n$, we have $\langX{Zz}{\emptyset} = \emptyset$ since we can only push more pairs $(\alpha,\ltr{0})$ on the stack, so we will never be able to apply $Z \xrightarrow{\Bcal_1,\ldots,\Bcal_n} D$.
			
			\item if $m = 2^n$ then we cannot apply $Z \to Z (\alpha,\ltr{0})$, by the previous item, as we would obtain more than $2^n$ symbols before the first $\bot$.
		\end{itemize} 
		
		\item Note that $B$, $D$ and $S$ all have a single rule. Meanwhile, $A$ has several but the top stack symbol determines which rule can be applied. 
		In conclusion, from every configuration $E z$, there is at most one rule that can be applied to lead to a complete derivation. As a consequence, the language of $\Gg$ contains at most one word.
	\end{enumerate}
\end{proof}

By combining the two previous statements we conclude that $\Gg_n$ recognizes the singleton language $\set{\ltr{a}^{\exp_3(n)}}$, while having "size@@IG" only quadratic in $n$. 
A trim NFA for this language must be acyclic, as otherwise it would recognize an infinite language, and thus have at least $\exp_3(n)$ states.

\subsection{Computational hardness}\label{app:computational-hardness}

We now use methods from \cite{Zetzsche16} to derive \cref{main-lower-bound-dfa}
and $\coNEXP[3]$-hardness in \cref{main-completeness} from our construction above.
For this, we rely on the notion of $\Delta(f)$ language classes~\cite{Zetzsche16}, which requires some terminology.
A ""transducer"" is a tuple $\Tt=(Q,\Sigma,\Gamma,E,q_0,F)$, where $Q$ is a finite set of \emph{states}, $\Sigma$ is its \emph{input alphabet}, $\Gamma$ is its \emph{output alphabet}, $E\subseteq Q\times\Sigma^*\times\Gamma^*\times Q$ is its finite set of \emph{edges}, $q_0\in Q$ is its \emph{initial state}, and $F\subseteq Q$ is its set of \emph{final states}. It describes a relation $\intro*\relTrans{\Tt}\subseteq\Sigma^*\times\Gamma^*$, namely the set of all pairs $(u,v)$ for which there are decompositions $u=u_1\cdots u_n$ and $v=v_1\cdots v_n$, states $q_0,q_1,\ldots,q_n$, and edges $(q_{i-1},u_i,v_i,q_i)\in E$ for $i=1,\ldots,n$ with $q_n\in F$. For a language $L\subseteq\Sigma^*$, we write $\Tt(L)=\{v\in\Gamma^* \mid \exists u\in L\colon (u,v)\in\relTrans{\Tt}\}$.

A \emph{language class} is a class of formal languages, together with some
means to represent them, such as grammars or automata. A language class $\Cc$
is an \emph{effective full trio} if for a given language $L$ from $\Cc$, we can
effectively compute a description of $\Tt(L)$. Now suppose $f\colon \NN\to\NN$ is
an amplifying function meaning there is a polynomial $p$ such that $f(p(n))\ge
f(n)^2$. Then, $\Cc$ is said to be $\Delta(f)$ if (i)~computing $\Tt(L)$ can be
done in polynomial time and (ii)~given $n$, one can compute a description of
the language $\{a^{f(n)}\}$ in polynomial time. 

In these terms, \cref{construction-lower-bound} tells us that the class of
indexed languages (represented by indexed grammars) are $\Delta(\exp_3)$:
Applying rational transductions to indexed languages is well-known to be
possible in polynomial time~(see, e.g.~\cite[Section 3.1]{Zetzsche15arxiv}). 
\begin{proposition}\label{indexed-delta-exp-3}
	The indexed languages (represented by indexed grammars) are $\Delta(\exp_3)$.
\end{proposition}

We will also need the notion of simple substitutions. For alphabets
$\Sigma,\Gamma$, a \emph{substitution} is a map $\sigma\colon
\Sigma\to\powerset{\Gamma^*}$ that replaces each letter in $\Sigma$ by a
language over $\Gamma$.  For language $L\subseteq\Sigma^*$, the language
$\sigma(L)$ is defined in the obvious way. The substitution $\sigma$ is said to
be \emph{simple for $L\subseteq\Sigma^*$} if $\Sigma\subseteq\Gamma$ and there
is a letter $a\in\Sigma$ such that $\sigma(a')=\{a'\}$ for each
$a'\in\Sigma\setminus\{a\}$. We say that a language class $\Cc$ is \emph{closed
	under simple subsititutions} if for any given $L$ from $\Cc$, and any simple
substitution $\sigma$ for $L$, the language $\sigma(L)$ belongs to $\Cc$, and a
representation can be computed in polynomial time. It is easy to see that the
indexed languages are closed under simple substitutions. In \cite[Theorem
15]{Zetzsche16}, it is shown that downward closure inclusion and equivalence
are both $\coNTIME(t)$-hard for any language class that is $\Delta(t)$ and
closed under simple substitutions. Hence, \cref{construction-lower-bound}
implies $\coNEXP[3]$-hardness of downward closure inclusion and equivalence.

\subsection{Proofs for DFA lower bound}\label{app:lower-bound-dfa}
Here, we prove a slightly more general result than discussed in
\cref{sec:lower-bound}: We show that for any $\Delta(f)$ language class, DFAs
for downward closures of languages with polynomial-sized descriptions require
at least size $2^{f(n)}$, provided that the language class is also closed under
simple subsitution:
\begin{proposition}\label{construction-lower-bound-dfa}
	Let $f\colon\NN\to\NN$ be a function such that a language class $\Cc$
	is $\Delta(f)$ and let $\Cc$ be closed under simple substitutions. Then
	there is a family languages $(L_n)_{n\ge 1}$ with polynomial
	description sizes such that any DFA for $\dcl{L_n}$ requires at least
	$2^{f(n)}$ states.
\end{proposition}
\begin{proof}
	We claim that we can construct a representation of 
	\[ L_n=\{uv \mid u,v\in\{\ltr{0},\ltr{1}\}^*,~|u|=|v|=f(n),~u\ne v\} \]
	in polynomial time.  Note that a DFA for $\dcl{L_n}$ requires at least
	$2^{f(n)}$ states: After reading distinct prefixes $u,u'\in\{\ltr{0},\ltr{1}\}^*$ of length $f(n)$, the DFA must enter distinct states, as otherwise, it would accept $uu$, which does not belong to $\dcl{L_n}$.
	
	To construct $L_n$ for given $n\in\NN$, we begin by building a
	representation of $\{\ltr{a}^{f(n)}\}$. Using a "transducer", we then
	insert a single occurrence of a letter $\ltr{b}$ into every word, and
	then substitute this $\ltr{b}$ with $\{\ltr{b}\ltr{a}^{f(n)}\ltr{c}\}$. This
	yields the language 
	\[ \{\ltr{a}^r \ltr{b}\ltr{a}^{f(n)} \ltr{c} \ltr{a}^s \mid r+s=f(n)\}. \] 
	Using another "transducer", we can remove (i)~one occurrence
	of $\ltr{a}$ within $\ltr{a}^r$ and $\ltr{a}^s$ and (ii)~one occurrence of $\ltr{a}$ within $\ltr{a}^{f(n)}$ to obtain
	\[ \{\ltr{a}^r \ltr{b}\ltr{a}^{f(n)-1}\ltr{c} \ltr{a}^s \mid r+s=f(n)-1\}. \]
	A final "transducer" then replaces (i)~each $\ltr{a}$  with $\ltr{0}$ or
	$\ltr{1}$ and (ii)~$\ltr{b}$ and $\ltr{c}$ with distinct letters in
	$\{\ltr{0},\ltr{1}\}$. This results in the language $L_n$.  \end{proof} 

Now, \cref{construction-lower-bound-dfa} and \cref{indexed-delta-exp-3}
together directly imply \cref{main-lower-bound-dfa}.

}

\label{afterbibliography}
\newoutputstream{pagestotal}
\openoutputfile{main.pagestotal.ctr}{pagestotal}
\addtostream{pagestotal}{\getpagerefnumber{afterbibliography}}
\closeoutputstream{pagestotal}

\newoutputstream{todos}
\openoutputfile{main.todos.ctr}{todos}
\addtostream{todos}{\arabic{@todonotes@numberoftodonotes}}
\closeoutputstream{todos}
\end{document}